\pdfoutput=1
\documentclass[12pt]{article}
\usepackage[right=1in,left=1in,top=1in,bottom=1in]{geometry}
\usepackage{setspace}
\usepackage{hyperref}
\hypersetup{colorlinks, citecolor=blue, filecolor=blue, linkcolor=blue, urlcolor=blue,
  pdftitle={Identification and Inference with Machine-Learned Instruments},
  pdfauthor={Fangzhou Yu}}
\usepackage{graphicx}
\usepackage{url}
\usepackage[round]{natbib}
\usepackage{amsmath,amsthm}
\usepackage{engord}
\usepackage{float}
\usepackage{placeins}
\usepackage{subfig}
\usepackage{pdflscape}
\usepackage{booktabs}
\usepackage{amssymb}
\usepackage{multirow}
\usepackage{xr}
\usepackage{xcolor}
\usepackage{listings}
\lstdefinestyle{crc}{basicstyle=\ttfamily\scriptsize,keywordstyle=\ttfamily,commentstyle=\ttfamily,breaklines=true,columns=fullflexible,keepspaces=true,showstringspaces=false,aboveskip=1pt,belowskip=1pt,xleftmargin=0pt}
\newsavebox{\crcRbox}
\newsavebox{\crcPYbox}

\newtheorem{theorem}{Theorem}
\newtheorem{example}{Example}
\newtheorem{definition}{Definition}
\newtheorem{assumption}{Assumption}
\newtheorem{proposition}{Proposition}
\newtheorem{lemma}{Lemma}
\newtheorem{corollary}{Corollary}
\newtheorem{remark}{Remark}

\DeclareMathOperator{\Var}{Var}
\DeclareMathOperator{\Cov}{Cov}
\DeclareMathOperator{\Supp}{Supp}
\newcommand{\E}{\mathbb{E}}
\newcommand{\indep}{\perp\!\!\!\perp}

\usepackage{sectsty}
\sectionfont{\large}
\subsectionfont{\normalsize}
\subsubsectionfont{\normalsize}

\title{\vspace*{-2.5cm} Identification and Inference with Machine-Learned Instruments}

\author{Fangzhou Yu\thanks{School of Economics, University of Sydney. \href{mailto:fangzhou.yu@sydney.edu.au}{fangzhou.yu@sydney.edu.au}}}

\date{\vspace*{0.5cm}September, 2026}

\begin{document}

\pagenumbering{roman}

\bgroup
\let\footnoterule\relax

\begin{singlespace}
  % \begin{doublespace}
  \maketitle

  \begin{abstract}
    \noindent Instrumental variable estimation increasingly pools many or high-dimensional instruments into a single machine-learned first stage, with rich controls partialled out. The resulting estimand, the partialled-out IV coefficient, is a signal-weighted average of the heterogeneous effects, which gives an opaque first stage a precise structural meaning. The average is convex whenever a covariance-monotonicity condition holds, and we provide a microfoundation for that condition based on vector monotonicity. With a learned signal, the usual debiased moment is not Neyman orthogonal, and its first-order bias is a drift toward the learner's own signal-weighted average, so naive inference remains valid only for that learner-dependent target. We construct a heterogeneity robust orthogonal score that restores $\sqrt{N}$ inference on the fixed, learner-invariant target at no efficiency cost, and provide a Hausman-type diagnostic and identification robust confidence sets.
  \end{abstract}

  \vspace{0.4em}

  \vspace{0.2em}
  % \noindent\textbf{JEL codes:} C14, C21, C26, C36, C55.

  % \end{doublespace}
\end{singlespace}
\thispagestyle{empty}

\clearpage
\egroup
\pagenumbering{arabic}
\setcounter{page}{1}

\section{Introduction}\label{sec:intro}

Instrumental variables estimation increasingly uses machine learning to combine many or high-dimensional instruments into a single predicted treatment signal, with rich controls partialled out nonparametrically. This approach yields the partially linear IV estimator implemented in the ``flexible IV'' routines of \texttt{DoubleML} and \texttt{ddml} packages. Applications include judge- and examiner-leniency designs, shift-share designs, and quarter-of-birth studies of returns to schooling \citep{angrist2022machine}. \citet{chen2021mostly} study this approach as an optimal-instrument problem. We examine the identification and inference of this estimator when treatment effects are heterogeneous and are correlated with treatment selection.

Two strands of the IV literature address different parts of this problem. The literature on causal weights shows that multiple instruments can give rise to negative weights on local effects \citep{mogstad2021causal,mogstad2024policy}, as can misspecified covariate adjustment \citep{blandhol2022tsls,sloczynski2024interpret}. These results clarify the causal interpretation of population estimands. The double/debiased machine learning literature develops methods for $\sqrt{N}$ inference with estimated nuisance functions, relying on orthogonal scores to control estimation error \citep{chernozhukov2018double,emmenegger2021regularizing,scheidegger2025machine}. For a learned instrument signal, however, conventional partialling out need not preserve this orthogonality under heterogeneity correlated with treatment selection, and first-stage estimation error can enter the moment. Our analysis connects the causal interpretation of the signal to the inference problem created by estimating it.

Our first contribution gives the partialled-out IV estimand a structural interpretation. Write the outcome as $Y = Y(0) + X\Delta$, where the unobserved effect $\Delta$ varies across individuals and may influence their choice of treatment $X$. For any square-integrable signal $g(Z,W)$ of instruments $Z$ and covariates $W$, including one constructed by machine learning, consider the partialled-out IV estimand
\begin{equation}\label{eq:betaiv-intro}
  \beta_{IV} = \frac{\E[\tilde{Y}\tilde{g}]}{\E[\tilde{X}\tilde{g}]},
  \qquad \tilde{A} := A - \E[A\mid W],
\end{equation}
where covariates are removed by subtracting conditional means \citep{robinson1988root}. This flexible covariate adjustment addresses the identification concerns raised by \citet{blandhol2022tsls} and \citet{sloczynski2024interpret}. Under conditional exogeneity and relevance, the estimand admits a signal-weighted average treatment effect (SWATE) representation,
\begin{equation}\label{eq:swate-intro}
  \beta_{IV} = \E_\Delta\big[\Delta\,\omega(\Delta)\big],
  \qquad \omega(\delta) = \frac{\E[\tilde{X}\tilde{g}\mid\Delta=\delta]}{\E[\tilde{X}\tilde{g}]},
\end{equation}
with weights that integrate to one. The weight on an effect is the covariance of residualized treatment and signal among individuals with that effect, normalized by the overall covariance. The representation therefore explains which effects an IV coefficient averages, even when its signal comes from an opaque learner. It requires no single-index restriction on choice behavior, because the weights are defined over effects, while the marginal treatment effect representations of \citet{heckman2005structural} and \citet{heckman2006understanding} weight effects over latent treatment resistance. We also relate this signal-based coefficient to the complier averages studied by \citet{imbens1994identification} and \citet{angrist1995two}.

The outcome model $Y = Y(0) + X\Delta$ is the correlated random coefficients (CRC) model \citep{garen1984,wooldridge1997two,heckman1998instrumental,card2001estimating}. It naturally nests the canonical binary treatment case of local average treatment effect (LATE), and the SWATE becomes a weighted combination of LATEs. It also extends the problem to continuous $X$ with linear effects. Two caveats worth noting. With continuous $X$, the linear effect form is a restriction adopted for tractability. The weight function $\omega(\cdot)$ is not identified, only $\beta_{IV}$ and its convexity are, and we do not attempt to estimate the weights. Classical CRC results give conditions under which IV recovers the average effect. The SWATE characterizes what IV recovers absent those conditions, which reappear as the special case $\omega = 1$.

Our second contribution provides a microfoundation for convex SWATE weights. For the population first stage $g_X = \E[X\mid Z,W]$, we show that the vector monotonicity of \citet{goff2024vector}, equivalently the actual monotonicity of \citet{mogstad2021causal}, together with positive association of $Z$ given $W$ delivers a new sufficient condition. We call it Conditional Covariance Monotonicity (CCM), under which SWATE weights are convex. We retain \citet{goff2024vector}'s choice model but add a restriction on the instrument distribution to interpret the projection coefficient directly, rather than a particular complier parameter. Scalar threshold model also implies CCM under extended exogeneity \citep{vytlacil2002independence}, but CCM can hold without a monotone single-index model.

Our interpretation holds for any signal. Estimating the optimal signal is where standard inference theory breaks, and our third contribution addresses this issue. The naive debiased moment, implemented in the standard DML routines such as \texttt{DoubleML} package in \texttt{R} and \texttt{Python}, is not Neyman orthogonal in $g_X$ under heterogeneity. Its first-order bias consists of the component of the instrument's effect on outcomes that the linear signal fails to absorb, and the first-stage estimation error of $g_X$. Because this bias is the covariance of those two terms, it has a geometric consequence. Rescalings and all covariate-measurable recalibrations of the signal leave the estimate unchanged, and only ``shape'' change aligned with that unabsorbed component moves the estimate. This is why well-tuned learners in the standard DML routines frequently appear valid, and why that appearance cannot be relied upon.

The first-order bias is a drift toward the learner's own SWATE, $\beta_{IV}(\hat{g}_X)$. Under our rate conditions, naive inference is asymptotically valid for this moving target, which is convex whenever CCM holds for the learned signal. Under heterogeneity along the signal, however, the procedure cannot attain valid inference for the fixed target uniformly.

For the fixed target, we construct a CRC robust orthogonal score with one new nuisance, $q_Y = \E[Y\mid Z,W]$. It is Neyman orthogonal in all four nuisances, and is the efficient influence function after Jacobian normalization. Its estimator is $\sqrt{N}$-consistent under a product-rate condition allowing a slowly learned $q_Y$. At the population nuisance values it equals the naive score under homogeneity, so the correction has no efficiency cost. 

The choice of $g_X$ is classical under homogeneity \citep{chamberlain1987asymptotic,newey1990semiparametric,belloni2012sparse}. In the binary examiner-IV specification, \citet{kolesar2013estimation} introduces the covariate-adjusted target, and \citet{sithole2026locally} develops locally robust estimation for it with a flexibly learned first stage. We embed that specialization in a broader theory. We show that arbitrary learned signals have SWATE interpretations, heterogeneity creates learner-induced target drift, and fixed-target analysis adds efficiency and robustness. We also develop a Hausman-type diagnostic and show how a projection first stage can correct the point estimate while leaving its conventional standard error inconsistent \citep{kolesar2013estimation,lee2018consistent}. An Anderson--Rubin confidence set can be used to address weak residual signal after nonparametric partialling, but not the classical many weak-instrument problem \citep{staiger1997instrumental,dufour1997some,andrews2019weak,mikusheva2022inference}.

The remainder of the paper is organized as follows. Section~\ref{sec:setup} defines the estimand, proving the representation theorems, microfounding CCM, and relating the framework to LATE. Sections~\ref{sec:estimation} and~\ref{sec:fixedtarget} then develop the moving target that naive DML estimates and the CRC robust score that recovers the fixed, learner-invariant target. SSection~\ref{sec:diagnostics} supplies the diagnostic and robust inference under a weak residual signal. Sections~\ref{sec:montecarlo}--\ref{sec:application} present the Monte Carlo and empirical evidence, and Section~\ref{sec:conclusion} concludes.

\section{Identification and Interpretation}\label{sec:setup}

\subsection{The CRC Model and the Partialled-out Estimand}\label{sec:setup-sub}

Let $(\Omega, \mathcal{F}, \mathbb{P})$ be a probability space. We observe an independent and identically distributed sample of random vectors $(Y, X, Z, W)$. For notational clarity, we suppress individual subscripts. Throughout, ``almost surely'' (a.s.) means with probability one. For conditional statements, it means for almost every value of the conditioning variables.
\begin{itemize}
  \item $Y \in \mathbb{R}$ is the observed continuous outcome.
  \item $X \in \mathcal{X} \subseteq \mathbb{R}$ is the observed endogenous treatment. This accommodates discrete treatments (e.g., binary $\mathcal{X} = \{0,1\}$) as well as continuous treatments (e.g., dosage, duration).
  \item $Z \in \mathcal{Z}$ is the observed instrumental variable, where $\mathcal{Z}$ can be a scalar, a discrete vector, or a high-dimensional feature space $\mathbb{R}^d$.
  \item $W \in \mathcal{W}$ is a high-dimensional vector of observed pre-treatment covariates.
  \item $g: \mathcal{Z} \times \mathcal{W} \to \mathbb{R}$ is a measurable signal function with finite second moments. The function $g(Z, W)$ acts as a predictive projection of the instrument (e.g., a non-parametric cross-fitted machine learner). For notational brevity, we refer to it as $g(Z)$ where the dependence on $W$ is implicit, or simply $g$.
\end{itemize}

We formalize the underlying causal data generating process within the canonical Correlated Random Coefficient (CRC) framework \citep{wooldridge1997two, heckman1998instrumental, angrist2000interpretation}. Let $Y(0)$ denote the untreated potential outcome, and let $\Delta$ denote the heterogeneous treatment effect. The structural outcome equation is
\begin{equation}\label{eq:crc}
  Y = Y(0) + X \Delta.
\end{equation}

\begin{remark}\label{rem:linearity}
  When $X$ is binary, Equation~\eqref{eq:crc} inherently holds under SUTVA, with $\Delta = Y(1)-Y(0)$. When $X$ is continuous, the CRC model imposes a constant marginal response for each individual, i.e., $\partial Y_i / \partial X_i = \Delta_i$. Here $X$ is endogenous because individuals may self-select intensity on their untreated $Y(0)$ or effects $\Delta$ (essential heterogeneity). This linearity is the restriction we pay for tractability with continuous $X$. \citet{angrist2000interpretation} relax this linearity restriction and interpret IV as a weighted average of marginal responses. Under our CRC model, Theorem~\ref{thm:cpt} below expresses the estimand as a signal-weighted average of the individual effects $\Delta$.
\end{remark}

Following \citet{robinson1988root} and the Frisch--Waugh--Lovell theorem, we partial out the covariates $W$. Let $\tilde{A} = A - \E[A \mid W]$ denote the conditional residual of any random variable $A$. The generalized partialled-out instrumental variable (IV) estimand is defined as
\begin{equation}\label{eq:betaiv}
  \beta_{IV} := \frac{\E[\tilde{Y} \tilde{g}]}{\E[\tilde{X} \tilde{g}]}.
\end{equation}

When $g$ is a fixed instrument and only the covariate projections $\E[Y\mid W], \E[X\mid W]$ are learned, the partialled-out IV moment is already Neyman orthogonal, so the double/debiased machine-learning theory of \citet{chernozhukov2018double} delivers valid $\sqrt{N}$ inference regardless of effect heterogeneity. When the signal is learned, with $g$ an estimated first stage pooling many instruments, orthogonality breaks under essential heterogeneity. The identification results of this section hold for any measurable $g$ and cover both regimes. The estimation theory of Sections~\ref{sec:estimation}--\ref{sec:diagnostics} is what the second regime requires.

Having fixed the estimand, we now state the identifying assumptions and prove the two representation theorems that give $\beta_{IV}$ its structural meaning.

\subsection{Assumptions and discussion}\label{sec:assumptions-disc}\label{sec:assumptions}

\begin{assumption}[Conditional Exogeneity of Instruments]\label{ass:exog}
  \[
    Z \indep \big( Y(0), \Delta \big) \mid W.
  \]
  We implicitly assume finite second moments for all relevant variables.
\end{assumption}

\begin{assumption}[Conditional Relevance]\label{ass:relevance}
  $\E[\tilde{X} \tilde{g}] \neq 0 $.
\end{assumption}

Because $\beta_{IV}$ is invariant to rescaling the signal, and in particular to its sign, formally shown in Lemma~\ref{lem:hgeom} below, we adopt throughout the harmless normalization $\E[\tilde{X}\tilde{g}]>0$. This is a convention rather than an additional restriction, as it fixes the denominator's sign so that the following discussion of non-negative weights and convex combination read literally.

\begin{assumption}[Conditional Covariance Monotonicity, CCM]\label{ass:ccm}
  Let $F_\Delta(\cdot)$ be the marginal cumulative distribution function of $\Delta$, with support $\Supp(\Delta)$. Then
  \[
    \E\big[\tilde{X} \tilde{g} \mid \Delta = \delta\big] \ge 0 \quad \text{a.s.}
  \]
\end{assumption}

Within an effect stratum, CCM requires that the signal not, on average, push treatment the wrong way for units sharing a given effect. The failure case is a high effect subgroup whose treatment responds negatively to the signal, which can then receive negative SWATE weight. Because CCM constrains only a conditional covariance rather than individual response functions, it is generally weaker than the monotonicity assumptions of the literature. We show in following sections that monotonicity in \citet{imbens1994identification} implies it in the scalar threshold model, and vector monotonicity \citep{goff2024vector} together with positive association of the instruments implies it for $g_X$ with multiple instruments.

Note that CCM is not directly testable, and its weight on $\Delta$ is not identified. Our defense is instead microfoundations from primitive choice behavior in Section~\ref{sec:microccm}, which restricts only the observable distribution of $Z$ given $W$.

Conditional Relevance also matters for inference. The SWATE representation remains valid whenever $\E[\tilde{X}\tilde{g}]\neq 0$, but estimation can be imprecise when the signal weakly predicts treatment after partialling out $W$. This motivates the confidence sets in Section~\ref{sec:diagnostics}, which invert tests based on the orthogonal score without dividing by an estimated Jacobian.

\subsection{The representation theorems}\label{sec:representation}

\begin{theorem}[Signal-weighted ATE]\label{thm:representation}\label{thm:cpt}\label{thm:weights}
  Suppose Assumptions~\ref{ass:exog} and~\ref{ass:relevance} hold, and let $g$ be any measurable signal. Then
  \begin{itemize}
    \item[(i)] the generalized partialled-out IV estimand $\beta_{IV}$ admits a representation as a Signal-Weighted Average Treatment Effect (SWATE),
          \[
            \beta_{IV} = \E_\Delta \big[ \Delta \cdot \omega(\Delta) \big] = \int_{\Supp(\Delta)} \delta \cdot \omega(\delta)\, dF_\Delta(\delta),
            \qquad
            \omega(\delta) = \frac{\E[\tilde{X} \tilde{g} \mid \Delta = \delta]}{\E[\tilde{X} \tilde{g}]};
          \]
    \item[(ii)] the weighting function integrates to one,
          \[
            \E_\Delta\big[\omega(\Delta)\big] = 1;
          \]
    \item[(iii)] if in addition Assumption~\ref{ass:ccm} (CCM) holds, $\omega(\delta) \ge 0$ a.s., and consequently $\beta_{IV}$ identifies a well-defined convex combination of the heterogeneous treatment effects $\Delta$.
  \end{itemize}
\end{theorem}

The representation draws only on conditional exogeneity and relevance. It therefore holds for any measurable $g$, in particular a black-box, machine-learned first stage. This is what equips otherwise opaque machine-learning IV estimands with an exact structural meaning.

In the scalar one-instrument case this weighting scheme is not new. \citet{huntingtonklein2020instruments} shows that IV under first-stage heterogeneity identifies the first stage effect weighted average $\E[\gamma\beta]/\E[\gamma]$, with $\gamma$ the individual first-stage response. This is Theorem~\ref{thm:representation} for a single raw instrument, with his $\gamma$ in the role of our $\E[\tilde{X}\tilde{g}\mid\Delta]$. Theorem~\ref{thm:representation} generalizes it to any measurable signal $g(Z,W)$. Our weights act on the structural effects $\Delta$, whereas the Heckman--Vytlacil identification invoked by \citet{chen2021mostly} places weights on marginal treatment effects, convex only under a single-index structure that Theorem~\ref{thm:representation} does not require.

Textbook just-identified IV is also covered. Taking the signal to be a fixed scalar instrument, $g = Z$, Theorem~\ref{thm:cpt} gives
\[
  \beta_{IV} = \frac{\E[\tilde{Y}\tilde{Z}]}{\E[\tilde{X}\tilde{Z}]}
  = \E_\Delta\big[\Delta\,\omega(\Delta)\big],
  \qquad
  \omega(\delta) = \frac{\E[\tilde{X}\tilde{Z}\mid\Delta=\delta]}{\E[\tilde{X}\tilde{Z}]},
\]
the Wald/2SLS estimand with covariates partialled out.

\subsection{Covariance Monotonicity and its Microfoundation}\label{sec:microccm}

As mentioned, CCM is not directly testable and its warrant is that primitive models of treatment choice imply it. It holds automatically in the single-index threshold model $X = h(g(Z), U, W)$ with $h$ weakly increasing in $g$. The map $g \mapsto \E[X \mid g, W, \Delta=\delta]$ is then non-decreasing and covaries non-negatively with $g$ by Chebyshev's association inequality, so CCM holds by construction. But a common monotone index requires any given change in the instruments to affect treatment in a common direction across individuals with the same covariates. \citet{mogstad2021causal} show why this restriction can be strong with multiple instruments. Their response is to relax single-index monotonicity to a multidimensional monotone choice model. They work with partial monotonicity, under which each instrument moves treatment in a common direction across individuals while that direction may itself depend on the levels of the other instruments. The stronger actual monotonicity, in which each instrument's direction is fixed globally, coincides with the vector monotonicity of \citet{goff2024vector}. We show that, combined with a positive association assumption, vector monotonicity implies CCM, ensuring non-negative SWATE weights without any single-index restriction.

Let the treatment be generated by an arbitrary measurable structural map
\begin{equation}\label{eq:structural-h}
  X = h(Z, W, V),
\end{equation}
where $V \in \mathcal{V}$ is an unobserved, possibly infinite-dimensional vector of choice heterogeneity. This nests the single-index model $h(g(Z), U, W)$ but allows the relative responses to different coordinates of $Z$ to vary across individuals. We take the signal to be the optimal first stage
\[
  g_X(Z,W) := \E[X \mid Z, W].
\]
The same projection is used for estimation in Section~\ref{sec:estimation}, and write $m_X(W) := \E[X \mid W] = \E[g_X \mid W]$, so that $\tilde{g}_X = g_X - m_X$.

\begin{assumption}[Extended Conditional Exogeneity]\label{ass:exogext}
  \[
    Z \indep \big( Y(0), \Delta, V \big) \mid W.
  \]
\end{assumption}

\begin{assumption}[Vector Monotonicity, VM]\label{ass:pm}
  Almost surely, the choice map $z \mapsto h(z, W, V)$ is coordinate-wise non-decreasing on $\mathbb{R}^d$.
\end{assumption}

\begin{assumption}[Positive Association, PA]\label{ass:pa}
  For any two coordinate-wise non-decreasing functions $\phi, \psi : \mathbb{R}^d \to \mathbb{R}$ with finite second moments,
  \[
    \Cov\big(\phi(Z), \psi(Z) \mid W\big) \ge 0 \quad \text{a.s.}
  \]
\end{assumption}

Assumption~\ref{ass:exogext} extends the conditional independence restriction in Assumption~\ref{ass:exog} to include the selection unobservable $V$. VM follows \citet{goff2024vector} and coincides with the actual monotonicity of \citet{mogstad2021causal}. It assumes each individual's choice map is coordinate-wise non-decreasing under a common orientation of the instruments. Extended exogeneity then ensures that both $g_X$ and the mean treatment response within each effect stratum inherit this monotonicity.

PA makes the covariance between these functions non-negative conditional on $W$ \citep{esary1967association}. PA holds under any of the following conditions on the distribution of $Z\mid W$: mutually independent coordinates, affiliation \citep{milgrom1982theory}, or joint Gaussianity with non-negative pairwise correlations \citep{pitt1982positively}. Outside such special cases, non-negative pairwise correlations alone do not guarantee PA.

Two practical points attach to PA. First, unlike CCM, it is refutable. It restricts only the observable law of $Z$ given $W$, so its pairwise implications (positive quadrant dependence of each $(Z_i,Z_j)$ given $W$) can be assessed with the data, even if full multivariate association is harder to test. Second, association is directional, so the instruments must be sign-normalized first. Orient each coordinate of $Z$ in the direction of its own first-stage effect on $X$ so that VM reads coordinate-wise non-decreasing, and impose PA on the oriented vector. This is a design step for the analyst rather than an extra assumption.

\begin{proposition}[Microfounding CCM via Vector Monotonicity]\label{prop:pmccm}
  Suppose Assumptions~\ref{ass:exogext}--\ref{ass:pa} hold and $\E[X^2] < \infty$. Then, for the signal $g_X(Z,W) = \E[X \mid Z,W]$, CCM holds,
  \[
    \E\big[\tilde{X}\,\tilde{g}_X \mid \Delta = \delta\big] \ge 0 \quad \text{a.s.}
  \]
  Consequently, under Assumption~\ref{ass:relevance}, Theorem~\ref{thm:weights} applies.
\end{proposition}

\begin{remark}\label{rem:pa-binding}
  Assumption~\ref{ass:exogext} and VM alone do not imply CCM. If individuals with effect $\delta$ respond mainly to one instrument while $g_X$ is driven by another, dependence between the instruments can make their mean treatment response covary negatively with $g_X$ conditional on $W$. A negative average of these covariances over $W\mid\Delta=\delta$ gives $\omega(\delta)<0$.
\end{remark}

Mutually exclusive assignment indicators illustrate a limitation of PA. Within a court $\times$ period cell, indicators for assignment to different judges have negative covariance. With three or more such indicators, sign changes cannot eliminate all negative pairwise covariances. The argument based on VM and PA therefore provides no convexity guarantee for this representation. Since PA is sufficient but not necessary for CCM, its failure does not imply negative SWATE weights.

\subsection{Relation to the LATE Framework}\label{sec:multipleLATE}

We now specialize to a threshold-crossing model with binary treatment. Under the conditions stated below, the SWATE decomposes into a convex combination of margin-specific LATEs.

First fix a covariate value $w$. Define the conditional first-stage covariance and the corresponding within-stratum covariance ratio by
\[
  D(w):=\E[\tilde X\tilde g\mid W=w],
  \qquad
  \beta_{IV}(w):=\frac{\E[\tilde Y\tilde g\mid W=w]}{D(w)}.
\]
Until Corollary~\ref{cor:late-unconditional}, all expectations, probabilities, and distributions are taken under the conditional distribution given $W=w$, and the dependence of the objects below on $w$ is suppressed. The paper's unconditional estimand $\beta_{IV}$ is not an unweighted average of the ratios $\beta_{IV}(w)$, and the corollary gives the weighted aggregation.

Consider a binary treatment $X\in\{0,1\}$ generated by a canonical threshold-crossing model
\[
  X \;=\; \mathbf{1}\{\,U \le g(Z,W)\,\},
\]
where $U$ is unobserved treatment resistance. We use Assumption~\ref{ass:exogext} with the selection unobservable specialized to $V=U$. Because $g$ is measurable in $(Z,W)$, this condition implies
\[
  g(Z,W)\indep\bigl(Y(0),\Delta,U\bigr)\mid W.
\]
Suppose the signal $g$ takes $L$ ordered values $g_1<g_2<\dots<g_L$ with $\mathbb{P}(g=g_\ell)=p_\ell$, write $\bar g:=\E[g]=\sum_\ell p_\ell g_\ell$, and recall $\tilde g = g-\bar g$, $\tilde X = X-\E[X]$. For $\ell=2,\dots,L$ define the margin-$\ell$ complier event
\[
  C_\ell \;:=\; \{\, g_{\ell-1} < U \le g_\ell \,\},
\]
so that an individual is treated once the signal reaches level $g_\ell$ but not at $g_{\ell-1}$. Units with $U\le g_1$ are always-takers and units with $U>g_L$ are never-takers. For both groups, $X$ is constant in $g$. Write $\pi_\ell(\delta):=\mathbb{P}(C_\ell\mid\Delta=\delta)$ and $\bar\pi_\ell:=\mathbb{P}(C_\ell)$, and define
\[
  \mathcal L_+ := \{\ell\in\{2,\dots,L\}:\bar\pi_\ell>0\},
  \qquad
  \mathrm{LATE}_\ell:=\E[\Delta\mid C_\ell],\quad \ell\in\mathcal L_+,
\]
where $\mathrm{LATE}_\ell$ is the local average treatment effect \citep{imbens1994identification} for margin-$\ell$ compliers. Zero-probability margins are omitted throughout. Define the covariance increments
\[
  S_\ell \;:=\; \E\!\big[(g-\bar g)\,\mathbf{1}\{g\ge g_\ell\}\big]
  \;=\; \mathbb{P}(g\ge g_\ell)\,\big(\E[g\mid g\ge g_\ell]-\E[g]\big)
  \;=\; \Cov\!\big(\mathbf{1}\{g\ge g_\ell\},\,g\big)\;\ge\;0 ,
\]
where the inequality holds because $\mathbf{1}\{g\ge g_\ell\}$ and $g$ are both nondecreasing in $g$.

\begin{proposition}[Multiple-Instrument LATE Decomposition]\label{thm:multipleLATE}
  Suppose Assumption~\ref{ass:exogext} holds with $V=U$. For almost every covariate value $w$ with $D(w)>0$, the threshold-crossing model above implies:
  \begin{itemize}
    \item[(i)] the projection estimand is a convex combination of margin-specific LATEs,
          \[
            \beta_{IV}(w) \;=\; \sum_{\ell\in\mathcal L_+} \lambda_\ell\,\mathrm{LATE}_\ell,
            \qquad
            \lambda_\ell \;:=\; \frac{S_\ell\,\bar\pi_\ell}{\sum_{\ell'\in\mathcal L_+} S_{\ell'}\,\bar\pi_{\ell'}},
          \]
          with $\lambda_\ell\ge0$ and $\sum_{\ell\in\mathcal L_+}\lambda_\ell=1$;
    \item[(ii)] for almost every $\delta$, the within-stratum SWATE weighting function
          \[
            \omega_w(\delta):=\frac{\E[\tilde X\tilde g\mid\Delta=\delta,W=w]}{D(w)}
          \]
          admits the closed form
          \begin{equation}\label{eq:complierdensity}
            \omega_w(\delta) \;=\; \sum_{\ell\in\mathcal L_+} \lambda_\ell\,
            \frac{dF_{\Delta\mid C_\ell}(\delta)}{dF_\Delta(\delta)},
          \end{equation}
          with the same weights $\lambda_\ell$, where all displayed distributions are conditional on $W=w$ and $dF_{\Delta\mid C_\ell}/dF_\Delta$ is the complier likelihood ratio, equivalently $\mathbb{P}(C_\ell\mid\Delta=\delta)/\mathbb{P}(C_\ell)$.
  \end{itemize}
\end{proposition}

\begin{corollary}[Unconditional Decomposition]\label{cor:late-unconditional}
  Under the threshold-crossing model, suppose Assumptions~\ref{ass:exogext} and~\ref{ass:relevance} hold, with $V=U$ in Assumption~\ref{ass:exogext} and the positive-denominator normalization, and suppose $g$ has finite conditional support for almost every covariate value. Restoring the dependence on $w$ in Proposition~\ref{thm:multipleLATE},
  \begin{equation}\label{eq:unconditional-late}
    \begin{aligned}
      \beta_{IV}
       & =\frac{\displaystyle\int_{\{w:\,D(w)>0\}}
      D(w)\beta_{IV}(w)\,dF_W(w)}{\E[D(W)]}                  \\
       & =\int_{\{w:\,D(w)>0\}}\sum_{\ell\in\mathcal L_+(w)}
      \frac{D(w)\lambda_\ell(w)}{\E[D(W)]}
      \mathrm{LATE}_\ell(w)\,dF_W(w).
    \end{aligned}
  \end{equation}
  The displayed joint $(W,\ell)$ weights are non-negative and integrate to one, so $\beta_{IV}$ is a convex mixture of the covariate- and margin-specific local effects.
\end{corollary}

Corollary~\ref{cor:late-unconditional} makes the covariate aggregation precise. A stratum is weighted by its conditional first-stage covariance $D(w)$ as well as its population frequency. A stratum with $D(w)=0$ contributes zero to both unconditional moments, even though its conditional ratio is undefined.

Unlike the aggregate two-stage least squares estimand, whose multiple-instrument weights can be negative \citep{mogstad2021causal}, the scalar-signal projection has non-negative weights both within covariate strata and unconditionally. The non-negative covariance increments $S_\ell$ generate the threshold-model CCM kernel, so collapsing the instruments into $g$ turns the ``causal salad'' back into a convex combination of margin-specific local effects. The $\lambda_\ell$ are the ordered instrument analogue of \citet{angrist1995two}. Where \citet{mogstad2021causal} obtain a decomposition of the 2SLS estimand and \citet{goff2024vector} identify complier parameters under vector monotonicity, we characterize what the workhorse projection estimand is, and Section~\ref{sec:estimation} shows how to make its inference robust.

Part (ii) shows where the within-stratum weight comes from. Equation~\eqref{eq:complierdensity} expresses $\omega_w$ as a $\lambda$-convex average of complier likelihood ratios, and the ratio $dF_{\Delta\mid C_\ell}/dF_\Delta$ measures how over- or under-represented an effect $\delta$ is among margin-$\ell$ compliers in that stratum, so $\omega_w(\delta)>1$ when $\delta$ is over-represented. An effect $\delta^\ast$ occurring only among always- or never-takers has $\pi_\ell(\delta^\ast)=0$, hence $\omega_w(\delta^\ast)=0$. This is the CRC generalization of LATE's always-/never-taker zero-weighting, and the formal counterpart of \citet{huntingtonklein2020instruments}'s scalar observation that units with $\gamma_i\approx 0$ receive vanishing weight.

\section{The Moving Target: What Naive DML Estimates}\label{sec:estimation}

Everything to this point concerns the population signal. We now show the fact that the optimal signal $g_X=\E[X\mid Z,W]$ must be learned, and that learning it breaks the standard moment.

The sample analogue of the covariance ratio is not robust to the first-order estimation error of a machine-learned signal. Under essential heterogeneity, the naive moment is not Neyman orthogonal, so plugging in a slowly converging $\hat{g}$ contaminates $\hat{\beta}_{IV}$ with regularization bias. This section formalizes that obstruction and quantifies it, showing that the naive plug-in validly estimates a moving target that drifts with the learner. Section~\ref{sec:fixedtarget} then constructs a locally robust score in the sense of \citet{chernozhukov2022locally}, which restores $\sqrt{N}$-consistent and asymptotically normal estimation of the fixed target.

\citet{chen2021mostly} establish asymptotic normality for a fixed coefficient in a linear IV model with machine-learned instruments. Their moment restrictions on the structural error eliminate the first-order effect of estimating the instrument. Selection on heterogeneous effects can violate these restrictions in our CRC model, allowing first-stage estimation error to enter the moment for the coefficient.

We adopt one notational convention for the whole section. The standalone symbol $\Delta$ always denotes the structural effect of~\eqref{eq:crc}, whereas a nuisance estimation error always carries a nuisance operand, as in $\Delta l_Y := \hat{l}_Y - l_Y$ and likewise $\Delta m_X, \Delta q_Y, \Delta g_X$, collected as $\Delta\eta := \hat{\eta} - \eta_0$. The direction $\delta_g \in L_2(\sigma(Z,W))$ denotes a perturbation of the signal.

\subsection{The Learned Signal and the Target}\label{sec:nuisances}

To make the estimand feasible to estimate, we fix the generic signal to the structurally optimal quantity
\begin{equation}\label{eq:gx}
  g(Z,W) = g_X(Z,W) := \E[X \mid Z, W].
\end{equation}
This is the projection of the instrument vector $Z$ onto the treatment, which is what a first-stage learner targets. The four nuisance functions are collected in the vector $\eta = (l_Y, m_X, q_Y, g_X)$,
\[
  \begin{gathered}
    l_Y(W) := \E[Y\mid W], \qquad m_X(W) := \E[X\mid W], \\[2pt]
    q_Y(Z,W) := \E[Y\mid Z,W], \qquad g_X(Z,W) := \E[X\mid Z,W].
  \end{gathered}
\]
Note that $\E[g_X\mid W] = m_X$, and the residualized signal is $\tilde{g} = g_X - m_X$. The estimand~\eqref{eq:betaiv} can be written as
\begin{equation}\label{eq:beta0}
  \beta_0 \;:=\; \beta_{IV}(g_X) \;=\; \frac{\E\big[(Y - l_Y)(g_X - m_X)\big]}{\E\big[(X - m_X)(g_X - m_X)\big]} \;=\; \frac{\E\big[(Y - l_Y)(g_X - m_X)\big]}{\E\big[(g_X - m_X)^2\big]},
\end{equation}
where the last equality uses $\E[(X - m_X)(g_X - m_X)] = \E[(g_X - m_X)^2]$, which holds because $\E[X - m_X\mid Z,W] = g_X - m_X$. As a specialization of $\beta_{IV}$ with optimal first stage $g_X$, the estimand $\beta_0$ inherits the SWATE interpretation of Theorem~\ref{thm:cpt}.

For a binary examiner-IV design, relabel $X$ as $T$ and define the residualized leniency instrument $\gamma := \E[T\mid Z,W]-\E[T\mid W]=g_X-m_X$. Because $\E[\gamma\mid W]=0$, Equation~\eqref{eq:beta0} becomes $\beta_0=\E[Y\gamma]/\E[T\gamma]$, and $\E[T\gamma]=\E[\gamma^2]$. Thus $\beta_0$ is the covariate-adjusted examiner-IV estimand of \citet{kolesar2013estimation} and the fixed target studied under flexible first-stage estimation by \citet{sithole2026locally}. At the population nuisance values, substituting Sithole's Riesz representers $\alpha_1=q_Y-\beta g_X$ and $\alpha_2=-l_Y+\beta m_X$ into his oracle score makes it algebraically identical to the CRC robust score of Section~\ref{sec:fixedtarget}. Sithole focuses on locally robust estimation of this binary examiner specialization, whereas our identification result covers arbitrary signals and our subsequent analysis characterizes learner-induced target drift and establishes the fixed target's efficient influence function.

Three properties recommend $g_X$ as the signal. First, it is the maximal-relevance choice as the projection maximizes covariance with $\tilde{X}$. Second, $\beta_{IV}$ is invariant to rescaling and to $W$-measurable recalibration of the signal (Lemma~\ref{lem:hgeom}). Third, $g_X$ is the classical optimal instrument under homogeneous treatment effects and homoskedastic outcome errors \citep{chamberlain1987asymptotic,newey1990semiparametric}.

\begin{definition}[Heterogeneity along the Signal]\label{def:esshet}
  The structural effects are correlated with selection in a manner that the instrument shifts,
  \[
    \mathbb{P}\Big( \E[X\Delta\mid Z,W] - \E[X\Delta\mid W] \;\neq\; \beta_0\,(g_X - m_X) \Big) > 0.
  \]
\end{definition}

To read this definition, define the heterogeneity residual
\begin{equation}\label{eq:Hdef}
  \mathcal{H}(Z,W) := \big(\E[X\Delta\mid Z,W] - \E[X\Delta\mid W]\big) - \beta_0\,(g_X - m_X).
\end{equation}
Definition~\ref{def:esshet} is the statement $\mathbb{P}(\mathcal{H}\neq 0) > 0$. When effects are constant, i.e., $\Delta = \beta_0$, one has $\E[X\Delta\mid Z,W] = \beta_0\,g_X$ and $\E[X\Delta\mid W] = \beta_0\,m_X$, so $\mathcal{H} = 0$. Under essential heterogeneity, $\mathcal{H}$ is the component of the instrument's effect on $\E[X\Delta\mid Z,W]$ that the linear signal $\beta_0(g_X - m_X)$ fails to absorb. We therefore call Definition~\ref{def:esshet} throughout as ``heterogeneity along the signal.''

The four nuisances $\eta$ are unknown and are estimated from the sample by $K$-fold cross-fitting \citep{chernozhukov2018double}. Partition the observations into $K \ge 2$ folds $\mathcal{I}_1, \dots, \mathcal{I}_K$ of equal size $n_k = N/K$, with $K$ fixed as $N \to \infty$, and write $\mathcal{D}_{-k}$ for the complementary subsample of observations outside $\mathcal{I}_k$. On each fold we fit the nuisance vector $\hat{\eta}_k = (\hat{l}_{Y,k}, \hat{m}_{X,k}, \hat{q}_{Y,k}, \hat{g}_{X,k})$ using only $\mathcal{D}_{-k}$, by any machine learner (such as lasso, random forests, gradient boosting, or an ensemble), and evaluate it on the held-out fold $\mathcal{I}_k$. Its $j$-th coordinate is written $\hat{\eta}_{j,k}$. Because $\hat{\eta}_k$ is trained without the observations at which it is used, it is independent of them and may be held fixed conditional on $\mathcal{D}_{-k}$, so overfitting does not enter the moment. When no fold index is displayed, $\hat{\eta}$ refers to this out-of-fold estimator, and fold by fold nuisance error is denoted as $\Delta\eta_k := \hat{\eta}_k - \eta_0$. The realized precision of the first stage is measured by the fold-wise $L_2$ errors
\[
  \rho_{j,k} := \|\hat{\eta}_{j,k} - \eta_{0,j}\|_{\mathbb{P},2} \ \ (j \in \{l_Y, m_X, q_Y, g_X\}), \qquad \rho_N := \max_{j,k}\rho_{j,k}.
\]
The next two assumptions place regularity and rate conditions on this construction.

\begin{assumption}[Boundedness and Moments]\label{ass:regB}
  (i) $\E[Y^4] + \E[X^4] \le C_4 < \infty$ and $\E[\Delta^2] + \E[X^2\Delta^2] < \infty$. (ii) There is $\bar{C} < \infty$ such that $\max_j \|\eta_{0,j}\|_\infty \le \bar{C}$ and, a.s., $\max_j \|\hat{\eta}_{j,k}\|_\infty \le \bar{C}$ for every fold $k$.
\end{assumption}

Part (ii) is innocuous. Truncating each fitted nuisance at $\bar{C} \ge \max_j\|\eta_{0,j}\|_\infty$ preserves cross-fitting and weakly reduces $L_2$ error, so it can alwaysbe imposed by the analyst.

\begin{assumption}[Product Rate Conditions]\label{ass:ratesP}
  (i) $\rho_N = o_{\mathbb{P}}(1)$; and (ii) $\max_k\big\{\rho_{m,k}^2 + \rho_{g,k}^2\big\} = o_{\mathbb{P}}(N^{-1/2})$ and $\max_k\big(\rho_{l,k} + \rho_{q,k}\big)\big(\rho_{m,k} + \rho_{g,k}\big) = o_{\mathbb{P}}(N^{-1/2})$.
\end{assumption}

Assumption~\ref{ass:ratesP}(ii) holds in particular under the single rate $\|\hat{\eta}_{j,k} - \eta_{0,j}\|_{\mathbb{P},2} = o_{\mathbb{P}}(N^{-1/4})$ for each nuisance, which is attainable for lasso, random forests, and gradient boosting under standard sparsity or smoothness. The product form relaxes this single rate condition. The outcome-side projections $l_Y, q_Y$ may converge arbitrarily slowly provided the treatment-side projections $m_X, g_X$ compensate, which is valuable because $q_Y = \E[Y \mid Z, W]$ is typically the hardest nuisance to learn.

The standard estimator in DML routines such as \texttt{DoubleML} package in \texttt{R} and \texttt{Python} sets the empirical average of the canonical partialled-out moment
\begin{equation}\label{eq:psinaive}
  \psi_{\mathrm{naive}}(O; \beta, \eta) := \big(Y - l_Y - \beta(X - m_X)\big)(g_X - m_X)
\end{equation}
to zero. This moment identifies $\beta_0$, since its expectation vanishes at $(\beta_0,\eta_0)$ by construction, but it is not robust to estimation error in the signal $g_X$.

With a single instrument and only the covariate nuisances $(l_Y,m_X)$ estimated, the partially linear IV score of \citet{chernozhukov2018double} is Neyman orthogonal irrespective of heterogeneity. The obstruction below is specific to the learned-signal variant, where $g_X=\E[X\mid Z,W]$ is itself estimated, and it arises only under heterogeneity along the signal in Definition~\ref{def:esshet}.

\begin{proposition}[Non-Orthogonality of the Naive Score]\label{prop:nonorth}
  Under Assumption~\ref{ass:exog}, view the naive moment as a functional of the signal, $g \mapsto \E[\psi_{\mathrm{naive}}(O;\beta_0,(l_Y,m_X,q_Y,g))]$. Its pathwise (G\^{a}teaux) derivative at $g = g_X$, in any direction $\delta_g \in L_2(\sigma(Z,W))$, is
  \[
    \begin{aligned}
      \partial_{g_X}\,\E\big[\psi_{\mathrm{naive}}(O;\beta_0,\eta_0)\big][\delta_g]
       & := \frac{d}{dt}\Big|_{t=0}\,\E\big[\psi_{\mathrm{naive}}\big(O;\beta_0,(l_Y,m_X,q_Y,g_X + t\,\delta_g)\big)\big] \\
       & = \E\big[\mathcal{H}(Z,W)\,\delta_g(Z,W)\big].
    \end{aligned}
  \]
  Consequently, under essential heterogeneity, the score is not Neyman orthogonal with respect to $g_X$. Taking $\delta_g = \mathcal{H}$ then gives $\E[\mathcal{H}^2] > 0$.
\end{proposition}

A black-box first-stage learner $\hat{g}_X$ converges at a non-parametric rate slower than $N^{-1/2}$, and Proposition~\ref{prop:nonorth} shows its error enters the moment at first order. The first-order influence of the estimated signal is $\E[\mathcal{H}\,\Delta g_X]$. Because it is a covariance, only the component of the first-stage error aligned with $\mathcal{H}$ can move the estimate. In the following section, we will show that this first-order consequence for the estimator is a drift of the estimand rather than noise.

\subsection{What Naive DML Estimates}\label{sec:geometry}

This section characterizes the first-order bias identified in Proposition~\ref{prop:nonorth}. We first describe the geometry of $\mathcal{H}$, which governs the first-stage errors that the naive moment can and cannot detect. We then show that the bias is a drift of the estimand, so that the naive plug-in performs valid inference on a moving target. For simplicity, we write $\|\cdot\| := \|\cdot\|_{\mathbb{P},2}$.

\begin{lemma}[Orthogonal Structure of $\mathcal{H}$]\label{lem:hgeom}
  Let Assumptions~\ref{ass:exog} and~\ref{ass:relevance} hold. With $g = g_X$,
  \begin{itemize}
    \item[(i)] $\E[\mathcal{H} \mid W] = 0$ a.s.;
    \item[(ii)] $\E[\mathcal{H}\,\tilde{g}_X] = 0$;
    \item[(iii)] the set
          $\mathcal{G} := \big\{\, c\,\tilde{g}_X + d \;:\; c \in \mathbb{R},\; d \in L_2(\sigma(W)) \,\big\} \subset L_2(\sigma(Z,W))$ is a closed linear subspace, $\mathcal{H} \perp \mathcal{G}$, and for every direction $\delta_g \in L_2(\sigma(Z,W))$,
          \[
            \E[\mathcal{H}\,\delta_g]
            \;=\; \E\big[\mathcal{H}\,(I - \Pi_{\mathcal{G}})\delta_g\big],
            \qquad
            \big|\E[\mathcal{H}\,\delta_g]\big|
            \;\le\; \|\mathcal{H}\| \cdot \big\|(I - \Pi_{\mathcal{G}})\delta_g\big\|,
          \]
          where $\Pi_{\mathcal{G}}$ denotes orthogonal projection onto $\mathcal{G}$;
    \item[(iv)] for any $g \in L_2(\sigma(Z,W))$ with $\E[\tilde{X}\tilde{g}] \neq 0$, every $a \in \mathbb{R}\setminus\{0\}$ and every $d \in L_2(\sigma(W))$, the signal $g' := a\,g + d$ satisfies $\E[\tilde{X}\tilde{g}'] = a\,\E[\tilde{X}\tilde{g}] \neq 0$ and $\beta_{IV}(g') = \beta_{IV}(g)$. The SWATE estimand is invariant to rescalings of the signal and to arbitrary $W$-measurable recalibrations.
  \end{itemize}
\end{lemma}

Part (iii) shows the operator of Proposition~\ref{prop:nonorth} eliminates every perturbation in $\mathcal{G}$, a rescaled signal plus a $W$-measurable shift, leaving only the residual shape component $(I-\Pi_{\mathcal{G}})\delta_g$ to move the naive moment at first order. The obstruction is broader than essential heterogeneity. Decomposing $\mathcal{H} = (\bar{\Delta}_W - \beta_0)\tilde{g}_X + \tilde{\kappa}$ into an observable-heterogeneity term with $\bar{\Delta}_W := \E[\Delta\mid W]$ and a selection-on-gains term where $\tilde{\kappa} := \kappa - \E[\kappa\mid W]$, $\kappa := \Cov(X,\Delta\mid Z,W)$, Lemma~\ref{lem:hdecomp} in Appendix~\ref{sec:appenda} shows observable heterogeneity alone makes $\mathcal{H}\neq 0$.

We now formalize the drifting estimand of naive DML. We write the naive moment as
\begin{equation}\label{eq:naiveexact}
  \E\big[\psi_{\mathrm{naive}}(O;\beta_0,\eta)\big]
  \;=\; \E[\mathcal{H}\,\Delta g_X]
  \;+\; \E\big[(-\Delta l_Y + \beta_0\,\Delta m_X)(\Delta g_X - \Delta m_X)\big].
\end{equation}
The leading term is the drift loading, and the second is a product of nuisance errors, hence higher order.

Define $A_k := \E\big[(Y - \hat{l}_{Y,k})(\hat{g}_{X,k} - \hat{m}_{X,k}) \mid \mathcal{D}_{-k}\big]$ and $B_k := \E\big[(X - \hat{m}_{X,k})(\hat{g}_{X,k} - \hat{m}_{X,k}) \mid \mathcal{D}_{-k}\big]$,
\begin{equation}\label{eq:pseudotrue}
  \beta_k^* := \frac{A_k}{B_k},
  \qquad
  \beta_N^* := \frac{\sum_{k=1}^K n_k A_k}{\sum_{k=1}^K n_k B_k}
  \;=\; \sum_{k=1}^K w_k\,\beta_k^*,
  \qquad
  w_k := \frac{n_k B_k}{\sum_{j=1}^K n_j B_j}.
\end{equation}
Let $\bar{D}_N := \sum_{k=1}^K (n_k/N)\, \E\big[\mathcal{H}\,\Delta g_{X,k} \mid \mathcal{D}_{-k}\big]$ denote the average realized drift loading, where $\Delta g_{X,k} := \hat{g}_{X,k} - g_X$.

\begin{theorem}[The Moving Target: Learner-Induced SWATE]\label{thm:driftvalid}
  Let Assumptions~\ref{ass:exog} and~\ref{ass:relevance} hold with $g = g_X$. For parts~(iii)--(iv) let Assumptions~\ref{ass:regB} and~\ref{ass:ratesP} also hold and let $\hat{\beta}_{\mathrm{naive}}$ solve the cross-fitted empirical analogue of~\eqref{eq:psinaive}. Then
  \begin{itemize}
    \item[(i)] for every $\delta_g \in L_2(\sigma(Z,W))$ with $J_0 + \E[\tilde{g}_X\,\delta_g] \neq 0$,
          \begin{equation}\label{eq:driftexact}
            \beta_{IV}(g_X + \delta_g) - \beta_0
            \;=\; \frac{\E[\mathcal{H}\,\delta_g]}{\,J_0 + \E[\tilde{g}_X\,\delta_g]\,};
          \end{equation}
    \item[(ii)] $\beta_{IV}$ is Fr\'echet differentiable at $g_X$ with derivative $\delta_g \mapsto J_0^{-1}\,\E[\mathcal{H}\,\delta_g]$ and a second-order remainder of order $\|\delta_g\|^2$. Combining with Lemma~\ref{lem:hgeom}(iii), $\beta_{IV}(g_X + \delta_g) - \beta_0 = \E[\mathcal{H}\,(I - \Pi_{\mathcal{G}})\delta_g] / (J_0 + \E[\tilde{g}_X\,\delta_g])$, so the numerator and hence the Fr\'echet derivative depend only on the residual shape $(I - \Pi_{\mathcal{G}})\delta_g$. The drift also depends on the component parallel to $\tilde{g}_X$ through the denominator, and it is zero for every $\delta_g \in \mathcal{G}$ for which the perturbed signal remains relevant;
    \item[(iii)] $|\bar{D}_N| \le \|\mathcal{H}\|\,\max_k \rho_{g,k} = o_{\mathbb{P}}(N^{-1/4})$, and
          \[
            \beta_N^* \;=\; \beta_0 + J_0^{-1}\,\bar{D}_N + o_{\mathbb{P}}(N^{-1/2}),
            \qquad
            \beta_N^* \;=\; \sum_{k=1}^K w_k\,\beta_{IV}\big(\hat{g}_{X,k}\big)
            + o_{\mathbb{P}}(N^{-1/2}),
          \]
          the naive procedure targets the $B_k$-weighted SWATE of the learner's own signals;
    \item[(iv)] with $U := Y - l_Y - \beta_0(X - m_X)$, $V := g_X - m_X$, and
          $\Omega_0^{\mathrm{nv}} := \E[U^2 V^2]$, suppose $\Omega_0^{\mathrm{nv}} > 0$. Then
          \[
            \sqrt{N}\,\big(\hat{\beta}_{\mathrm{naive}} - \beta_N^*\big)
            \;\xrightarrow{d}\; \mathcal{N}\big(0,\; J_0^{-2}\,\Omega_0^{\mathrm{nv}}\big),
          \]
          The asymptotic variance is consistently estimable by plug-ins $\hat{J}_{\mathrm{nv}} \xrightarrow{p} J_0$ and $\hat{\Omega}^{\mathrm{nv}} \xrightarrow{p} \Omega_0^{\mathrm{nv}}$, so the nominal Wald interval satisfies $\mathbb{P}\big(\beta_N^* \in [\hat{\beta}_{\mathrm{naive}} \pm z_{1-\alpha/2}\, \hat{J}_{\mathrm{nv}}^{-1}\sqrt{\hat{\Omega}^{\mathrm{nv}}/N}]\big) \to 1 - \alpha$.
  \end{itemize}
\end{theorem}

The Fréchet derivative in~(ii) equals the G\^{a}teaux derivative of the naive moment divided by $J_0$, so the ``regularization bias'' is, to first order, an estimand drift. The naive plug-in tracks $\beta_{IV}(\hat{g}_X)$, the SWATE of the learner's own signal, which is legitimate and convex only if CCM holds for $\hat{g}_X$.

Although CCM for a learned signal is not directly testable because it conditions on the latent effect $\Delta$, a sufficient condition can be imposed on the learner. Conditional on the training sample $\mathcal D_{-k}$, suppose the realized signal $\hat g_{X,k}(\cdot,w)$ is coordinate-wise non-decreasing in the sign-normalized instruments. Under extended exogeneity, VM, and PA, the association argument of Proposition~\ref{prop:pmccm} then applies with $g_X$ replaced by $\hat g_{X,k}$ and yields CCM for the properly residualized signal $\hat g_{X,k}-\E[\hat g_{X,k}\mid W,\mathcal D_{-k}]$. Provided this signal remains relevant, $\beta_{IV}(\hat g_{X,k})$ is therefore a convex SWATE. This restriction is operational, as non-negative sums or averages of functions monotone in the same directions remain monotone, and standard implementations expose monotonicity constraints, including \texttt{monotone\_constraints} in XGBoost and LightGBM and \texttt{monotonic\_cst} in scikit-learn learners. Monotone-constrained first stages can thus guarantee convexity of each learner-specific target conditional on the maintained behavioral and design assumptions. Failure of fitted monotonicity does not refute CCM, since monotonicity is sufficient but not necessary.

Combining parts~(iii) and~(iv), the naive DML estimator has the first-order expansion
\[
  \hat{\beta}_{\mathrm{naive}} - \beta_0 = \frac{1}{N}\sum_{i=1}^N J_0^{-1}\,\psi_{\mathrm{naive}}(O_i;\beta_0,\eta_0) \;+\; J_0^{-1}\,\E\big[\mathcal{H}(Z,W)\,\Delta g_X(Z,W)\big] \;+\; o_{\mathbb{P}}\big(N^{-1/2}\big),
\]
whose regularization bias term $J_0^{-1}\E[\mathcal{H}\,\Delta g_X]$ is $o_{\mathbb{P}}(N^{-1/4})$ but not, in general, $o_{\mathbb{P}}(N^{-1/2})$. Hence $\sqrt{N}(\hat{\beta}_{\mathrm{naive}}-\beta_0)$ generally has no mean-zero limit and conventional intervals for $\beta_0$ are invalid.

\begin{corollary}[Coverage of the Fixed Target]\label{cor:cover}
  In the setting of Theorem~\ref{thm:driftvalid}, let $\sigma_{\mathrm{nv}}^2 := J_0^{-2}\Omega_0^{\mathrm{nv}}$, $\mu_N := \sqrt{N}\,J_0^{-1}\bar{D}_N / \sigma_{\mathrm{nv}}$, and let $\mathrm{CI}_N$ denote the nominal $(1-\alpha)$ naive interval of part~(iv). Then
  \begin{itemize}
    \item[(a)] if $\mu_N \xrightarrow{p} 0$:
          $\mathbb{P}(\beta_0 \in \mathrm{CI}_N) \to 1-\alpha$;
    \item[(b)] if $\mu_N \xrightarrow{p} \mu \neq 0$ finite:
          $\mathbb{P}(\beta_0 \in \mathrm{CI}_N) \to \Phi(z_{1-\alpha/2} - \mu) - \Phi(-z_{1-\alpha/2} - \mu) < 1-\alpha$;
    \item[(c)] if $|\mu_N| \xrightarrow{p} \infty$:
          $\mathbb{P}(\beta_0 \in \mathrm{CI}_N) \to 0$.
  \end{itemize}
  In particular, $\sqrt{N}(\hat{\beta}_{\mathrm{naive}} - \beta_0)$ is asymptotically centered and the naive interval is asymptotically valid for $\beta_0$ if and only if $\sqrt{N}\,\bar{D}_N \xrightarrow{p} 0$, or equivalently, the realized $\E[\mathcal{H}\,\Delta g_X]$ is $o_{\mathbb{P}}(N^{-1/2})$.
\end{corollary}

\begin{remark}\label{rem:shrink}
  A well-tuned first stage can yield approximately valid naive inference even when $\mathcal{H}\neq 0$. The reason is that the drift depends on the covariance between the realized first-stage error and $\mathcal{H}$. When prediction errors are small and weakly correlated with $\mathcal{H}$, the resulting drift can be small relative to the standard error, leaving confidence interval coverage close to its nominal level. Common rescaling of the signal and additive shifts depending only on $W$ illustrate this point, as they leave the population estimand unchanged (Lemma~\ref{lem:hgeom}(iv)).

  On the other hand, good predictive performance alone does not guarantee valid inference. Omitted nonlinearities or uneven fit across covariate strata can produce errors correlated with $\mathcal{H}$. Under the maintained conditions, asymptotically valid naive inference requires $\sqrt{N}\,\bar{D}_N\xrightarrow{p}0$. A non-rejection by the diagnostic in Section~\ref{sec:diagnostics} does not establish this condition, and selecting the inference procedure using a preliminary test can distort coverage \citep[see the related IV discussion in][]{andrews2019weak}. We therefore recommend CRC robust inference by default and use the diagnostic to assess sensitivity to first-stage estimation.
\end{remark}

\section{The Fixed Target: CRC Robust Score}\label{sec:fixedtarget}

We now build the score whose target is the fixed, learner-invariant $\beta_0$, which is an orthogonalized moment that eliminates the first-order influence of the learned signal, restores $\sqrt{N}$ inference, attains the efficiency bound, and there is no cost under homogeneity.

\subsection{The CRC Robust Orthogonal Score}\label{sec:crcscore}

Neyman orthogonality is restored by augmenting the numerator and denominator moments with correction terms that eliminate the first-order influence of every nuisance. Define
\begin{align}
  \psi_Y(O;\eta) & := (Y - l_Y)(g_X - m_X) + (X - g_X)(q_Y - l_Y), \label{eq:psiY} \\
  \psi_X(O;\eta) & := (X - m_X)(g_X - m_X) + (X - g_X)(g_X - m_X), \label{eq:psiX}
\end{align}
and the CRC robust score
\begin{equation}\label{eq:psicrc}
  \psi_{\mathrm{CRC}}(O;\beta,\eta) := \psi_Y(O;\eta) - \beta\,\psi_X(O;\eta).
\end{equation}

\begin{lemma}[Score Difference Identity]\label{lem:scoreid}
  For every $\beta \in \mathbb{R}$ and every admissible nuisance value $\eta = (l, m, q, g)$,
  \[
    \psi_{\mathrm{CRC}}(O;\beta,\eta) - \psi_{\mathrm{naive}}(O;\beta,\eta)
    \;=\; (X - g)\,\big[(q - l) - \beta\,(g - m)\big].
  \]
  At $(\beta_0, \eta_0)$ the correction equals $(X - g_X)\,\mathcal{H}(Z,W)$, and, for every $\beta$,
  \[
    \E\Big[(X - g_X)\big\{(q_Y - l_Y) - \beta\,(g_X - m_X)\big\}\Big] \;=\; 0 .
  \]
  Moreover $(X - g_X)\,\mathcal{H} = 0$ a.s.\ if and only if $\sigma_X^2(Z,W)\,\mathcal{H}^2 = 0$ a.s., where $\sigma_X^2(Z,W) := \Var(X \mid Z, W)$.
\end{lemma}

The orthogonalization thus re-injects the first-stage residual weighted by the unabsorbed heterogeneity. Three consequences follow from the identity. First, the correction is mean zero at every $\beta$, so both scores identify $\beta_0$. Second, when $\sigma_X^2 > 0$ it vanishes when the essential heterogeneity of Definition~\ref{def:esshet} fails, which explains why the naive score is orthogonal under homogeneity. Third, the correction is the influence contribution of the estimated signal that the naive sandwich variance omits (Proposition~\ref{prop:sech} in Appendix~\ref{sec:appenda}).

\subsection{Asymptotic Distribution and Efficiency}\label{sec:crcasymp}

Write $\beta_0 = \theta_Y / J_0$ with $\theta_Y = \E[(Y - l_Y)(g_X - m_X)]$ and $J_0 = \E[(g_X - m_X)^2]$. The CRC robust score restores $\sqrt{N}$ inference on the fixed target, attains the efficiency bound, and reduces to the naive score under homogeneity. Orthogonality, together with cross-fitting at the rates of Assumption~\ref{ass:ratesP}, reduces the impact of nuisance estimation to an asymptotically negligible second-order remainder.

\begin{theorem}[Asymptotic Distribution and Efficiency]\label{thm:fixedtarget}\label{thm:orth}\label{thm:an}\label{prop:eif-swate}\label{rem:nocost}
  Let Assumptions~\ref{ass:exog}, \ref{ass:relevance}, \ref{ass:regB}, and~\ref{ass:ratesP} hold, and let $\hat{\beta}_{\mathrm{CRC}}$ be the root of the $K$-fold cross-fitted score equation
  \[
    \frac{1}{N}\sum_{k=1}^{K}\sum_{i\in\mathcal{I}_k} \psi_{\mathrm{CRC}}\big(O_i;\hat{\beta},\hat{\eta}_k\big) = 0,
  \]
  where $\hat{\eta}_k$ is trained on the observations outside fold $\mathcal{I}_k$. Then
  \begin{itemize}
    \item[(i)] $\E[\psi_{\mathrm{CRC}}(O;\beta_0,\eta_0)] = 0$, and it is Neyman orthogonal with respect to the nuisance vector $\eta = (l_Y, m_X, q_Y, g_X)$ at $\eta_0$, i.e., $\partial_{\eta}\,\E[\psi_{\mathrm{CRC}}(O;\beta_0,\eta_0)][\eta - \eta_0] = 0$;
    \item[(ii)] $\sqrt{N}\,(\hat{\beta}_{\mathrm{CRC}} - \beta_0) \xrightarrow{d} \mathcal{N}\big(0,\; J_0^{-2}\,\Omega_0\big)$, where $\Omega_0 = \E\big[\psi_{\mathrm{CRC}}(O;\beta_0,\eta_0)^2\big]$ with the asymptotic variance consistently estimated by $\hat{J}^{-2}\hat{\Omega}$, where $\hat{J} = \tfrac1N\sum_i (\hat{g}_{X,i} - \hat{m}_{X,i})^2$ and $\hat{\Omega} = \tfrac1N\sum_i \psi_{\mathrm{CRC}}(O_i;\hat{\beta},\hat{\eta})^2$;
    \item[(iii)] the Efficient Influence Function (EIF) of $\beta_0$ is
          \[
            \phi_{\beta_0}(O) \;=\; J_0^{-1}\,\psi_{\mathrm{CRC}}(O;\beta_0,\eta_0)
            \;=\; J_0^{-1}\big(UV + e\,\mathcal{H}\big),
          \]
          with $U = Y - l_Y - \beta_0(X - m_X)$, $V = g_X - m_X$, $e = X - g_X$. Consequently the semiparametric efficiency bound for $\beta_0$ is $J_0^{-2}\,\Omega_0$, which is attained by $\hat{\beta}_{\mathrm{CRC}}$;
    \item[(iv)] if effects are homogeneous, $\psi_{\mathrm{CRC}}(\cdot;\beta_0,\eta_0) = \psi_{\mathrm{naive}}(\cdot;\beta_0,\eta_0) = UV$ and $\Omega_0 = \Omega_0^{\mathrm{nv}}$.
  \end{itemize}
\end{theorem}

\begin{lrbox}{\crcRbox}
  \begin{minipage}{0.92\linewidth}\singlespacing
    \begin{lstlisting}[style=crc,language=R]
## For R: nuisances lY=E[Y|W], mX=E[X|W], qY=E[Y|Z,W], gX=E[X|Z,W]  (cross-fitted)
psi_X <- (df$x-mX)*(gX-mX) + (df$x-gX)*(gX-mX)       # denominator moment
psi_Y <- (df$y-lY)*(gX-mX) + (df$x-gX)*(qY-lY)       # numerator moment
crc <- function(y,z,d,l_hat,m_hat,r_hat,g_hat,smpls) list(psi_a=-psi_X, psi_b=psi_Y)
\end{lstlisting}
  \end{minipage}
\end{lrbox}%
\begin{lrbox}{\crcPYbox}
  \begin{minipage}{0.92\linewidth}\singlespacing
    \begin{lstlisting}[style=crc,language=Python]
## For Python: nuisances l=lY=E[Y|W], m=gX=E[X|Z,W], r=mX=E[X|W], g=qY=E[Y|Z,W] (cross-fitted)
def crc(y, z, d, l, m, r, g, smpls):
    lY, gX, mX, qY = l, m, r, g
    psiX = (d-mX)*(gX-mX) + (d-gX)*(gX-mX)            # denominator moment
    psiY = (y-lY)*(gX-mX) + (d-gX)*(qY-lY)            # numerator moment
    return -psiX, psiY                                # (psi_a, psi_b); pass score=crc
\end{lstlisting}
  \end{minipage}
\end{lrbox}%

The estimation of $\hat{\beta}_{\mathrm{CRC}}$ is the standard DML problem for a user-defined linear score, which can be directly handled by the \texttt{DoubleML} libraries in R and Python.\footnote{Here we provide code for using CRC robust score in R and Python:\newline\usebox{\crcRbox}\newline\usebox{\crcPYbox}} The standard library routine returns the same estimator and standard error of Theorem~\ref{thm:an}.

\subsection{Worst-case bias over a realization ball}\label{sec:diagnosticsbias}

Two first-stage estimates can have the same $L_2$ error yet induce different biases in the naive moment, because their errors may covary differently with $\mathcal{H}$. Cross-fitting separates training from evaluation but does not rule out this covariance. Although the fitted functions are observed, their errors relative to the true nuisance functions are unknown. We therefore compare the largest absolute population moment bias of the naive and CRC robust scores over all admissible nuisance estimates satisfying the same error bounds. Lemma~\ref{lem:exact} provides expansions for this comparison.

\begin{proposition}[Worst-case first-order bias over a realization ball]\label{prop:worst}
  Let Assumptions~\ref{ass:exog} and~\ref{ass:relevance} and the heterogeneity condition $\E[\mathcal{H}^2] > 0$ hold, fix radii $r = (r_l, r_m, r_q, r_g) \in (0,1]^4$, fix $\bar{C} \ge \max_j\|\eta_{0,j}\|_\infty + \|\mathcal{H}\|_\infty / \|\mathcal{H}\|_{\mathbb{P},2}$, and define the realization set
  \[
    \mathcal{T}(r) := \Big\{ \eta \text{ admissible} :
    \|\eta_j\|_\infty \le \bar{C},\;\;
    \|\eta_j - \eta_{0,j}\|_{\mathbb{P},2} \le r_j,\; j \in \{l_Y,m_X,q_Y,g_X\}
    \Big\}.
  \]
  Then
  \begin{itemize}
    \item[(i)]
          $\displaystyle r_g\,\|\mathcal{H}\| \;\le\; \sup_{\eta \in \mathcal{T}(r)} \big|\E[\psi_{\mathrm{naive}}(O;\beta_0,\eta)]\big| \;\le\; r_g\,\|\mathcal{H}\| + (r_l + |\beta_0|\,r_m)(r_g + r_m), $ with the lower bound attained at $\eta = \big(l_Y,\, m_X,\, q_Y,\, g_X + r_g\,\mathcal{H}/\|\mathcal{H}\|\big)$;
    \item[(ii)]
          $\displaystyle \sup_{\eta \in \mathcal{T}(r)} \big|\E[\psi_{\mathrm{CRC}}(O;\beta_0,\eta)]\big| \;\le\; r_l r_m + r_g r_q + |\beta_0|\big(r_m^2 + r_g^2\big). $
  \end{itemize}
  Hence over any realization set of nonparametric radius ($\sqrt{N}\,r_g \to \infty$), the naive moment's worst-case first-order bias is of the order $r_g$, while that of the CRC robust moment is of the order of squared radii. Uniform validity over $\mathcal{T}(r_N)$ holds for the CRC robust score under Assumption~\ref{ass:ratesP} and fails for the naive score whenever $\mathcal{H} \neq 0$.
\end{proposition}

Proposition~\ref{prop:worst} is more than a bound. In Corollary~\ref{cor:adv}, Appendix~\ref{sec:appenda}, we show that an explicit deterministic sequence $\hat{g}_X^{(N)} = g_X + N^{-\gamma}\mathcal{H}$ with $\gamma\in(1/4,1/2)$ drives naive coverage of $\beta_0$ to zero while the CRC robust interval stays nominal.

\section{Diagnostics and inference under a weak residual signal}\label{sec:diagnostics}

There are two remaining tasks for practice. First, telling whether the naive estimate is on target and inferring $\beta_0$ when partialling out $W$ leaves little residual signal. This section supplies an alignment diagnostic and gives an identification robust confidence set that stays valid when the instruments are weak.

\subsection{The Alignment Diagnostic}\label{sec:aligndiag}

Corollary~\ref{cor:cover} makes the practitioner's question precise: is the realized drift negligible, or is the naive estimator answering the question about $\beta_0$? Lemma~\ref{lem:scoreid} supplies a test. With $\hat{\mathcal{H}}_k(Z,W) := (\hat{q}_{Y,k} - \hat{l}_{Y,k}) - \hat{\beta}_{\mathrm{CRC}}\,(\hat{g}_{X,k} - \hat{m}_{X,k})$, define
\[
  \widehat{\mathcal{A}}_N := \frac{1}{N}\sum_{k=1}^{K}\sum_{i\in\mathcal{I}_k}
  \big(X_i - \hat{g}_{X,k,i}\big)\,\hat{\mathcal{H}}_{k,i},
  \quad
  \widehat{\Xi}_N := \frac{1}{N}\sum_{k=1}^{K}\sum_{i\in\mathcal{I}_k}
  \big(X_i - \hat{g}_{X,k,i}\big)^2\,\hat{\mathcal{H}}_{k,i}^2,
  \quad
  T_N := \frac{\sqrt{N}\,\widehat{\mathcal{A}}_N}{\widehat{\Xi}_N^{1/2}} .
\]

\begin{proposition}[Alignment Diagnostic]\label{prop:align}
  Let Assumptions~\ref{ass:exog}, \ref{ass:relevance}, \ref{ass:regB}, and~\ref{ass:ratesP} hold, and suppose $\Xi_0 := \E\big[(X - g_X)^2\,\mathcal{H}^2\big] > 0$. Then
  \begin{itemize}
    \item[(i)] $\displaystyle
            \sqrt{N}\,\widehat{\mathcal{A}}_N = \frac{1}{\sqrt{N}}\sum_{i=1}^{N}(X_i - g_{X,i})\,\mathcal{H}_i \;-\; \sqrt{N}\,\bar{D}_N \;+\; o_{\mathbb{P}}(1)$, and $\widehat{\Xi}_N \xrightarrow{p} \Xi_0$;
    \item[(ii)] under
          $H_0 : \sqrt{N}\,\bar{D}_N \xrightarrow{p} 0$ and $T_N \xrightarrow{d} \mathcal{N}(0,1)$, so the test rejecting when $|T_N| > z_{1-\alpha/2}$ has asymptotic size $\alpha$;
    \item[(iii)] if $\sqrt{N}\,|\bar{D}_N| \xrightarrow{p} \infty$, then
          $|T_N| \xrightarrow{p} \infty$, and the test is consistent;
    \item[(iv)] $\sqrt{N}\,\widehat{\mathcal{A}}_N = J_0\,\sqrt{N}\,\big(\hat{\beta}_{\mathrm{CRC}} - \hat{\beta}_{\mathrm{naive}}\big) + o_{\mathbb{P}}(1)$, so $T_N$ is asymptotically the studentized contrast of the two estimators.
  \end{itemize}
\end{proposition}

A rejection provides evidence of non-negligible drift in the naive estimand due to first-stage error correlated with $\mathcal{H}$. Under the maintained assumptions, the CRC robust estimator continues to provide valid inference for $\beta_0$. The diagnostic thus uses a Hausman-type comparison of estimators with different robustness to first-stage error \citep{hausman1978specification}. The diagnostic tests negligible drift, rather than $\mathcal{H}=0$, and therefore differs from tests of the CRC model such as \citet{heckman2010testing}. We recommend reporting $(T_N,\widehat{\Xi}_N)$ to interpret the estimates, without using the diagnostic to select between inference procedures (Section~\ref{sec:geometry}).

\begin{remark}\label{rem:degen}
  Under the maintained conditions, $\mathcal{H}=0$ implies $\Xi_0=0$, $\sqrt{N}\,\widehat{\mathcal{A}}_N=o_{\mathbb{P}}(1)$, and $\hat{\beta}_{\mathrm{naive}}-\hat{\beta}_{\mathrm{CRC}}=o_{\mathbb{P}}(N^{-1/2})$. The two estimators are then first-order equivalent, but the standard-normal calibration of $T_N$, which requires $\Xi_0>0$, does not apply. More generally, $\Xi_0=\E[\Var(X\mid Z,W)\mathcal{H}^2]$, so a small $\widehat{\Xi}_N$ may reflect little residual treatment variation as well as limited heterogeneity along the signal.
\end{remark}

Agreement between the point estimates can also arise when $\mathcal{H}\neq0$. Under Assumption~\ref{ass:fsal} with $a=\mathcal{H}$, the first stage supplies the correction implicitly, making the naive and CRC robust estimators first-order equivalent (Proposition~\ref{prop:sech}(iii)). When $\Xi_0>0$, this gives $T_N\xrightarrow{p}0$ because the diagnostic's leading stochastic terms cancel. The naive sandwich nevertheless converges to $J_0^{-2}\Omega_0^{\mathrm{nv}}$, which need not equal the true variance $J_0^{-2}\Omega_0$. A least-squares first stage whose dictionary spans $g_X$ and $\mathcal{H}$ provides a concrete example (Proposition~\ref{prop:series}), consistent with the implicit correction studied by \citet{newey1994asymptotic}. Thus a small diagnostic does not validate the naive standard error, a concern also in \citet{kolesar2013estimation}, \citet{evdokimov2018inference}, \citet{lee2018consistent}.

\subsection{Identification Robust Sets}\label{sec:robustsets}

Because $\psi_{\mathrm{CRC}}$ is affine in $\beta$, inference can avoid the Jacobian. This is valuable when partialling $W$ out of a strongly $W$-predictable instrument leaves little residual signal ($J_0 \approx 0$), the regime in which the Wald interval of Theorem~\ref{thm:an} degrades. Define
\begin{equation}\label{eq:arstat}
  \bar{\psi}(\beta) := \frac{1}{N}\sum_{k=1}^K\sum_{i\in\mathcal{I}_k}
  \psi_{\mathrm{CRC}}\big(O_i;\beta,\hat{\eta}_k\big),
  \qquad
  \hat{\sigma}^2(\beta) := \frac{1}{N}\sum_{k=1}^K\sum_{i\in\mathcal{I}_k}
  \psi_{\mathrm{CRC}}\big(O_i;\beta,\hat{\eta}_k\big)^2,
\end{equation}
$\mathrm{AR}_N(\beta) := \sqrt{N}\,\bar{\psi}(\beta)/\hat{\sigma}(\beta)$, and the confidence set $\mathcal{C}_{1-\alpha} := \{\beta \in \mathbb{R} : |\mathrm{AR}_N(\beta)| \le z_{1-\alpha/2}\}$.

\begin{proposition}[Validity with Weak Signal]\label{thm:ar}
  Let Assumptions~\ref{ass:exog}, \ref{ass:relevance}, \ref{ass:regB}, and~\ref{ass:ratesP} hold, and suppose $\Omega_0 = \E\big[\psi_{\mathrm{CRC}}(O;\beta_0,\eta_0)^2\big] > 0$. Then
  \[
    \mathbb{P}\big(\beta_0 \in \mathcal{C}_{1-\alpha}\big)
    \;\longrightarrow\; 1 - \alpha .
  \]
\end{proposition}

\begin{remark}\label{rem:fieller}
  Writing $\psi_{Y,i}, \psi_{X,i}$ for the plug-in values of the moments~\eqref{eq:psiY}--\eqref{eq:psiX}, let $S_Y := \frac{1}{N}\sum_i \psi_{Y,i}$, $S_X := \frac{1}{N}\sum_i \psi_{X,i}$, $M_{YY} := \frac{1}{N}\sum_i \psi_{Y,i}^2$, $M_{XY} := \frac{1}{N}\sum_i \psi_{Y,i}\psi_{X,i}$, $M_{XX} := \frac{1}{N}\sum_i \psi_{X,i}^2$. Since $\bar{\psi}(\beta) = S_Y - \beta S_X$ and $\hat{\sigma}^2(\beta) = M_{YY} - 2\beta M_{XY} + \beta^2 M_{XX}$, the set is the solution of a quadratic inequality,
  \[
    \mathcal{C}_{1-\alpha}
    = \big\{\beta : A\beta^2 - 2B\beta + C \le 0\big\},
  \]
  with
  \[
    A = N S_X^2 - z^2 M_{XX},\quad
    B = N S_X S_Y - z^2 M_{XY},\quad
    C = N S_Y^2 - z^2 M_{YY},
  \]
  with $z = z_{1-\alpha/2}$. The set is an interval when $A > 0$, and the complement of an interval or the whole line when $A \le 0$ \citep{anderson1949estimation,fieller1954some}. Under strong identification, $J_0$ is bounded away from zero, $\mathcal{C}_{1-\alpha}$ is asymptotically equivalent to the Wald interval of Theorem~\ref{thm:an}.
\end{remark}

\section{Simulation}\label{sec:montecarlo}

This section examines the finite-sample performance of the naive and CRC robust estimators and of the accompanying diagnostics. The experiments have three aims: to show that the drift of Theorem~\ref{thm:driftvalid} distorts naive inference at conventional sample sizes; to assess the coverage and efficiency of the CRC robust estimator, including its behavior under homogeneity; and to show the finite-sample size and power of the diagnostics of Section~\ref{sec:diagnostics} in designs where the target is known.

The data generating process is a threshold-crossing CRC model, in which a binary treatment is selected on its own gain
\begin{equation}\label{eq:mcdgp}
  \begin{aligned}
    \nu    & = \alpha_0 + Z^\top\alpha_Z + W^\top\alpha_W,         & X    & = \mathbf{1}\{U \le \nu\},            \\
    \Delta & = a_0 + a_W\,(W^\top\gamma) + b\,U + \sigma_\xi\,\xi, & Y(0) & = c_0 + W^\top\gamma_Y + \sigma_e\,e,
  \end{aligned}
\end{equation}
with $Y = Y(0) + X\Delta$ as in equation~\eqref{eq:crc}. The shocks $(U,\xi,e)$ are standard normal, $Z\sim\mathcal{N}(0,I_{p_Z})$, $W\sim\mathcal{N}(0,I_{p_W})$, and all five are mutually independent, so Assumption~\ref{ass:exog} holds by construction. Because the selection shock $U$ enters both the treatment equation and the return, the coefficient $b$ measures essential heterogeneity in the sense of Definition~\ref{def:esshet}, with $b=0$ the homogeneous benchmark, and $a_W$ measures heterogeneity in observables through $\E[\Delta\mid W]$. All four nuisance functions are available in closed form. In particular, $g_X=\E[X\mid Z,W]=\Phi(\nu)$.

We consider two main configurations. The first, the oracle design, has $p_Z=50$ instruments whose first-stage coefficient vector $\alpha_Z$ is approximately sparse, $p_W=10$ covariates, strong essential heterogeneity ($b=2.4$), and no heterogeneity in observables ($a_W=0$). Its nuisances are evaluated at their closed-form values rather than estimated, which separates the behavior of the estimators from first-stage estimation error. A variant perturbs the signal by an error of known size and direction. The second, the learnable design, is lower-dimensional and less noisy ($p_Z=8$, $p_W=4$, $a_W=1.2$, $b=0.6$), and its nuisances are cross-fitted with the lasso and gradient boosting. The configuration is chosen so that the product-rate condition of Assumption~\ref{ass:ratesP} is plausible for both learners. The two configurations imply the same target, $\beta_0\approx0.50$. For each sample size $N\in\{500,1000,2000,4000,8000\}$ we generate $R=1000$ samples. Throughout, coverage refers to the frequency with which a nominal $95\%$ confidence interval covers the fixed target $\beta_0$, and the interval is the Wald interval unless stated otherwise. Two auxiliary designs, constructed to evaluate the diagnostics of Section~\ref{sec:diagnostics}, are described together with the corresponding results below.

Figure~\ref{fig:collapse} reports the behavior of naive inference when the first-stage error is aligned with the heterogeneity residual. Starting from the oracle design, we replace the signal with $\hat{g}_X = g_X + c_N\,\mathcal{H}$ under two schemes for $c_N$. The fast error, $c_N = 2\,N^{-1/3}$, is of the form in Corollary~\ref{cor:adv} and drives $\mu_N$ in Corollary~\ref{cor:cover} to infinity. The slow error, $c_N = 2\sqrt{\log p_Z/N}$, holds $\mu_N$ at a nonzero constant. The two schemes realize branches (c) and (b) of Corollary~\ref{cor:cover}, respectively. Under the fast error, coverage of the naive Wald interval falls from $0.847$ at $N=500$ to $0.681$ at $N=8000$, the scaled bias $\sqrt{N}(\hat\beta_{\mathrm{naive}}-\beta_0)$ grows from $4.69$ to $7.08$, and the alignment test of Section~\ref{sec:aligndiag} rejects in every replication. This pattern is consistent with the drift of Theorem~\ref{thm:driftvalid}. Under the slow error, coverage is near $0.90$ and does not improve as $N$ grows. The CRC robust interval has coverage between $0.918$ and $0.950$ in both cases (Theorem~\ref{thm:orth}).

\begin{figure}[!htbp]
  \centering
  \caption{Coverage of the fixed target under a contaminated signal}
  \label{fig:collapse}
  \begin{minipage}{\linewidth}
    \centering
    \includegraphics[width=0.82\textwidth]{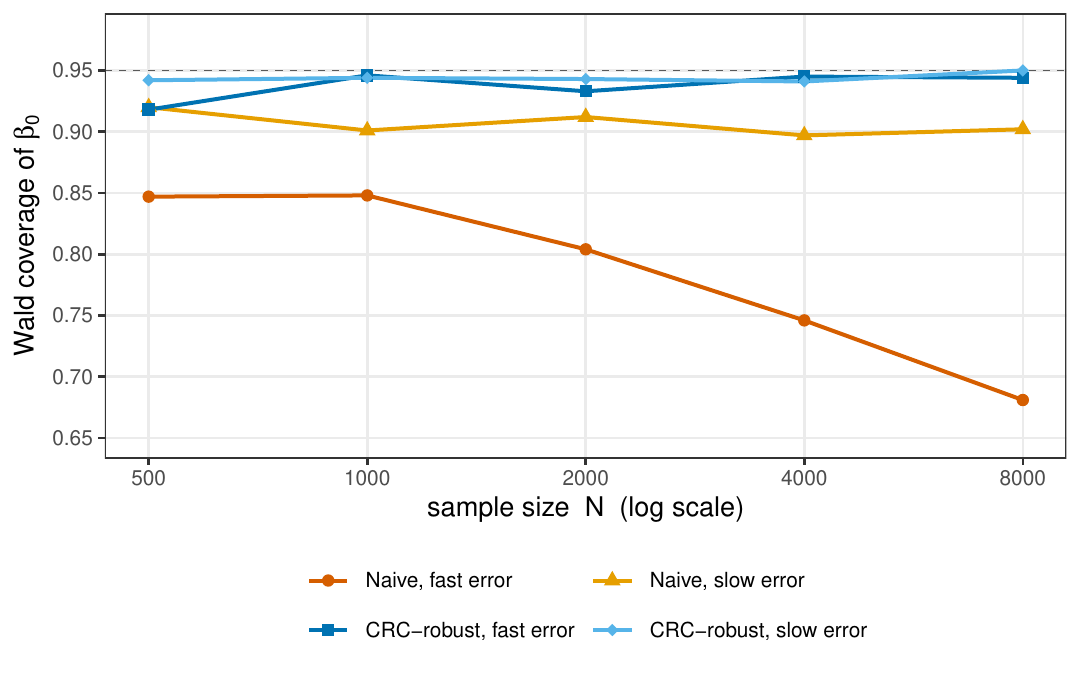}
    \par\vspace{0.4em}
    \footnotesize\setstretch{1}\raggedright
    \textit{Notes:} The figure plots the empirical coverage of nominal $95\%$ Wald intervals for the fixed target $\beta_0$ against the sample size $N$ (log scale) in the oracle design with $b=2.4$. The signal is contaminated along the heterogeneity residual, $\hat{g}_X=g_X+c_N\,\mathcal{H}$, with a fast error ($c_N=2\,N^{-1/3}$) and a slow error ($c_N=2\sqrt{\log p_Z/N}$). The naive interval realizes the two under-covering points of Corollary~\ref{cor:cover}. The CRC robust interval remains near the nominal level under both errors (Theorem~\ref{thm:orth}).
    \par
  \end{minipage}
\end{figure}

% Keep the coverage figure ahead of the simulation tables.
\FloatBarrier

\begin{table}[t]
  \centering
  \caption{CRC robust inference: no cost under homogeneity, efficiency under heterogeneity}
  \label{tab:mc_fix}
  \begin{minipage}{\linewidth}
    \centering
    \begin{tabular}{@{}ll ccccc@{}}
      \toprule
                      &                                                 & \multicolumn{5}{c}{Sample size $N$}                                         \\
      \cmidrule(lr){3-7}
      Design          & Statistic                                       & $500$                               & $1000$  & $2000$  & $4000$  & $8000$  \\
      \midrule
      \multirow{2}{*}{Homogeneous, $b=0$}
                      & Cov $\hat\beta_{\mathrm{naive}}$                & $0.947$                             & $0.954$ & $0.957$ & $0.949$ & $0.940$ \\
                      & Cov $\hat\beta_{\mathrm{CRC}}$                  & $0.944$                             & $0.957$ & $0.958$ & $0.949$ & $0.939$ \\
      \midrule
      Oracle, $b=2.4$ & $N\widehat{\mathrm{Var}}$ $(\to 20.63)$         & $21.17$                             & $20.67$ & $19.30$ & $19.46$ & $20.75$ \\
      \midrule
      \multirow{4}{*}{\shortstack[l]{Learnable                                                                                                        \\(cross-fitted)}}
                      & Lasso $N\widehat{\mathrm{Var}}$ $(\to 2.86)$    & $3.16$                              & $2.63$  & $3.00$  & $2.64$  & $2.75$  \\
                      & Boosting $N\widehat{\mathrm{Var}}$ $(\to 2.86)$ & $5.98$                              & $4.91$  & $4.16$  & $3.53$  & $3.74$  \\
                      & Lasso Cov $\beta_0$                             & $0.941$                             & $0.956$ & $0.949$ & $0.959$ & $0.952$ \\
                      & Boosting Cov $\beta_0$                          & $0.867$                             & $0.905$ & $0.925$ & $0.949$ & $0.947$ \\
      \bottomrule
    \end{tabular}
    \par\vspace{0.4em}
    \footnotesize\setstretch{1}\raggedright
    \textit{Notes:} Coverage (Cov) is the frequency of a nominal $95\%$ Wald interval covering the fixed target $\beta_0\approx0.50$. $N\widehat{\mathrm{Var}}$ denotes $N \times$ the mean squared error of $\hat\beta_{\mathrm{CRC}}$ about $\beta_0$ across replications. $b$ indexes the strength of essential heterogeneity. In the oracle design the nuisances are evaluated at their closed-form values; in the learnable design they are cross-fitted with the lasso and with gradient boosting.
    \par
  \end{minipage}
\end{table}

\begin{table}[t]
  \centering
  \caption{Identification robust inference under a weak residual signal}
  \label{tab:mc_diag}
  \begin{minipage}{\linewidth}
    \centering
    \begin{tabular}{@{}lccc@{}}
      \toprule
      $J_0$                         & Cov$_{\mathrm{AR/Fieller}}$ & Cov$_{\mathrm{Wald}}$ & Frac.\ unbounded \\
      \midrule
      $0.075$~{\scriptsize(strong)} & $0.944$                     & $0.949$               & $0.00$           \\
      $0.0013$~{\scriptsize(weak)}  & $0.947$                     & $1.000$               & $0.62$           \\
      \bottomrule
    \end{tabular}
    \par\vspace{0.4em}
    \footnotesize\setstretch{1}\raggedright
    \textit{Notes:} Scalar-instrument threshold model in which the instrument loading is scaled down so that the residual-signal strength $J_0$ falls. Cov$_{\mathrm{AR/Fieller}}$ is the coverage of the identification robust set of Proposition~\ref{thm:ar}, Cov$_{\mathrm{Wald}}$ the coverage of the plug-in Wald interval, and Frac.\ unbounded the fraction of replications in which the identification robust set is unbounded (Remark~\ref{rem:fieller}). Based on $R=1000$ replications.
    \par
  \end{minipage}
\end{table}

Table~\ref{tab:mc_fix} shows the result of the CRC robust estimator. In the first panel effects are homogeneous, so the naive and CRC robust scores coincide and the correction has no efficiency loss, and both coverage rates are close to the nominal rate. The second panel reports $N$ times the mean squared error of $\hat\beta_{\mathrm{CRC}}$ about $\beta_0$, denoted $N\widehat{\mathrm{Var}}$, in the oracle design. The semiparametric bound $J_0^{-2}\Omega_0$ equals $20.63$, and the simulated values lie between $19.30$ and $21.17$ across the sample-size grid, in line with the efficiency result of Theorem~\ref{thm:orth}(iii). The third panel repeats the exercise in the learnable design, where the bound equals $2.86$ and the nuisances are cross-fitted. With lasso nuisances the scaled variance is near the bound from moderate $N$ onward, between $2.63$ and $3.16$ across the grid, and coverage is close to nominal at every sample size, between $0.941$ and $0.959$. Because the design has a low-dimensional sparse linear index, the lasso recovers the nuisances at a near-parametric rate, so the second-order remainder that orthogonality leaves is already negligible at $N=500$ and the Wald interval is both centered and correctly scaled. With gradient boosting the scaled variance approaches the bound from above, from $5.98$ to $3.74$, and coverage rises from $0.867$ at $N=500$ to $0.947$ at $N=8000$.

Table~\ref{tab:mc_diag} shows the inference when partialling out $W$ leaves little residual signal. We scale down $J_0$ from $0.075$ to $0.0013$. The identification robust set of Section~\ref{sec:robustsets} has coverage $0.944$ and $0.947$ at the two values of $J_0$, whereas the plug-in Wald interval has coverage $0.949$ under the strong signal and $1.000$ under the weak one. In the weak-signal design the AR set is unbounded in $62\%$ of replications, which is the expected behavior of an identification robust set when $J_0$ is near zero (Proposition~\ref{thm:ar}, Remark~\ref{rem:fieller}).

\section{Application}\label{sec:application}

We revisit the quarter-of-birth (QOB) design for the returns to schooling, the natural experiment of \citet{angrist1991does} and the setting studied by \citet{angrist2022machine}, where the first stage is estimated by machine learning. The instruments are numerous and individually weak. The first stage is the high-dimensional problem of pooling many weak instruments into a single signal, and the returns to schooling are an example of essential heterogeneity, since individuals plausibly select schooling on their own unobserved return. The design is therefore a test of whether the estimand based on a machine-learned signal drifts with the first-stage learner.

\subsection{Data and design}

We use the \citet{angrist1991does} data from the 1980 U.S.\ Census\footnote{The data is downloaded from the Angrist Data Archive, https://economics.mit.edu/people/faculty/josh-angrist/angrist-data-archive}, restricted to men born 1930--1939, giving $N=329{,}509$. The outcome $Y$ is log weekly wage and the treatment $X$ is years of completed schooling. Following \citet{angrist2022machine}, the excluded instruments $Z$ are the three quarter-of-birth indicators fully interacted with year and state of birth, in total $186$ raw instruments. The covariates $W$ partialled out are a saturated set of year-of-birth and state-of-birth indicators together with race, marital status, and SMSA residence. The covariate projections $l_Y = \E[Y\mid W]$ and $m_X = \E[X\mid W]$ are estimated by saturated least squares on the discrete $W$ values, which is nonparametric and consistent because $W$ takes finitely many values. The signal $g_X = \E[X\mid Z,W]$ and the outcome projection $q_Y = \E[Y\mid Z,W]$ are the two functions estimated by the machine learner. All four nuisances are $5$-fold cross-fitted, and the procedure is repeated over $11$ independent random partitions of the sample; we report the median of the $11$ estimates with the median-adjusted variance, following the median aggregation of \citet{chernozhukov2018double}.

We estimate the signal with four first stages: an ordinary least-squares projection on the full instrument set (the ``optimal instrument'' 2SLS of \citet{chen2021mostly}), the lasso, random forest, and gradient boosting. For each we compute $\hat\beta_{\mathrm{naive}}$ and $\hat\beta_{\mathrm{CRC}}$, the alignment diagnostic pair $(T_N,\widehat\Xi_N)$, and the identification robust set $\mathcal{C}_{0.95}$.

\subsection{Results}

Table~\ref{tab:application} and Figure~\ref{fig:application} report the results. The naive $\hat\beta_{\mathrm{naive}}$ estimates a learner-dependent moving SWATE $\beta_{IV}(\hat g_X)$, while $\hat\beta_{\mathrm{CRC}}$ estimates the fixed, learner-invariant $\beta_0$, so the gap between the two estimates measures the drift of the naive estimand learner by learner. The alignment diagnostic $T_N$ tests, for each learner, whether the naive estimate targets $\beta_0$ or a drifted alternative.

The diagnostic behaves as Proposition~\ref{prop:align}(iv) predicts, which shows $T_N$ to be the studentized contrast between the naive and robust estimates. It rejects alignment for gradient boosting ($T_N=10.6$), whose naive and robust estimates differ by $0.044$, and does not reject for the lasso ($T_N=0.0$), whose two estimates coincide near $0.11$. The random forest lies just inside the band ($T_N=1.6$). The heterogeneity measure $\widehat\Xi_N$ is largest for the random forest, so heterogeneity along the signal is most detectable for the tree learners.

The greater cross-learner variation in the CRC point estimates reflects second-order error in the Jacobian, amplified by weak residual signal. To isolate this channel, condition on the training sample $\mathcal D_{-k}$ and set the outcome nuisances $(l_Y,q_Y)$ equal to their population values. Lemma~\ref{lem:exact} gives the fold-specific population root
\[
  \frac{\theta_Y}
  {J_0-\big\{\|\Delta g_{X,k}\|^2-\|\Delta m_{X,k}\|^2\big\}}
  =
  \beta_0
  \left[
    1-\frac{\|\Delta g_{X,k}\|^2-\|\Delta m_{X,k}\|^2}{J_0}
    \right]^{-1},
\]
provided the denominator is nonzero. Orthogonality eliminates first-order drift, but when $J_0$ is small, modest learner-specific differences in this second-order remainder can induce substantial cross-learner variation.

The same channel need not cause comparable variation in the naive estimand. By Theorem~\ref{thm:driftvalid}(ii), the numerator of its drift depends on $\Delta g_{X,k}$ only through the residual shape component $(I-\Pi_{\mathcal G})\Delta g_{X,k}$. Rescalings and $W$-measurable shifts belong to $\mathcal G$ and leave the target unchanged, provided the perturbed signal remains relevant. Learners that differ primarily along these invariant directions may therefore yield similar naive estimates.

The OLS-projection case is especially instructive. Its naive estimate, $0.083\,(0.024)$, coincides with the classic \citet{angrist1991does} two-stage least squares estimate. But partialling the saturated $W$ out of the $186$ weak and mechanically collinear QOB interactions leaves little residual signal with $\widehat J_0 = 0.012$, and the identification robust set is the whole real line. This is the weak residual signal regime of Proposition~\ref{thm:ar} and Remark~\ref{rem:fieller}, and it is the semiparametric counterpart of the \citet{bound1995problems} critique of the QOB design: the unbounded $\mathcal{C}_{0.95}$ reports that $\beta_0$ is not identified by unregularized pooling of these instruments, which the naive Wald interval does not reveal. The alignment statistic indicates the same degeneracy, with its large value $T_N=-4.7$ reflecting the distance between the naive and robust estimates ($0.083$ against $-0.008$).

\begin{table}[!htbp]
  \centering
  \caption{Returns to schooling: naive and CRC robust estimates}
  \label{tab:application}
  \begin{minipage}{\linewidth}
    \centering
    \setlength{\tabcolsep}{5.5pt}
    \begin{tabular}{lcccccc}
      \toprule
      First stage       & $\hat\beta_{\mathrm{naive}}$ & $\hat\beta_{\mathrm{CRC}}$   & $T_N$             & $\widehat\Xi_N$ & $\widehat J_0$ & $\mathcal{C}_{0.95}$ \\
      \midrule
      OLS projection    & $0.083$ $(0.024)$            & $-0.008$ $(0.044)$           & $-4.7$            & $0.003$         & $0.012$        & $(-\infty,\infty)$   \\
      Lasso             & $0.113$ $(0.123)$            & $\phantom{-}0.108$ $(0.012)$ & $\phantom{-}0.0$  & $0.004$         & $0.017$        & $[0.084,\,0.134]$    \\
      Random forest     & $0.114$ $(0.014)$            & $\phantom{-}0.132$ $(0.004)$ & $\phantom{-}1.6$  & $0.046$         & $0.143$        & $[0.118,\,0.147]$    \\
      Gradient boosting & $0.114$ $(0.006)$            & $\phantom{-}0.158$ $(0.007)$ & $\phantom{-}10.6$ & $0.031$         & $0.097$        & $[0.137,\,0.183]$    \\
      \bottomrule
    \end{tabular}
    \par\vspace{0.4em}
    \footnotesize\setstretch{1}\raggedright
    \textit{Notes:} The table compares naive and CRC robust estimates across first-stage learners in the quarter-of-birth design. Point estimates are medians over $11$ cross-fitting splits, with median-adjusted standard errors in parentheses. $T_N$ and $\widehat\Xi_N$ are the alignment diagnostic of Proposition~\ref{prop:align}, $\widehat J_0$ is the cross-fitted sample estimate of the residual-signal strength $J_0=\E[(g_X-m_X)^2]$, and $\mathcal{C}_{0.95}$ is the AR identification robust set of Proposition~\ref{thm:ar}.
    \par
  \end{minipage}
\end{table}

\begin{figure}[!htbp]
  \centering
  \caption{Naive and CRC robust estimates and the alignment diagnostic $T_N$}
  \label{fig:application}
  \begin{minipage}{\linewidth}
    \centering
    \includegraphics[width=\textwidth]{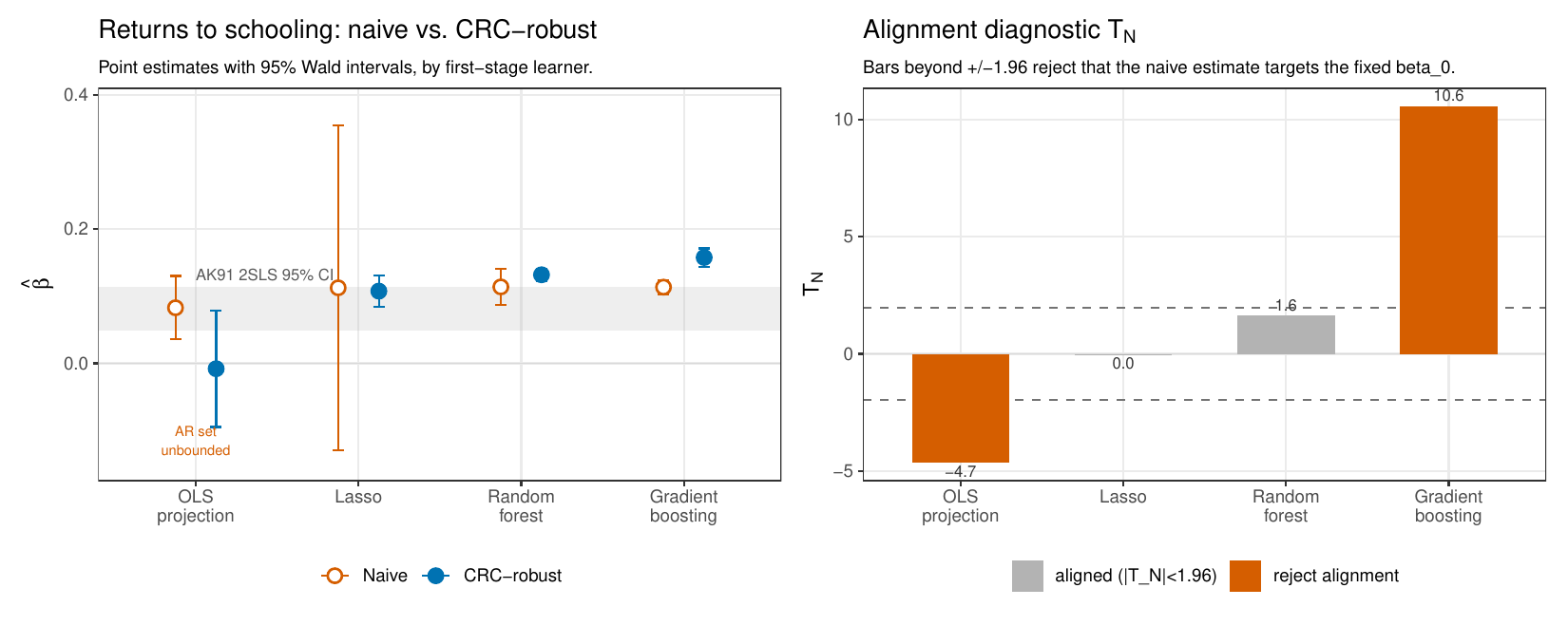}
    \par\vspace{0.4em}
    \footnotesize\setstretch{1}\raggedright
    \textit{Notes:} Left: naive $\hat\beta_{\mathrm{naive}}$ (hollow circles) and CRC robust $\hat\beta_{\mathrm{CRC}}$ (filled) with $95\%$ Wald intervals, quarter-of-birth returns to schooling for four first-stage learners. The shaded band is the classic \citet{angrist1991does} 2SLS $95\%$ confidence interval. Right: the alignment diagnostic $T_N$. Bars outside the $\pm 1.96$ band reject the hypothesis that the naive estimate targets.
    \par
  \end{minipage}
\end{figure}

Regularizing the first stage restores a bounded target. The lasso, random forest, and gradient boosting all yield bounded sets, though only the two tree learners deliver a strong residual signal. The lasso is an intermediate case. Its residual signal $\widehat J_0 = 0.017$ is barely above the OLS projection's, and its naive Wald interval is too wide to be informative, yet the CRC robust score returns a tight estimate, $0.108\,(0.012)$, inside the bounded set $[0.084,\,0.134]$. The orthogonal score and its identification robust set thus recover an interpretable, convex signal-weighted average of returns to schooling from a first stage whose naive interval would be uninformative.

The three reports give the empirical recommendation. We report $\hat\beta_{\mathrm{CRC}}$ as the estimate, a signal-weighted average of the heterogeneous returns to schooling that is convex under covariance monotonicity by Theorem~\ref{thm:representation}. We read $(T_N,\widehat\Xi_N)$ as interpretation rather than as a rule for selecting a preferred learner. A rejecting $T_N$ is an informative warning, but a non-rejecting $T_N$ does not certify naive validity, which cannot be verified without estimating $\mathcal{H}$. Where the residual signal is weak, we read the unbounded $\mathcal{C}_{0.95}$ as a statement about the interpretability and validity of the estimand.

% Keep all application floats before the conclusion.
\FloatBarrier
\section{Conclusion}\label{sec:conclusion}

In this paper, we give the IV estimand based on a machine-learned signal a structural meaning and a valid inference theory under correlated random coefficients. The partialled-out estimand built from any signal of the instruments is a signal-weighted average of the heterogeneous effects, convex under a covariance-monotonicity condition, which we microfound from vector monotonicity and positive association. When the signal is learned, the standard debiased moment is not Neyman orthogonal for the fixed target, and its first-order bias is a drift of the estimand. Therefore, naive DML validly answers a moving, learner-dependent question while inference on the fixed, learner-invariant target requires the heterogeneity robust orthogonal score we construct, which is the efficient influence function for the target.

Our recommended estimator is $\hat{\beta}_{\mathrm{CRC}}$, which stays valid where the naive one is inconsistent even when the point estimate is centered. Alongside it we report the pair $(T_N,\widehat{\Xi}_N)$, read as interpretation and never as a selector, since pre-testing distorts coverage. When the residual signal is weak, we add the identification robust set $\mathcal{C}_{1-\alpha}$, whose unbounded realizations indicate that the estimand's interpretability is questionable. Together, the estimator and its diagnostics allow a flexible machine-learned first stage to be used for instrumental variables estimation while the target remains a fixed, interpretable parameter on which valid inference is available.

%%%%%%%%%%%%%%%%%%%%%%%%%%%%%%%%%%%%%%%%%%%%%%%%%
\clearpage
\begin{singlespace}
  \bibliographystyle{ecta}
  \bibliography{references.bib}

@article{chernozhukov2018double,
  title     = {{Double/Debiased Machine Learning for Treatment and Structural Parameters}},
  author    = {Chernozhukov, Victor and Chetverikov, Denis and Demirer, Mert and Duflo, Esther and Hansen, Christian and Newey, Whitney and Robins, James},
  journal   = {The Econometrics Journal},
  volume    = {21},
  number    = {1},
  pages     = {C1--C68},
  year      = {2018},
  publisher = {Oxford University Press}
}

@article{wooldridge1997two,
  title     = {On two stage least squares estimation of the average treatment effect in a random coefficient model},
  author    = {Wooldridge, Jeffrey M.},
  journal   = {Economics Letters},
  volume    = {56},
  number    = {2},
  pages     = {129--133},
  year      = {1997},
  publisher = {Elsevier}
}

@article{heckman1998instrumental,
  title     = {Instrumental variables methods for the correlated random coefficient model: Estimating the average rate of return to schooling when the return is correlated with schooling},
  author    = {Heckman, James J. and Vytlacil, Edward J.},
  journal   = {Journal of Human Resources},
  volume    = {33},
  number    = {4},
  pages     = {974--987},
  year      = {1998},
  publisher = {JSTOR}
}

@article{angrist2000interpretation,
  title     = {The interpretation of instrumental variables estimators in simultaneous equations models with an application to the demand for fish},
  author    = {Angrist, Joshua D. and Graddy, Kathryn and Imbens, Guido W.},
  journal   = {Review of Economic Studies},
  volume    = {67},
  number    = {3},
  pages     = {499--527},
  year      = {2000},
  publisher = {Wiley-Blackwell}
}

@article{robinson1988root,
  title     = {Root-N-consistent semiparametric regression},
  author    = {Robinson, Peter M.},
  journal   = {Econometrica},
  volume    = {56},
  number    = {4},
  pages     = {931--954},
  year      = {1988},
  publisher = {JSTOR}
}

@article{imbens1994identification,
  title     = {Identification and estimation of local average treatment effects},
  author    = {Imbens, Guido W. and Angrist, Joshua D.},
  journal   = {Econometrica},
  volume    = {62},
  number    = {2},
  pages     = {467--475},
  year      = {1994},
  publisher = {JSTOR}
}

@article{heckman2005structural,
  title     = {Structural equations, treatment effects, and econometric policy evaluation},
  author    = {Heckman, James J. and Vytlacil, Edward},
  journal   = {Econometrica},
  volume    = {73},
  number    = {3},
  pages     = {669--738},
  year      = {2005},
  publisher = {Wiley Online Library}
}

@article{chernozhukov2022locally,
  title     = {{Locally Robust Semiparametric Estimation}},
  author    = {Chernozhukov, Victor and Escanciano, Juan Carlos and Ichimura, Hidehiko and Newey, Whitney K. and Robins, James M.},
  journal   = {Econometrica},
  volume    = {90},
  number    = {4},
  pages     = {1501--1535},
  year      = {2022},
  publisher = {Wiley Online Library}
}

@article{angrist1995two,
  title     = {Two-stage least squares estimation of average causal effects in models with variable treatment intensity},
  author    = {Angrist, Joshua D. and Imbens, Guido W.},
  journal   = {Journal of the American Statistical Association},
  volume    = {90},
  number    = {430},
  pages     = {431--442},
  year      = {1995},
  publisher = {Taylor \& Francis}
}

@article{mogstad2021causal,
  title     = {The causal interpretation of two-stage least squares with multiple instrumental variables},
  author    = {Mogstad, Magne and Torgovitsky, Alexander and Walters, Christopher R.},
  journal   = {American Economic Review},
  volume    = {111},
  number    = {11},
  pages     = {3663--3698},
  year      = {2021},
  publisher = {American Economic Association}
}

@article{esary1967association,
  title     = {Association of random variables, with applications},
  author    = {Esary, James D. and Proschan, Frank and Walkup, David W.},
  journal   = {The Annals of Mathematical Statistics},
  volume    = {38},
  number    = {5},
  pages     = {1466--1474},
  year      = {1967},
  publisher = {Institute of Mathematical Statistics}
}

@article{pitt1982positively,
  title     = {Positively correlated normal variables are associated},
  author    = {Pitt, Loren D.},
  journal   = {The Annals of Probability},
  volume    = {10},
  number    = {2},
  pages     = {496--499},
  year      = {1982},
  publisher = {Institute of Mathematical Statistics}
}

@article{milgrom1982theory,
  title     = {A theory of auctions and competitive bidding},
  author    = {Milgrom, Paul R. and Weber, Robert J.},
  journal   = {Econometrica},
  volume    = {50},
  number    = {5},
  pages     = {1089--1122},
  year      = {1982},
  publisher = {The Econometric Society}
}

@article{hausman1978specification,
  author  = {Hausman, Jerry A.},
  title   = {Specification Tests in Econometrics},
  journal = {Econometrica},
  year    = {1978}, volume = {46}, number = {6}, pages = {1251--1271}
}

@article{newey1994asymptotic,
  author  = {Newey, Whitney K.},
  title   = {The Asymptotic Variance of Semiparametric Estimators},
  journal = {Econometrica},
  year    = {1994}, volume = {62}, number = {6}, pages = {1349--1382}
}

@article{anderson1949estimation,
  author  = {Anderson, Theodore W. and Rubin, Herman},
  title   = {Estimation of the Parameters of a Single Equation in a
             Complete System of Stochastic Equations},
  journal = {Annals of Mathematical Statistics},
  year    = {1949}, volume = {20}, number = {1}, pages = {46--63}
}

@article{fieller1954some,
  author  = {Fieller, Edgar C.},
  title   = {Some Problems in Interval Estimation},
  journal = {Journal of the Royal Statistical Society. Series B},
  year    = {1954}, volume = {16}, number = {2}, pages = {175--185}
}

@article{dufour1997some,
  author  = {Dufour, Jean-Marie},
  title   = {Some Impossibility Theorems in Econometrics with
             Applications to Structural and Dynamic Models},
  journal = {Econometrica},
  year    = {1997}, volume = {65}, number = {6}, pages = {1365--1387}
}

@unpublished{kolesar2013estimation,
  author = {Koles{\'a}r, Michal},
  title  = {Estimation in an Instrumental Variables Model with
            Treatment Effect Heterogeneity},
  note   = {Working paper, Princeton University},
  year   = {2013}
}

@unpublished{sithole2026locally,
  author = {Sithole, Lonjezo},
  title  = {A Locally Robust Semiparametric Approach to Examiner {IV} Designs},
  note   = {Working paper, University of Michigan},
  year   = {2026},
  url    = {https://lonjezosithole.github.io/uploads/a-locally-robust-semiparametric-approach-to-examiner-iv-designs.pdf}
}

@unpublished{evdokimov2018inference,
  author = {Evdokimov, Kirill S. and Koles{\'a}r, Michal},
  title  = {Inference in Instrumental Variables Analysis with
            Heterogeneous Treatment Effects},
  note   = {Working paper, Princeton University},
  year   = {2018}
}

@article{lee2018consistent,
  author  = {Lee, Seojeong},
  title   = {A Consistent Variance Estimator for {2SLS} When Instruments
             Identify Different {LATEs}},
  journal = {Journal of Business \& Economic Statistics},
  year    = {2018}, volume = {36}, number = {3}, pages = {400--410}
}

@article{heckman2006understanding,
  author  = {Heckman, James J. and Urzua, Sergio and Vytlacil, Edward},
  title   = {Understanding Instrumental Variables in Models with Essential Heterogeneity},
  journal = {The Review of Economics and Statistics},
  year    = {2006}, volume = {88}, number = {3}, pages = {389--432}
}

@article{mogstad2024policy,
  author  = {Mogstad, Magne and Torgovitsky, Alexander and Walters, Christopher R.},
  title   = {Policy Evaluation with Multiple Instrumental Variables},
  journal = {Journal of Econometrics},
  year    = {2024}, volume = {243}, number = {1}, pages = {105718}
}

@article{goff2024vector,
  author  = {Goff, Leonard},
  title   = {A Vector Monotonicity Assumption for Multiple Instruments},
  journal = {Journal of Econometrics},
  year    = {2024}, volume = {241}, number = {1}, pages = {105735}
}

@techreport{blandhol2022tsls,
  author      = {Blandhol, Christine and Bonney, John and Mogstad, Magne and Torgovitsky, Alexander},
  title       = {When is {TSLS} Actually {LATE}?},
  institution = {National Bureau of Economic Research},
  type        = {NBER Working Paper}, number = {29709}, year = {2022},
  note        = {Forthcoming, Review of Economic Studies}
}

@article{sloczynski2024interpret,
  author  = {S{\l}oczy{\'n}ski, Tymon},
  title   = {When Should We (Not) Interpret Linear {IV} Estimands as {LATE}?},
  journal = {Review of Economic Studies}, year = {2024},
  note    = {Forthcoming}
}

@article{vytlacil2002independence,
  author  = {Vytlacil, Edward},
  title   = {Independence, Monotonicity, and Latent Index Models: An Equivalence Result},
  journal = {Econometrica},
  year    = {2002}, volume = {70}, number = {1}, pages = {331--341}
}

@article{huntingtonklein2020instruments,
  author  = {Huntington-Klein, Nick},
  title   = {Instruments with Heterogeneous Effects: Bias, Monotonicity, and Localness},
  journal = {Journal of Causal Inference}, volume = {8}, number = {1}, pages = {182--208}, year = {2020}
}

@article{garen1984,
  author  = {Garen, John},
  title   = {The Returns to Schooling: A Selectivity Bias Approach with a Continuous Choice Variable},
  journal = {Econometrica}, volume = {52}, number = {5}, pages = {1199--1218}, year = {1984}
}

@article{card2001estimating,
  author  = {Card, David},
  title   = {Estimating the Return to Schooling: Progress on Some Persistent Econometric Problems},
  journal = {Econometrica}, volume = {69}, number = {5}, pages = {1127--1160}, year = {2001}
}

@article{heckman2010testing,
  author  = {Heckman, James J. and Schmierer, Daniel and Urzua, Sergio},
  title   = {Testing the Correlated Random Coefficient Model},
  journal = {Journal of Econometrics}, volume = {158}, number = {2}, pages = {177--203}, year = {2010}
}

@article{newey1990semiparametric,
  author  = {Newey, Whitney K.},
  title   = {Semiparametric Efficiency Bounds},
  journal = {Journal of Applied Econometrics}, volume = {5}, number = {2}, pages = {99--135}, year = {1990}
}

@article{chamberlain1987asymptotic,
  author  = {Chamberlain, Gary},
  title   = {Asymptotic Efficiency in Estimation with Conditional Moment Restrictions},
  journal = {Journal of Econometrics}, volume = {34}, number = {3}, pages = {305--334}, year = {1987}
}

@article{belloni2012sparse,
  author  = {Belloni, Alexandre and Chen, Daniel and Chernozhukov, Victor and Hansen, Christian},
  title   = {Sparse Models and Methods for Optimal Instruments With an Application to Eminent Domain},
  journal = {Econometrica}, volume = {80}, number = {6}, pages = {2369--2429}, year = {2012}
}

@article{emmenegger2021regularizing,
  author  = {Emmenegger, Corinne and B\"{u}hlmann, Peter},
  title   = {Regularizing Double Machine Learning in Partially Linear Endogenous Models},
  journal = {Electronic Journal of Statistics}, volume = {15}, number = {2}, pages = {6461--6543}, year = {2021}
}

@article{scheidegger2025machine,
  author  = {Scheidegger, Cyrill and Guo, Zijian and B\"{u}hlmann, Peter},
  title   = {Inference for Heterogeneous Treatment Effects with Efficient Instruments and Machine Learning},
  journal = {Electronic Journal of Statistics}, volume = {20}, number = {1}, pages = {718--770}, year = {2026}
}

@article{angrist2022machine,
  author  = {Angrist, Joshua D. and Frandsen, Brigham},
  title   = {Machine Labor},
  journal = {Journal of Labor Economics}, volume = {40}, number = {S1}, pages = {S97--S140}, year = {2022}
}

@unpublished{chen2021mostly,
  author = {Chen, Jiafeng and Chen, Daniel L. and Lewis, Greg},
  title  = {Mostly Harmless Machine Learning: Learning Optimal Instruments in Linear {IV} Models},
  note   = {arXiv:2011.06158; NeurIPS 2020 Workshop on Machine Learning for Economic Policy},
  year   = {2021}
}

@article{staiger1997instrumental,
  author  = {Staiger, Douglas and Stock, James H.},
  title   = {Instrumental Variables Regression with Weak Instruments},
  journal = {Econometrica}, year = {1997}, volume = {65}, number = {3}, pages = {557--586}
}

@article{andrews2019weak,
  author  = {Andrews, Isaiah and Stock, James H. and Sun, Liyang},
  title   = {Weak Instruments in Instrumental Variables Regression: Theory and Practice},
  journal = {Annual Review of Economics}, year = {2019}, volume = {11}, pages = {727--753}
}

@article{mikusheva2022inference,
  author  = {Mikusheva, Anna and Sun, Liyang},
  title   = {Inference with Many Weak Instruments},
  journal = {The Review of Economic Studies}, year = {2022}, volume = {89}, number = {5}, pages = {2663--2686}
}

@book{bickel1993efficient,
  author    = {Bickel, Peter J. and Klaassen, Chris A. J. and Ritov, Ya'acov and Wellner, Jon A.},
  title     = {Efficient and Adaptive Estimation for Semiparametric Models},
  publisher = {Johns Hopkins University Press},
  address   = {Baltimore},
  year      = {1993}
}

@book{vandervaart1998asymptotic,
  author    = {van der Vaart, Aad W.},
  title     = {Asymptotic Statistics},
  publisher = {Cambridge University Press},
  address   = {Cambridge},
  year      = {1998}
}

@article{angrist1991does,
  author    = {Angrist, Joshua D. and Krueger, Alan B.},
  title     = {Does Compulsory School Attendance Affect Schooling and Earnings?},
  journal   = {The Quarterly Journal of Economics},
  volume    = {106},
  number    = {4},
  pages     = {979--1014},
  year      = {1991},
  publisher = {Oxford University Press}
}

@article{bound1995problems,
  author    = {Bound, John and Jaeger, David A. and Baker, Regina M.},
  title     = {Problems with Instrumental Variables Estimation When the Correlation Between the Instruments and the Endogenous Explanatory Variable Is Weak},
  journal   = {Journal of the American Statistical Association},
  volume    = {90},
  number    = {430},
  pages     = {443--450},
  year      = {1995},
  publisher = {Taylor \& Francis}
}
\end{singlespace}
%%%%%%%%%%%%%%%%%%%%%%%%%%%%%%%%%%%%%%%%%%%%%%%%%

%%%%%%%%%%%%%%%%%%%%%%%%%%%%%%%%%%%%%%%%%%%%%%%%%
%%%%% These commands start the appendix and change the Table & Figure numbering
\newpage
\appendix
\setcounter{table}{0}
\renewcommand{\tablename}{Appendix Table}
\renewcommand{\figurename}{Appendix Figure}
\renewcommand{\thetable}{A\arabic{table}}
\setcounter{figure}{0}
\renewcommand{\thefigure}{A\arabic{figure}}
%%%%%%%%%%%%%%%%%%%%%%%%%%%%%%%%%%%%%%%%%%%%%%%%%

\section{Proofs and supplementary results}\label{sec:appenda}

\numberwithin{equation}{section}
\renewcommand{\theequation}{\thesection\arabic{equation}}
\setcounter{equation}{0}

%%%%% Appendix "A" prefix (no dot) for every numbered theorem-like environment,
%%%%% matching the Table/Figure A-numbering above. Counters restart at the appendix;
%%%%% \numberwithin re-resets them at each further appendix section (B, C, ...).
%%%%% Placed inside the appendix, so main-text numbering (Theorem 1, Lemma 1, ...) is untouched.
\numberwithin{theorem}{section}     \renewcommand{\thetheorem}{\thesection\arabic{theorem}}         \setcounter{theorem}{0}
\numberwithin{lemma}{section}       \renewcommand{\thelemma}{\thesection\arabic{lemma}}             \setcounter{lemma}{0}
\numberwithin{proposition}{section} \renewcommand{\theproposition}{\thesection\arabic{proposition}} \setcounter{proposition}{0}
\numberwithin{corollary}{section}   \renewcommand{\thecorollary}{\thesection\arabic{corollary}}     \setcounter{corollary}{0}
\numberwithin{definition}{section}  \renewcommand{\thedefinition}{\thesection\arabic{definition}}   \setcounter{definition}{0}
\numberwithin{assumption}{section}  \renewcommand{\theassumption}{\thesection\arabic{assumption}}   \setcounter{assumption}{0}
\numberwithin{example}{section}     \renewcommand{\theexample}{\thesection\arabic{example}}         \setcounter{example}{0}
\numberwithin{remark}{section}      \renewcommand{\theremark}{\thesection\arabic{remark}}           \setcounter{remark}{0}
%%%%%%%%%%%%%%%%%%%%%%%%%%%%%%%%%%%%%%%%%%%%%%%%%

\subsection*{Proofs for Section~\ref{sec:microccm}: Covariance Monotonicity and its Microfoundation}

\begin{proof}[Proof of Proposition~\ref{prop:pmccm}]
  Step 1 (within-stratum covariance). Conditioning on $W$ and using $m_X(W) = \E[g_X \mid W]$,
  \[
    \E[\tilde{X}\,\tilde{g}_X \mid \Delta = \delta] = \E_{W \mid \Delta = \delta}\Big[\E\big[(X - m_X)(g_X - m_X) \mid W, \Delta = \delta\big]\Big].
  \]
  Since $g_X$ is a function of $(Z,W)$ and $Z \indep \Delta \mid W$ (a consequence of Assumption~\ref{ass:exogext}), we have $\E[g_X \mid W, \Delta = \delta] = \E[g_X \mid W] = m_X(W)$. The terms involving $m_X$ then collapse and the inner expectation reduces to $\Cov(X, g_X \mid W, \Delta = \delta)$, so
  \[
    \E[\tilde{X}\,\tilde{g}_X \mid \Delta = \delta] = \E_{W \mid \Delta = \delta}\big[\Cov(X, g_X \mid W, \Delta = \delta)\big].
  \]

  Step 2 (conditional response functions). Define $f_\delta(z,w) := \E[X \mid Z = z, W = w, \Delta = \delta]$. By Assumption~\ref{ass:exogext}, $Z \indep (\Delta, V) \mid W$, and weak union gives $Z \indep V \mid (W, \Delta)$. Hence the conditional law $F_{V \mid W, \Delta}(\cdot \mid w, \delta)$ does not depend on $z$, and
  \[
    f_\delta(z,w) = \int_{\mathcal{V}} h(z, w, v)\, dF_{V \mid W, \Delta}(v \mid w, \delta).
  \]
  Because $h(\cdot, w, v)$ is increasing in each coordinate under VM and the integrating measure is invariant to $z$, $f_\delta(\cdot,w)$ is increasing in each coordinate. The same argument with $F_{V \mid W}$, using $Z \indep V \mid W$, gives this property for $g_X(z, w) = \int_{\mathcal{V}} h(z, w, v)\, dF_{V \mid W}(v \mid w)$.

  Step 3 (association). By the law of total covariance and the $(Z,W)$-measurability of $g_X(Z,W)$,
  \[
    \Cov(X, g_X \mid W, \Delta = \delta) = \Cov\big(f_\delta(Z,W),\, g_X(Z,W) \mid W, \Delta = \delta\big).
  \]
  The integrand $f_\delta$ retains its dependence on $\delta$ through the law of $V$; only the distribution of $Z$ entering the covariance matters, and since $Z \indep \Delta \mid W$ that distribution is $F_{Z \mid W}$. Both $f_\delta(\cdot, w)$ and $g_X(\cdot, w)$ are coordinate-wise non-decreasing (Step 2), so PA yields
  \[
    \Cov\big(f_\delta(Z,W),\, g_X(Z,W) \mid W = w\big) \ge 0 \quad \text{a.s.}
  \]

  Step 4 (integration over $W$). The outer expectation in Step~1 preserves the inequality, so $\E[\tilde{X}\,\tilde{g}_X \mid \Delta = \delta] \ge 0$ a.s. Theorem~\ref{thm:representation} then gives non-negative weights and the convex-combination conclusion.
\end{proof}

\subsection*{Supplement to Section~\ref{sec:microccm}: A non-monotone example}

At the level of propensity rationalizability, Conditional Covariance Monotonicity is strictly weaker than the monotone selection of \citet{imbens1994identification}. Three instrument support points are the minimum needed for this separation.

\begin{example}[CCM does not require monotone selection]\label{ex:ccm-not-monotone}
  Work within a single stratum of $W$ (so $\tilde{A} = A - \E[A]$). Let the instrument have three points $Z\in\{0,1,2\}$, uniform ($\mathbb{P}(Z=z)=\tfrac13$), and let the heterogeneous effect take two values $\Delta\in\{a,b\}$ with $\mathbb{P}(\Delta=a)=\mathbb{P}(\Delta=b)=\tfrac12$, drawn independently of $Z$ (conditional exogeneity, $W$ trivial). Treatment $X\in\{0,1\}$ is Bernoulli given $(Z,\Delta)$ with propensity $p(z,\delta):=\mathbb{P}(X=1\mid Z=z,\Delta=\delta)$:
  \[
    \begin{array}{c|ccc}
                 & z=0 & z=1 & z=2                              \\\hline
      p(\cdot,a) & 0.1 & 0.5 & 0.9 \quad(\text{increasing})     \\
      p(\cdot,b) & 0.2 & 0.8 & 0.3 \quad(\text{single-peaked}).
    \end{array}
  \]
  Since selection depends on $\Delta$, essential heterogeneity is present. The signal is $g_X(z)=\E[X\mid Z=z]=\tfrac12\{p(z,a)+p(z,b)\}=(0.15,\,0.65,\,0.60)$, with $\mu:=\E[X]=\tfrac{7}{15}$ and centered values $g_X(z)-\mu=(-\tfrac{19}{60},\,\tfrac{11}{60},\,\tfrac{8}{60})$.

  (i) CCM holds. Because $Z\indep\Delta$, the within-stratum object reduces to a covariance across $Z$,
  \[
    \E[\tilde{X}\,\tilde{g}_X\mid\Delta=\delta]
    =\sum_{z}\tfrac13\big(p(z,\delta)-\mu\big)\big(g_X(z)-\mu\big)
    =\Cov_Z\!\big(p(\cdot,\delta),\,g_X\big),
  \]
  and direct computation gives $\E[\tilde{X}\tilde{g}_X\mid\Delta=a]=\tfrac{3}{50}=0.060>0$ and $\E[\tilde{X}\tilde{g}_X\mid\Delta=b]=\tfrac{37}{900}\approx0.041>0$. Thus CCM holds, and relevance follows from $\E[\tilde{X}\tilde{g}_X]=\tfrac{91}{1800}=\Var(g_X)>0$. The density weights are $\omega(a)=\tfrac{108}{91}$ and $\omega(b)=\tfrac{74}{91}$; after multiplying by the effect probabilities, the convex masses are $\tfrac{54}{91}$ and $\tfrac{37}{91}$. Hence $\beta_{IV}=\tfrac{54}{91}\,a+\tfrac{37}{91}\,b$ (Theorem~\ref{thm:representation}).

  (ii) No monotone single-index rule rationalizes the data. Any threshold rule $X=\mathbf{1}\{U\le\nu(Z)\}$ with $(U,\Delta)\indep Z$ (or the continuous analogue $X=h(\nu(Z),U)$, $h$ non-decreasing) yields $p(z,\delta)=F_{U\mid\Delta=\delta}(\nu(z))$, weakly increasing in $\nu(z)$ for every $\delta$; it thus requires a single ordering of $\{0,1,2\}$ making $z\mapsto p(z,\delta)$ weakly monotone for $\delta=a$ and $\delta=b$ simultaneously. None exists: $p(\cdot,a)$ is strict, so only the orders $0\!<\!1\!<\!2$ and $2\!<\!1\!<\!0$ render it monotone, and under both the single-peaked $p(\cdot,b)$ is non-monotone (a tie in $\nu$ is excluded, as it would force $p(z,a)=p(z',a)$). The two effect strata require incompatible orderings, so no monotone-selection model in the sense of \citet{imbens1994identification} can rationalize the propensity table. Yet the covariances stay non-negative: $g_X$ reorders $Z$ by predicted treatment ($0\!<\!2\!<\!1$), type $b$ is monotone in that ordering while type $a$ is monotone in the natural one, and each stratum's propensity still covaries positively with the aggregate signal.
\end{example}

\begin{remark}\label{rem:threepoints}
  The separation in Example~\ref{ex:ccm-not-monotone} cannot occur with a binary instrument at the level of propensity rationalizability. If $Z\in\{0,1\}$ then $\E[\tilde{X}\tilde{g}_X\mid\Delta=\delta]=\pi_0\pi_1\,(p(1,\delta)-p(0,\delta))\,D$ and $\E[\tilde{X}\tilde{g}_X]=\pi_0\pi_1 D^2$ with $D:=g_X(1)-g_X(0)$. CCM and relevance ($D\neq0$) then force $(p(1,\delta)-p(0,\delta))D\ge0$ for all $\delta$. After relabeling $Z$ if necessary, $p(1,\delta)\ge p(0,\delta)$ for every $\delta$, so the conditional propensities admit a common-order monotone threshold rationalization. This is an existence statement: without an additional restriction on the potential-treatment coupling, it does not imply that the true coupling contains no defiers. Hence $|\Supp(Z)|\ge 3$ is required for CCM to be incompatible with every such common-order monotone rationalization, and the example is minimal in that sense.
\end{remark}

\subsection*{Proofs for Section~\ref{sec:representation}: The representation theorem}

\begin{proof}[Proof of Theorem~\ref{thm:representation}, part~(i)]
  Step 1 (numerator). Substituting the CRC equation $Y = Y(0) + X\Delta$ into $\E[\tilde{Y}\tilde{g}]$ gives
  \[
    \tilde{Y} = Y - \E[Y\mid W] = \big(Y(0) - \E[Y(0)\mid W]\big) + \big(X\Delta - \E[X\Delta\mid W]\big).
  \]
  Thus, the numerator becomes
  \[
    \E[\tilde{Y}\tilde{g}] = \E\big[\big(Y(0) - \E[Y(0)\mid W]\big)\tilde{g}\big] + \E\big[\big(X\Delta - \E[X\Delta\mid W]\big)\tilde{g}\big].
  \]
  Assumption~\ref{ass:exog} gives $Z \indep Y(0) \mid W$. Because $\tilde{g}$ is a function of $(Z,W)$ and $\E[\tilde{g}\mid W]=0$, iterated expectations set the first term to zero. They also give $\E\big[\E[X\Delta\mid W]\,\tilde{g}\big] = \E\big[\E[X\Delta\mid W]\,\E[\tilde{g}\mid W]\big] = 0$. Therefore
  \[
    \E[\tilde{Y}\tilde{g}] = \E[X\Delta\tilde{g}].
  \]

  Step 2 (residualizing $X$). Substitute $X = \tilde{X} + \E[X\mid W]$:
  \[
    \E[X\Delta\tilde{g}] = \E\big[(\tilde{X} + \E[X\mid W])\Delta\tilde{g}\big] = \E[\tilde{X}\Delta\tilde{g}] + \E\big[\E[X\mid W]\,\Delta\tilde{g}\big].
  \]
  Assumption~\ref{ass:exog} also gives $Z \indep \Delta \mid W$, so $\E[\Delta\tilde{g}\mid W] = \E[\Delta\mid W]\,\E[\tilde{g}\mid W] = 0$. The second term is therefore zero, and
  \[
    \E[\tilde{Y}\tilde{g}] = \E[\tilde{X}\Delta\tilde{g}].
  \]

  Step 3 (conditioning on $\Delta$). Iterated expectations give
  \[
    \E[\tilde{X}\Delta\tilde{g}] = \E_\Delta\big[\E[\tilde{X}\Delta\tilde{g}\mid \Delta]\big] = \E_\Delta\big[\Delta\cdot\E[\tilde{X}\tilde{g}\mid \Delta]\big].
  \]

  Step 4 (forming the ratio). Dividing by the denominator gives
  \[
    \beta_{IV} = \frac{\E_\Delta\big[\Delta\cdot\E[\tilde{X}\tilde{g}\mid \Delta]\big]}{\E[\tilde{X}\tilde{g}]} = \E_\Delta\!\left[\Delta\cdot\left(\frac{\E[\tilde{X}\tilde{g}\mid \Delta]}{\E[\tilde{X}\tilde{g}]}\right)\right] \equiv \E_\Delta\big[\Delta\cdot\omega(\Delta)\big]. \qedhere
  \]
\end{proof}

\begin{proof}[Proof of Theorem~\ref{thm:representation}, parts~(ii)--(iii)]
  Taking the unconditional expectation of the conditional numerator in part~(i) gives
  \[
    \E_\Delta\big[\omega(\Delta)\big] = \frac{\E_\Delta\big[\E[\tilde{X}\tilde{g}\mid \Delta]\big]}{\E[\tilde{X}\tilde{g}]} = \frac{\E[\tilde{X}\tilde{g}]}{\E[\tilde{X}\tilde{g}]} = 1.
  \]
  Assumption~\ref{ass:ccm} gives $\E[\tilde{X}\tilde{g}\mid \Delta=\delta] \ge 0$ a.s., while Assumption~\ref{ass:relevance} and the maintained sign normalization make $\E[\tilde{X}\tilde{g}]$ strictly positive. Hence $\omega(\delta)\ge 0$ a.s. Together with part~(i), this makes $\beta_{IV}$ a convex combination of heterogeneous treatment effects.
\end{proof}

\subsection*{Proofs for Section~\ref{sec:multipleLATE}: Relation to the LATE Framework}

\begin{lemma}[Margin decomposition of the residual kernel]\label{lem:margin}
  Fix a covariate value $w$ outside the relevant null set and suppress $w$ from the notation. Assumption~\ref{ass:exogext} with $V=U$ then gives $g\indep(U,\Delta)$ in this conditional law. Under the threshold-crossing model, for almost every $\delta\in\Supp(\Delta)$,
  \[
    \E\big[\tilde X\,\tilde g \mid \Delta=\delta\big]
    \;=\; \sum_{\ell=2}^{L} S_\ell\,\pi_\ell(\delta),
    \qquad\text{and in particular}\qquad
    \E\big[\tilde X\,\tilde g\big] \;=\; \sum_{\ell=2}^{L} S_\ell\,\bar\pi_\ell .
  \]
\end{lemma}

\begin{proof}[Proof of Lemma~\ref{lem:margin}]
  Joint independence $g\indep(U,\Delta)$ implies both $\E[g\mid\Delta=\delta]=\bar g$ and $U\indep g\mid\Delta$. Hence centering $X$ does not change its conditional covariance with $\tilde g$, and
  \[
    \E[\tilde X\,\tilde g\mid\Delta=\delta]
    =\Cov\big(X,g\mid\Delta=\delta\big)
    =\E\big[(g-\bar g)\,\mathbb{P}(U\le g\mid \Delta=\delta)\big].
  \]
  Let $F^{\delta}(t):=\mathbb{P}(U\le t\mid\Delta=\delta)$ denote the conditional c.d.f.\ of $U$. Then
  \[
    \E\big[\tilde X\,\tilde g\mid\Delta=\delta\big]
    \;=\; \sum_{j=1}^{L} p_j\,(g_j-\bar g)\,F^{\delta}(g_j).
  \]
  Write $F^{\delta}(g_j)=\sum_{\ell\le j}\pi_\ell(\delta)$ with $\pi_1(\delta):=\mathbb{P}(U\le g_1\mid\Delta=\delta)$ and $\pi_\ell(\delta)=F^{\delta}(g_\ell)-F^{\delta}(g_{\ell-1})$ for $\ell\ge2$. Substituting and exchanging the order of summation (Abel/summation-by-parts),
  \[
    \sum_{j=1}^{L} p_j(g_j-\bar g)\sum_{\ell\le j}\pi_\ell(\delta)
    \;=\; \sum_{\ell=1}^{L}\pi_\ell(\delta)\underbrace{\sum_{j\ge \ell} p_j(g_j-\bar g)}_{=\,S_\ell}
    \;=\; \sum_{\ell=1}^{L} S_\ell\,\pi_\ell(\delta).
  \]
  The $\ell=1$ term vanishes because $S_1=\sum_{j\ge1}p_j(g_j-\bar g)=\E[g-\bar g]=0$; equivalently, always-takers ($U\le g_1$) have $X\equiv1$ and contribute nothing to a covariance with $\tilde g$. Hence $\E[\tilde X\,\tilde g\mid\Delta=\delta]=\sum_{\ell=2}^{L}S_\ell\,\pi_\ell(\delta)$. Taking $\E_\Delta$ and using $\E[\pi_\ell(\Delta)]=\bar\pi_\ell$ (law of iterated expectations) gives
  \[
    D(w)=\E[\tilde X\,\tilde g\mid W=w]=\sum_{\ell=2}^{L}S_\ell\,\bar\pi_\ell.
  \]
  Proposition~\ref{thm:multipleLATE} assumes this within-stratum denominator is strictly positive when it forms $\beta_{IV}(w)$.
\end{proof}

\begin{proof}[Proof of Proposition~\ref{thm:multipleLATE}, part~(i)]
  Work in the conditional law given $W=w$. Assumption~\ref{ass:exogext} with $V=U$ supplies the signal exogeneity required by Theorem~\ref{thm:cpt}, while $D(w)>0$ supplies relevance in this stratum. Applying that theorem conditionally and inserting Lemma~\ref{lem:margin} gives
  \[
    \begin{aligned}
      \E_\Delta\big[\Delta\cdot\E[\tilde X\,\tilde g\mid\Delta]\big]
       & =\sum_{\ell\in\mathcal L_+} S_\ell\,\E\big[\Delta\,\pi_\ell(\Delta)\big]
      =\sum_{\ell\in\mathcal L_+} S_\ell\int \delta\,\mathbb{P}(C_\ell\mid\Delta=\delta)\,dF_\Delta(\delta) \\
       & =\sum_{\ell\in\mathcal L_+} S_\ell\,\bar\pi_\ell\,\mathrm{LATE}_\ell,
    \end{aligned}
  \]
  where the last equality follows from $\E[\Delta\,\pi_\ell(\Delta)]=\bar\pi_\ell\,\mathrm{LATE}_\ell$ by iterated expectations, equivalently by the complier-density Bayes identity in part~(ii). Dividing by $D(w)=\sum_{\ell'\in\mathcal L_+}S_{\ell'}\bar\pi_{\ell'}$ gives $\beta_{IV}(w)=\sum_\ell\lambda_\ell\,\mathrm{LATE}_\ell$. Because $S_\ell\ge0$ and $\bar\pi_\ell=\mathbb{P}(C_\ell)>0$, the weights are non-negative and sum to one.
\end{proof}

\begin{proof}[Proof of Proposition~\ref{thm:multipleLATE}, part (ii)]
  By the conditional version of the weighting function in Theorem~\ref{thm:cpt} and Lemma~\ref{lem:margin},
  \[
    \omega_w(\delta)
    \;=\; \frac{\E[\tilde X\,\tilde g\mid\Delta=\delta,W=w]}{D(w)}
    \;=\; \frac{\sum_{\ell\in\mathcal L_+} S_\ell\,\pi_\ell(\delta)}
    {\sum_{\ell'\in\mathcal L_+} S_{\ell'}\,\bar\pi_{\ell'}}.
  \]
  Apply Bayes' rule to each $\ell\in\mathcal L_+$:
  \[
    \begin{aligned}
      \pi_\ell(\delta)\,dF_\Delta(\delta)
       & = \mathbb{P}(C_\ell\mid\Delta=\delta)\,dF_\Delta(\delta)
      = \bar\pi_\ell\,dF_{\Delta\mid C_\ell}(\delta),             \\
       & \Longrightarrow\quad
      \frac{\pi_\ell(\delta)}{\bar\pi_\ell}=\frac{dF_{\Delta\mid C_\ell}(\delta)}{dF_\Delta(\delta)}
      \quad\text{a.e.}
    \end{aligned}
  \]
  Substituting $\pi_\ell(\delta)=\bar\pi_\ell\,dF_{\Delta\mid C_\ell}/dF_\Delta$ into the numerator,
  \[
    \omega_w(\delta)
    \;=\; \frac{\sum_{\ell\in\mathcal L_+} S_\ell\,\bar\pi_\ell\,
      \big(dF_{\Delta\mid C_\ell}(\delta)/dF_\Delta(\delta)\big)}
    {\sum_{\ell'\in\mathcal L_+} S_{\ell'}\,\bar\pi_{\ell'}}
    \;=\; \sum_{\ell\in\mathcal L_+} \lambda_\ell\,
    \frac{dF_{\Delta\mid C_\ell}(\delta)}{dF_\Delta(\delta)},
  \]
  which is~\eqref{eq:complierdensity}. This identity is taken under the conditional law of $(\Delta,U,g)$ given $W=w$. Each likelihood ratio integrates to one under $F_\Delta$, so $\E_\Delta[\omega_w(\Delta)]=\sum_\ell\lambda_\ell=1$, as required by Theorem~\ref{thm:representation}.
\end{proof}

\begin{proof}[Proof of Corollary~\ref{cor:late-unconditional}]
  For $D(w)>0$, Proposition~\ref{thm:multipleLATE} gives
  \[
    \E[\tilde Y\tilde g\mid W=w]
    =D(w)\beta_{IV}(w)
    =D(w)\sum_{\ell\in\mathcal L_+(w)}\lambda_\ell(w)\mathrm{LATE}_\ell(w).
  \]
  The margin calculation in Lemma~\ref{lem:margin} also shows that both conditional moments vanish when $D(w)=0$. Iterated expectations therefore give~\eqref{eq:unconditional-late}. Non-negativity follows from $D(w)\ge0$ and $\lambda_\ell(w)\ge0$, while normalization follows from
  \[
    \int_{\{w:\,D(w)>0\}}\sum_{\ell\in\mathcal L_+(w)}
    \frac{D(w)\lambda_\ell(w)}{\E[D(W)]}\,dF_W(w)=1.
  \]
\end{proof}

\subsection*{Proofs for Section~\ref{sec:estimation}: The Moving Target}

Throughout this part, $\eta_0 = (l_Y, m_X, q_Y, g_X)$ denotes the true nuisance functions of Section~\ref{sec:nuisances}, $\mathcal{H} = \mathcal{H}(Z,W)$ is the heterogeneity residual in~\eqref{eq:Hdef}, and perturbation directions inherit the measurability of the corresponding nuisance. Thus directions for $l_Y,m_X$ are functions of $W$, whereas directions for $q_Y,g_X$ are functions of $(Z,W)$. The CRC equation and Assumption~\ref{ass:exog} imply the following identity, in which the $Y(0)$ term cancels:
\begin{equation}\label{eq:qminusl}
  q_Y - l_Y = \E[Y\mid Z,W] - \E[Y\mid W] = \E[X\Delta\mid Z,W] - \E[X\Delta\mid W].
\end{equation}

\begin{proof}[Proof of Proposition~\ref{prop:nonorth}]
  The moment depends on the signal coordinate only through the factor $(g_X - m_X)$. Perturbing $g_X \mapsto g_X + t\,\delta_g$ and differentiating at $t = 0$,
  \[
    \partial_{g_X}\,\E\big[\psi_{\mathrm{naive}}(O;\beta_0,\eta_0)\big][\delta_g] = \E\big[\big(Y - l_Y - \beta_0(X - m_X)\big)\,\delta_g\big].
  \]
  Since $\delta_g$ is $(Z,W)$-measurable, condition the first factor on $(Z,W)$, using $\E[Y\mid Z,W] = q_Y$ and $\E[X\mid Z,W] = g_X$:
  \[
    = \E\big[\big(q_Y - l_Y - \beta_0(g_X - m_X)\big)\,\delta_g\big].
  \]
  By~\eqref{eq:qminusl} the bracket equals $\mathcal{H}$, so the derivative is $\E[\mathcal{H}\,\delta_g]$. By the heterogeneity condition (Definition~\ref{def:esshet}), $\mathbb{P}(\mathcal{H}\neq 0) > 0$; taking $\delta_g = \mathcal{H}$ gives $\E[\mathcal{H}^2] > 0$, so the derivative does not vanish in all directions and Neyman orthogonality fails.
\end{proof}

\begin{proof}[Derivation of the expansion following Theorem~\ref{thm:driftvalid}]
  Write $U = Y - l_Y - \beta_0(X - m_X)$ and $V = g_X - m_X$. By a standard cross-fitted $Z$-estimation expansion, $\hat{\beta}_{\mathrm{naive}} - \beta_0 = J_0^{-1}\,\tfrac1N\sum_i \psi_{\mathrm{naive}}(O_i;\beta_0,\hat{\eta}) + o_{\mathbb{P}}(N^{-1/2})$, the Jacobian being
  \begin{equation}\label{eq:jacobian}
    J_0 = -\,\partial_\beta\,\E[\psi_{\mathrm{naive}}(O;\beta_0,\eta_0)] = \E[(X - m_X)(g_X - m_X)] = \E[(g_X - m_X)^2],
  \end{equation}
  where the last equality uses $\E[X - m_X\mid Z,W] = g_X - m_X$ under the specialization $g = g_X$. Because $\E[\psi_{\mathrm{naive}}(O;\beta_0,\eta_0)] = 0$, cross-fitting gives
  \[
    \tfrac1N\textstyle\sum_i \psi_{\mathrm{naive}}(O_i;\beta_0,\hat{\eta}) = \tfrac1N\textstyle\sum_i \psi_{\mathrm{naive}}(O_i;\beta_0,\eta_0) + \E\big[\psi_{\mathrm{naive}}(O;\beta_0,\hat{\eta})\big] + o_{\mathbb{P}}(N^{-1/2}),
  \]
  with the middle term evaluated at the independent, out-of-fold $\hat{\eta} = \eta_0 + \Delta\eta$. Substituting into~\eqref{eq:psinaive},
  \[
    \psi_{\mathrm{naive}}(O;\beta_0,\hat{\eta}) = \big(U - \Delta l_Y + \beta_0\,\Delta m_X\big)\big(V + \Delta g_X - \Delta m_X\big).
  \]
  Take expectations term by term. The constant term $\E[UV]$ is zero by the definition of $\beta_0$. Terms linear in the $W$-measurable errors vanish because $\E[V\mid W]=0$ and $\E[U\mid W]=0$. The remaining linear term is
  \[
    \E[U\,\Delta g_X] = \E\big[\Delta g_X\,\E[U\mid Z,W]\big] = \E\big[\Delta g_X\,\big(q_Y - l_Y - \beta_0(g_X - m_X)\big)\big] = \E[\mathcal{H}\,\Delta g_X].
  \]
  Cauchy--Schwarz bounds the quadratic term $\E\big[(-\Delta l_Y + \beta_0\Delta m_X)(\Delta g_X - \Delta m_X)\big]$ by sums of products of $L_2(\mathbb{P})$ nuisance-error norms. Assumption~\ref{ass:ratesP}(ii) makes each product $o_{\mathbb{P}}(N^{-1/2})$. Thus $\E[\psi_{\mathrm{naive}}(O;\beta_0,\hat{\eta})] = \E[\mathcal{H}\,\Delta g_X] + o_{\mathbb{P}}(N^{-1/2})$, and multiplication by $J_0^{-1}$ gives the stated expansion. If $\Delta \equiv \beta_0$, then~\eqref{eq:qminusl} gives $q_Y - l_Y = \beta_0(g_X - m_X)$, so $\mathcal{H}\equiv 0$ and the bias is zero.
\end{proof}

\subsection*{Proofs for Section~\ref{sec:fixedtarget}: The Fixed Target}

\begin{proof}[Proof of Theorem~\ref{thm:fixedtarget}, part~(i)]
  Identification. At $\eta_0$ each augmentation term carries the factor $(X - g_X)$, with $\E[X - g_X\mid Z,W] = 0$, multiplied by a $(Z,W)$-measurable factor ($q_Y - l_Y$ in $\psi_Y$; $g_X - m_X$ in $\psi_X$). Conditioning on $(Z,W)$, both vanish in expectation, leaving $\E[\psi_Y(O;\eta_0)] = \E[(Y - l_Y)(g_X - m_X)]$ and $\E[\psi_X(O;\eta_0)] = \E[(X - m_X)(g_X - m_X)]$, whose ratio is $\beta_0$ by~\eqref{eq:beta0}. Hence $\E[\psi_{\mathrm{CRC}}(O;\beta_0,\eta_0)] = 0$.

  Orthogonality. Differentiate $\E[\psi_{\mathrm{CRC}}] = \E[\psi_Y] - \beta_0\,\E[\psi_X]$ in each coordinate at $\eta_0$, in a direction $h$ of the appropriate measurability.

  (i) Direction $l_Y$ (appears only in $\psi_Y$):
  \[
    \partial_{l_Y}\E[\psi_Y][h] = \E\big[-h(g_X - m_X) - h(X - g_X)\big] = \E[-h(X - m_X)] = \E\big[-h\,\E[X - m_X\mid W]\big] = 0.
  \]

  (ii) Direction $q_Y$ (appears only in $\psi_Y$):
  \[
    \partial_{q_Y}\E[\psi_Y][h] = \E[h(X - g_X)] = \E\big[h\,\E[X - g_X\mid Z,W]\big] = 0.
  \]

  (iii) Direction $m_X$. In $\psi_Y$, $\partial_{m_X}\E[\psi_Y][h] = \E[-h(Y - l_Y)] = \E[-h\,\E[Y - l_Y\mid W]] = 0$. In $\psi_X$ the parameter appears three times, and
  \[
    \partial_{m_X}\E[\psi_X][h] = \E\big[-h(g_X - m_X) - h(X - m_X) - h(X - g_X)\big] = \E[-2h(X - m_X)] = 0,
  \]
  again by $\E[X - m_X\mid W] = 0$. The total derivative is $0 - \beta_0\cdot 0 = 0$.

  (iv) Direction $g_X$. In $\psi_Y$,
  \[
    \partial_{g_X}\E[\psi_Y][h] = \E\big[h(Y - l_Y) - h(q_Y - l_Y)\big] = \E[h(Y - q_Y)] = \E\big[h\,\E[Y - q_Y\mid Z,W]\big] = 0.
  \]
  In $\psi_X$ (three occurrences),
  \[
    \partial_{g_X}\E[\psi_X][h] = \E\big[h(X - m_X) - h(g_X - m_X) + h(X - g_X)\big] = \E[2h(X - g_X)] = 0.
  \]
  The total derivative is again $0 - \beta_0\cdot 0 = 0$. All four coordinate derivatives vanish, so $\psi_{\mathrm{CRC}}$ is Neyman orthogonal at $\eta_0$.
\end{proof}

\begin{proof}[Proof of Theorem~\ref{thm:fixedtarget}, part~(ii)]
  Since $\psi_{\mathrm{CRC}}$ is affine in $\beta$, the Jacobian is exact: $-\partial_\beta\E[\psi_{\mathrm{CRC}}(O;\beta_0,\eta_0)] = \E[\psi_X(O;\eta_0)] = \E[(g_X - m_X)^2] = J_0$, strictly positive by Assumption~\ref{ass:relevance}. The cross-fitted $Z$-estimator admits the expansion
  \[
    \sqrt{N}\,(\hat{\beta} - \beta_0) = J_0^{-1}\,\frac{1}{\sqrt{N}}\sum_{i=1}^N \psi_{\mathrm{CRC}}(O_i;\beta_0,\eta_0) + \sqrt{N}\,J_0^{-1} R_N + o_{\mathbb{P}}(1),
  \]
  where $R_N = \E[\psi_{\mathrm{CRC}}(O;\beta_0,\hat{\eta}) - \psi_{\mathrm{CRC}}(O;\beta_0,\eta_0)]$. The coefficient is $+J_0^{-1}$ because the first-order condition and $\partial_\beta\E[\psi_{\mathrm{CRC}}] = -J_0$ give $\hat{\beta} - \beta_0 = J_0^{-1}\tfrac1N\sum_i\psi_{\mathrm{CRC}}(O_i;\beta_0,\hat{\eta}) + o_{\mathbb{P}}(N^{-1/2})$. Each of $\psi_Y$ and $\psi_X$ is quadratic in the nuisances, so its Taylor expansion stops at second order. Theorem~\ref{thm:fixedtarget}(i) removes the linear terms. After the products $\pm\Delta l_Y\,\Delta g_X$ cancel, the remainder is
  \[
    R_Y = \E\big[\Delta l_Y\,\Delta m_X - \Delta g_X\,\Delta q_Y\big], \qquad R_X = \E\big[(\Delta m_X)^2 - (\Delta g_X)^2\big], \qquad R_N = R_Y - \beta_0 R_X.
  \]
  By Cauchy--Schwarz,
  \[
    |R_N| \le \|\Delta l_Y\|_{\mathbb{P},2}\|\Delta m_X\|_{\mathbb{P},2} + \|\Delta g_X\|_{\mathbb{P},2}\|\Delta q_Y\|_{\mathbb{P},2} + |\beta_0|\big(\|\Delta m_X\|_{\mathbb{P},2}^2 + \|\Delta g_X\|_{\mathbb{P},2}^2\big),
  \]
  and each product is $o_{\mathbb{P}}(N^{-1/2})$ by Assumption~\ref{ass:ratesP}(ii), so $\sqrt{N}\,R_N \xrightarrow{p} 0$. The leading term is a sample average of the mean-zero, finite-variance score $\psi_{\mathrm{CRC}}(O;\beta_0,\eta_0)$; the Lindeberg--L\'evy central limit theorem and Slutsky's theorem yield $\sqrt{N}(\hat{\beta} - \beta_0)\xrightarrow{d}\mathcal{N}(0,\,J_0^{-2}\Omega_0)$ with $\Omega_0 = \E[\psi_{\mathrm{CRC}}(O;\beta_0,\eta_0)^2]$.
\end{proof}

\begin{proof}[Proof of Theorem~\ref{thm:fixedtarget}, part~(iii)]
  Throughout, expectations are under the true $P$; write $V = g_X - m_X$, $e = X - g_X$, and note the identities
  \begin{equation}\label{eq:effz}
    \E[V\mid W] = \E[g_X\mid W] - m_X = 0, \qquad
    \E[e\mid Z,W] = 0, \qquad \E[Y - l_Y\mid W] = 0 .
  \end{equation}

  Step 0 (tangent space and uniqueness). The efficiency calculation uses the unrestricted observed-law model in which the four conditional expectations are functionals and finite moments are the only restrictions. Scores of regular one-dimensional submodels therefore have closed linear span $L_2^0(P)$. Any gradient $\phi$ satisfying $\tfrac{d}{dt}\beta(P_t)|_{t=0} = \E[\phi\,s]$ for every score $s$ is unique. Indeed, if $\phi_1$ and $\phi_2$ are gradients, then $\E[(\phi_1 - \phi_2)s] = 0$ for every $s\in L_2^0(P)$; taking $s = \phi_1 - \phi_2$ gives $\phi_1 = \phi_2$ a.s. This unique gradient is the efficient influence function \citep{bickel1993efficient,vandervaart1998asymptotic}. The CRC and exogeneity assumptions supply the causal interpretation of the observed-law functional.

  Step 1 (pathwise derivatives). For $A\in L_2$ and a coordinate $\sigma$-field $\mathcal{C}$, along a submodel with bounded score $s$,
  \begin{equation}\label{eq:effcm}
    \tfrac{d}{dt}\E_t[A\mid\mathcal{C}]\big|_{t=0}
    = \E\big[(A - \E[A\mid\mathcal{C}])\,s\mid\mathcal{C}\big],
  \end{equation}
  so, writing $\dot{(\cdot)}$ for the $t$-derivative at $0$, $\dot{g}_X = \E[e\,s\mid Z,W]$, $\dot{m}_X = \E[(X - m_X)s\mid W]$, $\dot{l}_Y = \E[(Y - l_Y)s\mid W]$; and $\tfrac{d}{dt}\E_t[h_t]|_{t=0} = \E[\dot{h}] + \E[h\,s]$.

  Step 2 (EIF of $\theta_Y$). Differentiating $\theta_Y(P_t) = \E_t[(Y - l_{Y,t})(g_{X,t} - m_{X,t})]$ in its four channels,
  \[
    \dot{\theta}_Y = \E[(Y - l_Y)V\,s] - \E[\dot{l}_Y V]
    + \E[(Y - l_Y)\dot{g}_X] - \E[(Y - l_Y)\dot{m}_X].
  \]
  By~\eqref{eq:effz} and $W$-measurability, $\E[\dot{l}_Y V] = 0$ and $\E[(Y - l_Y)\dot{m}_X] = 0$. As $\dot{g}_X$ is $(Z,W)$-measurable, projecting the outer factor gives $\E[(Y - l_Y)\dot{g}_X] = \E[(q_Y - l_Y)\,\E[e s\mid Z,W]] = \E[(q_Y - l_Y)e\,s]$. Thus $q_Y$ enters when differentiation of the embedded $g_X = \E[X\mid Z,W]$ reprojects $Y$ onto $(Z,W)$. Hence $\mathrm{EIF}(\theta_Y) = (Y - l_Y)V + (q_Y - l_Y)e - \theta_Y$, which is mean zero since $\E[(q_Y - l_Y)e] = 0$, and equals $\psi_Y(O;\eta_0) - \theta_Y$.

  Step 3 (EIF of $J_0$). From $J_0(P_t) = \E_t[V_t^2]$, $\dot{J}_0 = \E[V^2 s] + 2\E[V\dot{g}_X] - 2\E[V\dot{m}_X]$; the last term is zero by~\eqref{eq:effz}, and $\E[V\dot{g}_X] = \E[V e\,s]$ since $V$ is $(Z,W)$-measurable, so $\mathrm{EIF}(J_0) = V^2 + 2Ve - J_0 = V(2X - g_X - m_X) - J_0 = \psi_X(O;\eta_0) - J_0$.

  Step 4 (ratio and identification). As $J_0 > 0$, the map $(\theta,J)\mapsto\theta/J$ is Hadamard differentiable, so $\beta_0 = \theta_Y/J_0$ has gradient $\phi_{\beta_0} = J_0^{-1}(\mathrm{EIF}(\theta_Y) - \beta_0\,\mathrm{EIF}(J_0))$, and the constant $\theta_Y - \beta_0 J_0 = \E[UV] = 0$ cancels. Expanding $UV + e\mathcal{H}$ from the definitions ($U = Y - l_Y - \beta_0(X - m_X)$, $\mathcal{H} = (q_Y - l_Y) - \beta_0 V$),
  \[
    UV + e\mathcal{H}
    = (Y - l_Y)V + (q_Y - l_Y)e - \beta_0 V(2X - g_X - m_X)
    = J_0\,\phi_{\beta_0},
  \]
  so $\phi_{\beta_0} = J_0^{-1}(UV + e\mathcal{H}) = J_0^{-1}\psi_{\mathrm{CRC}}(O;\beta_0,\eta_0)$. By Step~0 this is the efficient influence function; the bound is $\Var(\phi_{\beta_0}) = J_0^{-2}\E[(UV + e\mathcal{H})^2] = J_0^{-2}\Omega_0$, attained by $\hat{\beta}_{\mathrm{CRC}}$ (Theorem~\ref{thm:an}).
\end{proof}

\begin{proof}[Proof of Theorem~\ref{thm:fixedtarget}, part~(iv)]
  If effects are homogeneous, $\Delta\equiv\beta_0$ and hence $\mathcal H=0$. Lemma~\ref{lem:scoreid} then gives $\psi_{\mathrm{CRC}}(\cdot;\beta_0,\eta_0)=\psi_{\mathrm{naive}}(\cdot;\beta_0,\eta_0)=UV$, so $\Omega_0=\Omega_0^{\mathrm{nv}}$.
\end{proof}

\subsection*{Conventions for the proofs of
  Sections~\ref{sec:estimation}--\ref{sec:diagnostics}} \label{app:conventions}

Unless stated otherwise, the following conventions use Assumption~\ref{ass:regB}.

Population residuals at the truth. Write
\[
  U := Y - l_Y - \beta_0(X - m_X), \qquad
  V := g_X - m_X, \qquad
  e := X - g_X .
\]
By Assumption~\ref{ass:regB}(i) and Cauchy--Schwarz, $U, V, e \in L_2(\mathbb{P})$, and $UV, ea \in L_2(\mathbb{P})$ for bounded $a$. Three conditional-mean identities recur. By $l_Y = \E[Y\mid W]$, $m_X = \E[X \mid W]$ and the tower property,
\begin{equation}\label{eq:UW}
  \E[U \mid W] = \E[Y\mid W] - l_Y - \beta_0(\E[X\mid W] - m_X) = 0 .
\end{equation}
By the identity~\eqref{eq:qminusl} and the definition~\eqref{eq:Hdef}, $q_Y - l_Y = \E[X\Delta\mid Z,W] - \E[X\Delta\mid W]$ and $\mathcal{H} = (q_Y - l_Y) - \beta_0 V$; since $\E[Y\mid Z,W] = q_Y$, $\E[X\mid Z,W] = g_X$ and hence $\E[X - m_X\mid Z,W] = g_X - m_X = V$,
\begin{equation}\label{eq:UZW}
  \E[U \mid Z, W]
  = (q_Y - l_Y) - \beta_0 V = \mathcal{H}(Z,W).
\end{equation}
Taking a further expectation and using~\eqref{eq:UW} gives $\E[\mathcal{H}\mid W] = 0$ (Lemma~\ref{lem:hgeom}(i)). Finally, by definition of $g_X = \E[X\mid Z,W]$, for every $(Z,W)$-measurable $f \in L_2$,
\begin{equation}\label{eq:eorth}
  \E[e\,f \mid Z,W] = 0, \qquad \text{hence}\quad \E[e\,f] = 0 .
\end{equation}

Cross-fitting notation. For a score coordinate $\phi(O;\beta,\eta)$ set $\Delta\eta_{j,k} := \hat{\eta}_{j,k} - \eta_{0,j}$, $\rho_{j,k} := \|\Delta\eta_{j,k}\|_{\mathbb{P},2}$, $\rho_N := \max_{j,k}\rho_{j,k}$, as in the cross-fitting construction of Section~\ref{sec:nuisances}. Conditional expectations $\E[\,\cdot \mid \mathcal{D}_{-k}]$ hold the $\mathcal{D}_{-k}$-measurable estimate $\hat{\eta}_k$ fixed and integrate a fresh $O \perp \mathcal{D}_{-k}$; thus $A_k, B_k, \beta_k^*$ of~\eqref{eq:pseudotrue} and the loadings $\E[\mathcal{H}\,\Delta g_{X,k}\mid\mathcal{D}_{-k}]$ are $\mathcal{D}_{-k}$-measurable random variables. Finite constants depending only on $(\bar{C}, C_4, |\beta_0|)$ are written $L, c, C$ and may change from line to line. We abbreviate $\|\cdot\| := \|\cdot\|_{\mathbb{P},2}$.

Two moment inequalities. For $(Z,W)$- or $W$-measurable $f$ with $\|f\|_\infty < \infty$ and $S \in L_4$,
\begin{equation}\label{eq:fS}
  \E[f^2 S^2] \le \|f\|_\infty^2\,\E[S^2],
  \qquad
  \big|\E[f\,S^2]\big| \le \|f\|\,\sqrt{\E[S^4]} .
\end{equation}
For bounded $b$ ($\|b\|_\infty<\infty$), the interpolation bound is $\|b\|_{\mathbb{P},4}^2 = (\E[b^4])^{1/2} \le (\|b\|_\infty^2\,\E[b^2])^{1/2} = \|b\|_\infty\,\|b\|$.

\begin{lemma}[Modulus of continuity of the score coordinates]\label{lem:l2cont}
  Let Assumption~\ref{ass:regB} hold, and consider $\phi \in \{\psi_{\mathrm{naive}}, \psi_Y, \psi_X, \psi_{\mathrm{CRC}}, \phi_{\mathrm{diag}}, j\}$, where $\phi_{\mathrm{diag}}(O;\beta,\eta) := (X - g)\{(q - l) - \beta(g - m)\}$ and $j(O;\eta) := (X - m)(g - m)$. There is $L = L(\bar{C}, C_4, |\beta_0|) < \infty$ such that, for all admissible $(\beta,\eta), (\beta',\eta')$ with nuisance coordinates bounded by $\bar{C}$ in sup-norm and $|\beta|,|\beta'| \le \bar{C}$,
  \[
    \big\| \phi(\,\cdot\,;\beta',\eta') - \phi(\,\cdot\,;\beta,\eta) \big\|
    \;\le\; L\,\Big( |\beta' - \beta|
    \;+\; \sum_{j} \|\eta'_j - \eta_j\|^{1/2} \Big).
  \]
  In particular, if $|\beta'-\beta| \to 0$ and $\max_j \|\eta'_j - \eta_j\| \to 0$, then $\|\phi(\,\cdot\,;\beta',\eta') - \phi(\,\cdot\,;\beta,\eta)\| \to 0$; and if $\phi$ does not depend on $\beta$ (the kernel $j$), the $|\beta'-\beta|$ term is absent.
\end{lemma}

\begin{proof}
  Every listed kernel is a finite sum of bilinear terms $F = F_1 F_2$ in which each factor is affine in $(\beta,\eta)$; the coefficients are either constants, bounded nuisance coordinates, or the components $Y, X \in L_4$ (Assumption~\ref{ass:regB}(i)). It suffices to bound the $L_2$ increment of a representative term of each of the two shapes that occur: (a) both factors bounded, and (b) one factor an $L_4$ component ($Y$ or $X$) and the other bounded. We use the telescoping identity $F' - F = (F_1' - F_1)F_2' + F_1(F_2' - F_2)$.

  Shape (a): $F=b_1b_2$, where both $b_1$ and $b_2$ are bounded affine factors. Each increment $b_i'-b_i$ is a bounded nuisance difference. By the first inequality in~\eqref{eq:fS}, $\|(b_1'-b_1)b_2'\| \le \|b_2'\|_\infty\,\|b_1'-b_1\| \le \bar{C}'\,\|b_1'-b_1\|$, and similarly for the second summand, where $\bar{C}'$ is a constant multiple of $\bar{C}$. Hence such a term contributes at most $C\sum_j\|\eta_j'-\eta_j\|$. On the relevant regime this is at most $C\sum_j\|\eta_j'-\eta_j\|^{1/2}$ because deviations are bounded by $2\bar C$ and $t \le (2\bar C)^{1/2} t^{1/2}$.

  Shape (b): $F = A\cdot b$ with $A \in \{Y, X\}$ (up to bounded shifts, themselves of shape (a)) and $b$ a bounded affine factor. The increment is $A(b' - b)$ plus shape-(a) remainders. Now
  \[
    \|A(b'-b)\|^2 = \E[A^2 (b'-b)^2]
    \;\overset{\text{C--S}}{\le}\; \|A\|_{\mathbb{P},4}^2\,\|b'-b\|_{\mathbb{P},4}^2
    \;\le\; \|A\|_{\mathbb{P},4}^2\,\|b'-b\|_\infty\,\|b'-b\|
    \;\le\; C\,\|b'-b\|,
  \]
  using the interpolation bound and $\|b'-b\|_\infty \le 2\bar C$, $\|A\|_{\mathbb{P},4}^2 \le \sqrt{C_4}$. As $b'-b$ is an affine function of the nuisance differences, $\|b'-b\| \le \sum_j\|\eta_j'-\eta_j\|$, so this term contributes $\le C\big(\sum_j\|\eta_j'-\eta_j\|\big)^{1/2} \le C\sum_j\|\eta_j'-\eta_j\|^{1/2}$.

  The $\beta$-increment. Each kernel is affine in $\beta$ with slope a bounded$\,\times\,L_4$ bilinear form $r(O;\eta)$ (e.g.\ $-j$ for $\psi_{\mathrm{naive}}$ and $-(X-g)(g-m)$ for $\phi_{\mathrm{diag}}$), so $\|\phi(\cdot;\beta',\eta') - \phi(\cdot;\beta,\eta')\| = |\beta'-\beta|\, \|r(\cdot;\eta')\| \le C\,|\beta'-\beta|$, since $\|r(\cdot;\eta')\| \le C$ by the $L_4$ bound and boundedness of the nuisances.

  Summing the finitely many terms and combining the $\beta$-increment (bounding $\phi(\beta',\eta')-\phi(\beta,\eta)$ by $\|\phi(\beta',\eta')-\phi(\beta,\eta')\| + \|\phi(\beta,\eta')-\phi(\beta,\eta)\|$) gives the stated bound with a suitable $L$. The two convergence claims follow from this bound.
\end{proof}

\subsection*{Proofs for Section~\ref{sec:geometry}}

\begin{proof}[Proof of Lemma~\ref{lem:hgeom}, parts~(i)--(iii)]
  Part (i) is~\eqref{eq:UW}--\eqref{eq:UZW} in the Conventions: $\E[\mathcal{H}\mid W] = \E[\E[U\mid Z,W]\mid W] = \E[U\mid W] = 0$.

  (ii) Using~\eqref{eq:UZW}, $\tilde{g}_X = g_X - \E[g_X\mid W]$, and that $\tilde g_X$ is $(Z,W)$-measurable,
  \[
    \E[\mathcal{H}\,\tilde{g}_X]
    = \E\big[\E[U\mid Z,W]\,\tilde{g}_X\big]
    = \E[U\,\tilde{g}_X].
  \]
  Now $V = g_X - m_X = \tilde g_X + (\E[g_X\mid W] - m_X)$, and the second summand is $W$-measurable, so by~\eqref{eq:UW}, $\E[U(\E[g_X\mid W]-m_X)] = \E[(\E[g_X\mid W]-m_X)\E[U\mid W]] = 0$; hence $\E[U\tilde g_X] = \E[UV]$. By the definition~\eqref{eq:beta0} of $\beta_0$ as the least-squares coefficient of $Y - l_Y$ on $V$,
  \[
    \E[UV] = \E\big[(Y - l_Y)V\big] - \beta_0\,\E\big[(X - m_X)V\big]
    = \beta_0 J_0 - \beta_0 J_0 = 0,
  \]
  because $\E[(Y-l_Y)V] = \beta_0\,\E[V^2] = \beta_0 J_0$ by~\eqref{eq:beta0} and $\E[(X-m_X)V] = \E[V\cdot V] + \E[(X-g_X)V]$; the last term is $\E[eV]=0$ by~\eqref{eq:eorth} since $V$ is $(Z,W)$-measurable, so $\E[(X-m_X)V]=\E[V^2]=J_0$. Thus $\E[\mathcal H\tilde g_X]=\E[UV]=0$.

  (iii) $\mathcal{G} = \mathbb{R}\tilde{g}_X \oplus L_2(\sigma(W))$ is the sum of a one-dimensional subspace and the closed subspace $L_2(\sigma(W))\subset L_2(\sigma(Z,W))$. The sum is $L_2$-orthogonal: for $d \in L_2(\sigma(W))$, $\E[\tilde g_X d] = \E[d\,\E[\tilde g_X\mid W]] = 0$ since $\E[\tilde g_X \mid W] = 0$. An orthogonal sum of a finite-dimensional subspace and a closed subspace is closed, so $\mathcal G$ is closed and $\Pi_{\mathcal G}$ is well defined. By (i), $\mathcal H \perp L_2(\sigma(W))$ (take $d=\E[\mathcal H\mid W]$-projections: $\E[\mathcal H d]=\E[d\E[\mathcal H\mid W]]=0$); by (ii), $\mathcal H\perp\tilde g_X$. Hence $\mathcal H\perp\mathcal G$, so $\E[\mathcal H\Pi_{\mathcal G}\delta_g]=0$ and $\E[\mathcal H\delta_g]=\E[\mathcal H(I-\Pi_{\mathcal G})\delta_g]$; the bound is Cauchy--Schwarz.
\end{proof}

The next lemma decomposes the heterogeneity residual into the observable and selection-on-gains components used in Section~\ref{sec:geometry}.

\begin{lemma}[Sources of the Heterogeneity Residual]\label{lem:hdecomp}
  Let Assumptions~\ref{ass:exog} and~\ref{ass:relevance} hold and suppose $\E[\Delta^2] + \E[X^2\Delta^2] < \infty$. Define $\bar{\Delta}_W := \E[\Delta \mid W]$ and the selection-on-gains kernel $\kappa(Z,W) := \Cov(X, \Delta \mid Z, W)$, with $\tilde{\kappa} := \kappa - \E[\kappa \mid W]$. Then:
  \begin{itemize}
    \item[(i)] $\displaystyle \mathcal{H}
            \;=\; \big(\bar{\Delta}_W - \beta_0\big)\,\tilde{g}_X \;+\; \tilde{\kappa}$;
    \item[(ii)] if $\kappa = 0$ a.s.\ (in particular if
          $X \indep \Delta \mid (Z,W)$), then, with $v(W) := \E[\tilde{g}_X^2 \mid W]$,
          \[
            \beta_0 \;=\; \frac{\E\big[v(W)\,\bar{\Delta}_W\big]}{\E[v(W)]},
            \qquad
            \E[\mathcal{H}^2] \;=\; \E\big[v(W)\,\big(\bar{\Delta}_W - \beta_0\big)^2\big];
          \]
    \item[(iii)] consequently, the condition of Definition~\ref{def:esshet} holds, and the
          naive score is non-orthogonal, whenever the conditional average effect $\bar{\Delta}_W$ varies across covariate strata that the instrument moves ($v(W) > 0$), even in the complete absence of selection on gains.
  \end{itemize}
\end{lemma}

\begin{proof}[Proof of Lemma~\ref{lem:hdecomp}]
  (i) By Assumption~\ref{ass:exog} ($Z \indep (Y(0),\Delta)\mid W$), for any $(Z,W)$-measurable conditioning, $\E[\Delta \mid Z,W] = \E[\Delta\mid W] = \bar\Delta_W$. Decompose the reduced-form residual using~\eqref{eq:qminusl}:
  \[
    q_Y - l_Y
    = \E[X\Delta\mid Z,W] - \E[X\Delta\mid W].
  \]
  Write $\E[X\Delta\mid Z,W] = \E[X\mid Z,W]\,\E[\Delta\mid Z,W] + \Cov(X,\Delta\mid Z,W) = g_X\,\bar\Delta_W + \kappa$, and $\E[X\Delta\mid W] = \E[X\mid W]\bar\Delta_W + \Cov(X,\Delta\mid W) = m_X\bar\Delta_W + \E[\kappa\mid W] + \Cov(\E[X\mid Z,W],\bar\Delta_W\mid W)$. Because $\bar\Delta_W$ is $W$-measurable, $\Cov(X,\Delta\mid W) = \E[\Cov(X,\Delta\mid Z,W)\mid W] + \Cov(g_X,\bar\Delta_W\mid W) = \E[\kappa\mid W] + \bar\Delta_W\cdot 0$; more directly, subtracting,
  \[
    q_Y - l_Y = (g_X - m_X)\,\bar\Delta_W + \big(\kappa - \E[\kappa\mid W]\big)
    = \bar\Delta_W\,V + \tilde\kappa .
  \]
  Then $\mathcal H = (q_Y - l_Y) - \beta_0 V = (\bar\Delta_W - \beta_0)V + \tilde\kappa$. Since $m_X = \E[X\mid W] = \E[g_X\mid W]$ by the tower property, $V = g_X - m_X = g_X - \E[g_X\mid W] = \tilde g_X$ exactly, so $\mathcal H = (\bar\Delta_W - \beta_0)\tilde g_X + \tilde\kappa$, the stated identity.

  (ii) If $\kappa \equiv 0$ then $\tilde\kappa = 0$ and $\mathcal H = (\bar\Delta_W-\beta_0)\tilde g_X$. By Lemma~\ref{lem:hgeom}(ii), $0 = \E[\mathcal H\tilde g_X] = \E[(\bar\Delta_W-\beta_0)\tilde g_X^2] = \E[(\bar\Delta_W-\beta_0)v(W)]$, using $\E[\tilde g_X^2\mid W] = v(W)$ and $W$-measurability of $\bar\Delta_W-\beta_0$. Solving gives $\beta_0 = \E[v(W)\bar\Delta_W]/\E[v(W)]$. Substitution into $\mathcal H$ gives $\E[\mathcal H^2] = \E[(\bar\Delta_W-\beta_0)^2\tilde g_X^2] = \E[v(W)(\bar\Delta_W-\beta_0)^2]$.

  (iii) By (ii), $\E[\mathcal H^2] = \E[v(W)(\bar\Delta_W-\beta_0)^2] > 0$ whenever $\bar\Delta_W$ is non-constant on $\{v(W)>0\}$. Definition~\ref{def:esshet} states $\mathbb P(\mathcal H\neq 0)>0$, equivalently $\E[\mathcal H^2]>0$, and this can therefore occur with $\kappa\equiv0$.
\end{proof}

\begin{proof}[Proof of Lemma~\ref{lem:hgeom}, part~(iv)]
  For $g' = ag + d$ with $d$ $W$-measurable, $\widetilde{g'} = g' - \E[g'\mid W] = a(g - \E[g\mid W]) + (d - d) = a\tilde g$. Hence $\E[\tilde X\widetilde{g'}] = a\E[\tilde X\tilde g]$ and $\E[\tilde Y\widetilde{g'}] = a\E[\tilde Y\tilde g]$, so $\beta_{IV}(g') = \E[\tilde Y\widetilde{g'}]/\E[\tilde X\widetilde{g'}] = (a\E[\tilde Y\tilde g])/(a\E[\tilde X\tilde g]) = \beta_{IV}(g)$, the factor $a\neq0$ cancelling. In particular $\E[\tilde X\widetilde{g'}] = a\E[\tilde X\tilde g]\neq0$.
\end{proof}

\begin{proof}[Proof of Theorem~\ref{thm:driftvalid}, parts~(i)--(ii)]
  Write $\delta := \delta_g$ and recall $\beta_{IV}(g) = \E[\tilde Y\tilde g]/\E[\tilde X\tilde g]$ with $\tilde\cdot$ the $W$-demeaning. For $g_X + \delta$, linearity of demeaning gives $\widetilde{g_X+\delta} = \tilde g_X + \tilde\delta$, and, arguing as in~\eqref{eq:eorth} (numerator/denominator inner products depend on $\delta$ only through its $(Z,W)$-measurable self),
  \[
    \E[\tilde X\,\widetilde{(g_X+\delta)}] = \E[\tilde X\tilde g_X] +
    \E[\tilde X\tilde\delta] = J_0 + \E[\tilde g_X\delta],
  \]
  since $\E[\tilde X\tilde g_X] = \E[(X-m_X)\tilde g_X] = \E[V\tilde g_X] + \E[e\tilde g_X] = \E[\tilde g_X^2] = J_0$ (the cross term $\E[e\tilde g_X]=0$ by~\eqref{eq:eorth}, and $\E[V\tilde g_X]=\E[\tilde g_X^2]=J_0$ as in the proof of Lemma~\ref{lem:hgeom}(ii)); and $\E[\tilde X\tilde\delta] = \E[\tilde g_X\delta]$ likewise. For the numerator, $\E[\tilde Y\,\widetilde{(g_X+\delta)}] = \E[\tilde Y\tilde g_X] + \E[\tilde Y\tilde\delta]$. Now $\E[\tilde Y\tilde g_X] = \beta_0 J_0$ by definition of $\beta_0$ (equivalently $\beta_{IV}(g_X)=\beta_0$), and
  \[
    \E[\tilde Y\tilde\delta] = \E[(Y-l_Y)\delta] =
    \E[q_Y\delta - l_Y\delta] = \E[(q_Y - l_Y)\delta]
    = \E[(\mathcal H + \beta_0 V)\delta] = \E[\mathcal H\delta] +
    \beta_0\E[\tilde g_X\delta],
  \]
  using $\E[(Y-l_Y)\delta]=\E[(\E[Y\mid Z,W]-l_Y)\delta]=\E[(q_Y-l_Y)\delta]$ for $(Z,W)$-measurable $\delta$, $\mathcal H = (q_Y-l_Y)-\beta_0 V$, and $\E[V\delta]=\E[\tilde g_X\delta]$ (same demeaning argument). Therefore
  \[
    \beta_{IV}(g_X+\delta) = \frac{\beta_0 J_0 + \E[\mathcal H\delta] +
      \beta_0\E[\tilde g_X\delta]}{J_0 + \E[\tilde g_X\delta]}
    = \beta_0 + \frac{\E[\mathcal H\delta]}{J_0 + \E[\tilde g_X\delta]},
  \]
  which is~\eqref{eq:driftexact}, proving (i).

  (ii) Let $D := J_0 + \E[\tilde g_X\delta]$. By Cauchy--Schwarz $|\E[\tilde g_X\delta]| \le \|\tilde g_X\|\,\|\delta\|$, so if $\|\delta\| < J_0/\|\tilde g_X\|$ then $D \ge J_0 - \|\tilde g_X\|\|\delta\| > 0$. Subtracting the linear term,
  \[
    \beta_{IV}(g_X+\delta) - \beta_0 - J_0^{-1}\E[\mathcal H\delta]
    = \E[\mathcal H\delta]\Big(\frac1D - \frac1{J_0}\Big)
    = -\,\frac{\E[\mathcal H\delta]\,\E[\tilde g_X\delta]}{J_0\,D}.
  \]
  Bounding $|\E[\mathcal H\delta]| \le \|\mathcal H\|\|\delta\|$, $|\E[\tilde g_X\delta]| \le \|\tilde g_X\|\|\delta\|$ and $D \ge J_0 - \|\tilde g_X\|\|\delta\|$ gives, whenever $\|\delta\| < J_0/\|\tilde g_X\|$, the second-order remainder bound
  \begin{equation}\label{eq:driftremainder}
    \Big| \beta_{IV}(g_X + \delta) - \beta_0 - J_0^{-1}\E[\mathcal H\delta] \Big|
    \;\le\; \frac{\|\mathcal H\|\,\|\tilde g_X\|}{J_0\big(J_0 - \|\tilde g_X\|\,\|\delta\|\big)}\;\|\delta\|^2 ;
  \end{equation}
  since this is $O(\|\delta\|^2)$, $\beta_{IV}$ is Fr\'echet differentiable at $g_X$ with derivative $\delta\mapsto J_0^{-1}\E[\mathcal H\delta]$ (a bounded linear functional by Cauchy--Schwarz).

  The null-space form in~(ii) follows by substituting Lemma~\ref{lem:hgeom}(iii), $\E[\mathcal H\delta]=\E[\mathcal H(I-\Pi_{\mathcal G})\delta]$, into \eqref{eq:driftexact}. More explicitly, write $\delta=r+c\tilde g_X+d$, where $r:=(I-\Pi_{\mathcal G})\delta$, $c:=\E[\tilde g_X\delta]/J_0$, and $d\in L_2(\sigma(W))$. Then
  \[
    \beta_{IV}(g_X+\delta)-\beta_0
    =\frac{\E[\mathcal H r]}{(1+c)J_0},
    \qquad c\neq-1.
  \]
  Thus the numerator and the Fr\'echet derivative depend only on $r$, while the exact finite drift also depends on the retained signal scale $1+c$. For $\delta\in\mathcal G$, $r=0$, so the drift vanishes whenever $c\neq-1$; when $c=-1$, the perturbed denominator is zero and $\beta_{IV}(g_X+\delta)$ is undefined.
\end{proof}

\subsection*{The exact second-order expansions (Lemma~\ref{lem:exact})}

\begin{lemma}[Exact Second-Order Moment Expansions]\label{lem:exact}
  Let Assumptions~\ref{ass:exog} and~\ref{ass:relevance} hold with $g = g_X$. For any admissible nuisance value $\eta = (l, m, q, g)$ with square-integrable coordinates of the correct measurability, writing $\Delta\eta := \eta - \eta_0$ componentwise,
  \begin{align*}
    \E\big[\psi_{\mathrm{naive}}(O;\beta_0,\eta)\big]
     & \;=\; \E[\mathcal{H}\,\Delta g_X]
    \;+\; \E\big[(-\Delta l_Y + \beta_0\,\Delta m_X)(\Delta g_X - \Delta m_X)\big], \\
    \E\big[\psi_{\mathrm{CRC}}(O;\beta_0,\eta)\big]
     & \;=\; \E\big[\Delta l_Y\,\Delta m_X - \Delta g_X\,\Delta q_Y\big]
    \;-\; \beta_0\,\E\big[(\Delta m_X)^2 - (\Delta g_X)^2\big],
  \end{align*}
  Both identities hold for every admissible $\eta$, not only near $\eta_0$. The naive expansion contains the linear term $\E[\mathcal{H}\,\Delta g_X]$ from Proposition~\ref{prop:nonorth}; the CRC expansion is quadratic in $\eta-\eta_0$, as Theorem~\ref{thm:fixedtarget}(i) requires.
\end{lemma}

\begin{proof}[Proof of Lemma~\ref{lem:exact}]
  Write $\Delta l := \Delta l_Y$, $\Delta m := \Delta m_X$, $\Delta q := \Delta q_Y$, $\Delta g := \Delta g_X$ for the coordinates of $\eta - \eta_0$, and recall $U, V, e$ and the residuals $R := Y - l_Y$ (with $\E[R\mid W]=0$, $\E[R\mid Z,W] = q_Y - l_Y =: P$) and $\mathcal H = P - \beta_0 V$.

  Naive score. At $\beta_0$,
  \[
    Y - l - \beta_0(X - m) = U - \Delta l + \beta_0\Delta m,
    \qquad
    g - m = V + \Delta g - \Delta m,
  \]
  so $\psi_{\mathrm{naive}}(O;\beta_0,\eta) = (U - \Delta l + \beta_0\Delta m)(V + \Delta g - \Delta m)$. Take expectations and expand into four inner products:
  \begin{itemize}
    \item $\E[UV] = 0$ (proof of Lemma~\ref{lem:hgeom}(ii));
    \item $\E[U(\Delta g - \Delta m)] = \E[\mathcal H\,\Delta g]$: indeed
          $\E[U\Delta g] = \E[\E[U\mid Z,W]\Delta g] = \E[\mathcal H\Delta g]$ by~\eqref{eq:UZW}, and $\E[U\Delta m] = \E[\E[U\mid W]\Delta m] = 0$ by~\eqref{eq:UW};
    \item $\E[(-\Delta l + \beta_0\Delta m)V] =
            \E[(-\Delta l + \beta_0\Delta m)\,\E[V\mid W]] = 0$, since $\E[V\mid W] = \E[g_X\mid W] - m_X = 0$ (as $m_X = \E[X\mid W] = \E[\E[X\mid Z,W]\mid W] = \E[g_X\mid W]$) and $-\Delta l+\beta_0\Delta m$ is $W$-measurable;
    \item the remaining product $\E[(-\Delta l + \beta_0\Delta m)
              (\Delta g - \Delta m)]$ is retained.
  \end{itemize}
  Summing the four terms proves the first identity.

  CRC score. We compute $\E[\psi_Y(O;\eta)]$ and $\E[\psi_X(O;\eta)]$ separately (neither depends on $\beta$). Using $Y - l = R - \Delta l$, $g - m = V + \Delta g - \Delta m$, $X - g = e - \Delta g$, $q - l = P + \Delta q - \Delta l$:
  \begin{align*}
    \E[(Y-l)(g-m)]
     & = \E[(R-\Delta l)(V + \Delta g - \Delta m)]                                  \\
     & = \underbrace{\E[RV]}_{=\beta_0 J_0}
    + \underbrace{\E[R(\Delta g - \Delta m)]}_{=\E[P\Delta g]}
    - \underbrace{\E[\Delta l\,V]}_{=0}
    - \E[\Delta l(\Delta g - \Delta m)]                                             \\
     & = \beta_0 J_0 + \E[P\Delta g] - \E[\Delta l\Delta g] + \E[\Delta l\Delta m],
  \end{align*}
  where $\E[RV] = \E[(Y-l_Y)V] = \beta_0 J_0$ by~\eqref{eq:beta0}; $\E[R\Delta g] = \E[\E[R\mid Z,W]\Delta g] = \E[P\Delta g]$ and $\E[R\Delta m] = \E[\E[R\mid W]\Delta m] = 0$; and $\E[\Delta l V] = \E[\Delta l\,\E[V\mid W]] = 0$. Next,
  \begin{align*}
    \E[(X-g)(q-l)]
     & = \E[(e - \Delta g)(P + \Delta q - \Delta l)]                             \\
     & = \underbrace{\E[eP]}_{=0} + \underbrace{\E[e(\Delta q - \Delta l)]}_{=0}
    - \E[\Delta g\,P] - \E[\Delta g(\Delta q - \Delta l)]                        \\
     & = -\E[P\Delta g] - \E[\Delta g\Delta q] + \E[\Delta g\Delta l],
  \end{align*}
  the first two terms vanishing by~\eqref{eq:eorth} ($P, \Delta q, \Delta l$ are $(Z,W)$-measurable). Adding, the $\E[P\Delta g]$ terms cancel and $-\E[\Delta l\Delta g] + \E[\Delta g\Delta l] = 0$, leaving
  \[
    \E[\psi_Y(O;\eta)] = \beta_0 J_0 + \E[\Delta l\Delta m - \Delta g\Delta q].
  \]
  For $\psi_X = (g-m)(2X - m - g)$, use $2X - m - g = 2e + V - \Delta m - \Delta g$ (since $2X - m_X - g_X = 2e + V$):
  \begin{align*}
    \E[\psi_X(O;\eta)]
     & = \E[(V + \Delta g - \Delta m)(2e + V - \Delta m - \Delta g)] \\
     & = 2\,\underbrace{\E[e(V + \Delta g - \Delta m)]}_{=0}
    + \E[(V + \Delta g - \Delta m)(V - \Delta m - \Delta g)]         \\
     & = \E[V^2] - 2\,\underbrace{\E[V\Delta m]}_{=0}
    - \E[(\Delta g)^2] + \E[(\Delta m)^2]
    = J_0 + \E[(\Delta m)^2 - (\Delta g)^2],
  \end{align*}
  where $\E[e(\cdot)] = 0$ by~\eqref{eq:eorth} ($V,\Delta g,\Delta m$ are $(Z,W)$-measurable), the cross terms $-\E[V\Delta g] + \E[\Delta g V]$ and $-\E[\Delta m V]-\E[V\Delta m]$ combine as shown, and $\E[V\Delta m] = \E[\Delta m\,\E[V\mid W]] = 0$. Therefore
  \[
    \E[\psi_{\mathrm{CRC}}(O;\beta_0,\eta)]
    = \E[\psi_Y] - \beta_0\E[\psi_X]
    = \E[\Delta l\Delta m - \Delta g\Delta q]
    - \beta_0\,\E[(\Delta m)^2 - (\Delta g)^2],
  \]
  which proves the second identity.
\end{proof}

\begin{proof}[Proof of Theorem~\ref{thm:driftvalid}, parts~(iii)--(iv)]
  Write $\psi := \psi_{\mathrm{naive}}$ and let $j(O;\eta) := (X-m)(g-m)$ be its negative $\beta$-slope. Affinity gives $\psi(O;\beta,\eta) = \psi(O;\beta',\eta) - (\beta-\beta')\,j(O;\eta)$ for all $\beta,\beta'$. Also, since $m_X = \E[g_X\mid W]$,
  \begin{equation}\label{eq:Vtildeg}
    V = g_X - m_X = g_X - \E[g_X\mid W] = \tilde g_X,
  \end{equation}
  so $\E[V\Delta g_k] = \E[\tilde g_X\Delta g_k]$. The estimator solves $N^{-1}\sum_k\sum_{i\in\mathcal I_k}\psi(O_i;\hat\beta_{\mathrm{naive}}, \hat\eta_k)=0$; by affinity,
  \begin{equation}\label{eq:betahat}
    \hat\beta_{\mathrm{naive}} = \frac{\hat A}{\hat J_{\mathrm{nv}}},\quad
    \hat A := \frac1N\sum_k\sum_{i\in\mathcal I_k}
    (Y_i-\hat l_{k,i})(\hat g_{k,i}-\hat m_{k,i}),\quad
    \hat J_{\mathrm{nv}} := \frac1N\sum_k\sum_{i\in\mathcal I_k}
    (X_i-\hat m_{k,i})(\hat g_{k,i}-\hat m_{k,i}).
  \end{equation}

  Step 1 (pseudo-true values by fold and after aggregation). Condition on $\mathcal D_{-k}$ and apply Lemma~\ref{lem:exact} to a new observation. This is valid because $\hat\eta_k$ is $\mathcal D_{-k}$-measurable and the identity holds for every fixed admissible $\eta$. To shorten the fold notation, set $\Delta l_k:=\hat l_{Y,k}-l_Y$, $\Delta m_k:=\hat m_{X,k}-m_X$, $\Delta q_k:=\hat q_{Y,k}-q_Y$, and $\Delta g_k:=\hat g_{X,k}-g_X$. Write $D_k := \E[\mathcal H\,\Delta g_k\mid\mathcal D_{-k}]$ and $r_{2,k} := \E[(-\Delta l_k+\beta_0\Delta m_k)(\Delta g_k-\Delta m_k) \mid\mathcal D_{-k}]$. The analogous calculation for $j$ in the proof of Lemma~\ref{lem:exact} gives
  \[
    B_k = J_0 + s_k, \qquad
    s_k := \E[\tilde g_X\Delta g_k\mid\mathcal D_{-k}]
    - \E[\Delta m_k(\Delta g_k-\Delta m_k)\mid\mathcal D_{-k}],
  \]
  and the conditional identities are
  \[
    A_k - \beta_0 B_k
    = \E[\psi(O;\beta_0,\hat\eta_k)\mid\mathcal D_{-k}]
    = D_k + r_{2,k},
    \qquad
    \beta_k^* - \beta_0 = \frac{D_k + r_{2,k}}{J_0 + s_k}.
  \]
  By Cauchy--Schwarz, $|D_k|\le\|\mathcal H\|\,\rho_{g,k}$, $|s_k|\le\|\tilde g_X\|\,\rho_{g,k}+\rho_{m,k}(\rho_{g,k}+\rho_{m,k})$, and $|r_{2,k}|\le(\rho_{l,k}+|\beta_0|\rho_{m,k})(\rho_{g,k}+\rho_{m,k})$. Assumption~\ref{ass:ratesP} gives $\max_k|s_k| = O_{\mathbb P}(\rho_N)= o_{\mathbb P}(1)$. The arithmetic--geometric inequality $\rho_m\rho_g\le\tfrac12(\rho_m^2+\rho_g^2)$ and the bound $\rho_{l,k}(\rho_{g,k}+\rho_{m,k})\le(\rho_{l,k}+\rho_{q,k}) (\rho_{m,k}+\rho_{g,k})$ give
  \begin{equation}\label{eq:r2bound}
    \max_k|r_{2,k}| \le (1+|\beta_0|)\Big[(\rho_{l,k}+\rho_{q,k})
      (\rho_{m,k}+\rho_{g,k})+\tfrac32(\rho_{m,k}^2+\rho_{g,k}^2)\Big]_{\max k}
    = o_{\mathbb P}(N^{-1/2}).
  \end{equation}
  On the event $\mathcal E_N := \{\max_k|s_k|\le J_0/2\}$, which has $\mathbb P(\mathcal E_N)\to1$, each $B_k\ge J_0/2>0$, so $\beta_k^*$ and the weights $w_k = n_kB_k/\sum_j n_jB_j$ are well defined, and
  \[
    \beta_k^* - \beta_0 - J_0^{-1}D_k
    = \frac{r_{2,k}}{J_0+s_k} - \frac{D_k\,s_k}{J_0(J_0+s_k)} .
  \]
  On $\mathcal E_N$ the first term is $\le (2/J_0)|r_{2,k}|= o_{\mathbb P}(N^{-1/2})$; the second is, using $|D_ks_k|\le \|\mathcal H\|\rho_{g,k}(\|\tilde g_X\|\rho_{g,k}+\rho_{m,k}\rho_{g,k}+ \rho_{m,k}^2)\le C(\rho_{g,k}^2+\rho_{m,k}^2)$ (leading term $\|\mathcal H\|\|\tilde g_X\|\rho_{g,k}^2$, remainder higher order), $\le(2/J_0^2)|D_ks_k|= O_{\mathbb P}(\rho_{m,k}^2+\rho_{g,k}^2) = o_{\mathbb P}(N^{-1/2})$. Hence
  \begin{equation}\label{eq:foldpt}
    \max_k\big|\beta_k^* - \beta_0 - J_0^{-1}D_k\big|
    = o_{\mathbb P}(N^{-1/2}),
    \qquad
    \max_k|\beta_k^*-\beta_0| = O_{\mathbb P}(\rho_{g,N})+o_{\mathbb P}(N^{-1/2}),
  \end{equation}
  with $\rho_{g,N}:=\max_k\rho_{g,k}=o_{\mathbb P}(N^{-1/4})$. Since $\sum_kw_k=1$, $\beta_N^*-\beta_0=\sum_kw_k(\beta_k^*-\beta_0) =J_0^{-1}\sum_kw_kD_k+o_{\mathbb P}(N^{-1/2})$. With equal folds $n_k=N/K$, $w_k=B_k/\sum_jB_j$ and $w_k-\tfrac1K = (Ks_k-\sum_j s_j)/(K\sum_jB_j)=O_{\mathbb P}(\rho_N)$ on $\mathcal E_N$; as $|D_k|\le\|\mathcal H\|\rho_{g,k}$ and $(w_k-\tfrac1K)$ depends only on $\{s_j\}$ (hence only on $\rho_{g},\rho_{m}$), $\sum_k(w_k-\tfrac1K)D_k = O_{\mathbb P} (\rho_{g,N}^2+\rho_{m,N}\rho_{g,N})=o_{\mathbb P}(N^{-1/2})$. Therefore
  \begin{equation}\label{eq:crossfitpt}
    \beta_N^*-\beta_0 = J_0^{-1}\bar D_N + o_{\mathbb P}(N^{-1/2}),
    \qquad
    |\bar D_N|\le\|\mathcal H\|\,\rho_{g,N}=o_{\mathbb P}(N^{-1/4}),
  \end{equation}
  which is the first display of part (iii). For the second display of (iii), apply Theorem~\ref{thm:driftvalid}(ii) to $\hat g_k$ (true $l_Y,m_X,q_Y$, only $g$ perturbed): $\beta_{IV}(\hat g_k)-\beta_0 = J_0^{-1}D_k+O_{\mathbb P} (\rho_{g,k}^2)=J_0^{-1}D_k+o_{\mathbb P}(N^{-1/2})$, so $\beta_{IV}(\hat g_k)$ and $\beta_k^*$ share the expansion~\eqref{eq:foldpt}; taking $w$-weighted averages, $\sum_kw_k\beta_{IV}(\hat g_k)=\beta_0+J_0^{-1}\bar D_N+o_{\mathbb P}(N^{-1/2}) =\beta_N^*+o_{\mathbb P}(N^{-1/2})$.

  Step 2 (exact affine identity for the estimator). By affinity and $\sum_k\sum_i\psi(O_i;\hat\beta_{\mathrm{naive}},\hat\eta_k)=0$,
  \begin{equation}\label{eq:step2}
    \hat\beta_{\mathrm{naive}}-\beta_N^*
    = \hat J_{\mathrm{nv}}^{-1}\,\frac1N\sum_{k}\sum_{i\in\mathcal I_k}
    \psi(O_i;\beta_N^*,\hat\eta_k),
  \end{equation}
  on $\{\hat J_{\mathrm{nv}}\neq0\}$.

  Step 3 (recentering fold by fold). Using affinity again, $\psi(O_i;\beta_N^*,\hat\eta_k)=\psi(O_i;\beta_k^*,\hat\eta_k)- (\beta_N^*-\beta_k^*)j(O_i;\hat\eta_k)$. Let $\bar j_k:=\tfrac1{n_k}\sum_{i\in\mathcal I_k}j(O_i;\hat\eta_k)$, so $\E[\bar j_k\mid\mathcal D_{-k}]=B_k$ and, since $\|j(\cdot;\hat\eta_k)\|= O_{\mathbb P}(1)$ (finite $\|j(\cdot;\eta_0)\|$ by the $L_4$ bound on $X$, plus Lemma~\ref{lem:l2cont}), $\mathrm{Var}(\bar j_k\mid\mathcal D_{-k})= O_{\mathbb P}(n_k^{-1})$, whence $\bar j_k-B_k=O_{\mathbb P}(N^{-1/2})$ by conditional Chebyshev. Then
  \[
    \frac1N\sum_k\sum_i\psi(O_i;\beta_N^*,\hat\eta_k)
    = \frac1N\sum_k\sum_i\psi(O_i;\beta_k^*,\hat\eta_k)
    - \sum_k\frac{n_k}{N}(\beta_N^*-\beta_k^*)\bar j_k,
  \]
  and the last sum equals $\sum_k\tfrac{n_k}{N}(\beta_N^*-\beta_k^*)B_k +\sum_k\tfrac{n_k}{N}(\beta_N^*-\beta_k^*)(\bar j_k-B_k)$. The first piece is $\tfrac1N\sum_kn_kB_k(\beta_N^*-\beta_k^*)=0$, because $\sum_kn_kB_k\beta_k^*=\sum_kn_kA_k=(\sum_jn_jB_j)\beta_N^*$ by the definition of $\beta_N^*$. The second piece is bounded by $\max_k|\beta_N^*-\beta_k^*|\cdot\max_k|\bar j_k-B_k|= (O_{\mathbb P}(\rho_{g,N})+o_{\mathbb P}(N^{-1/2}))\cdot O_{\mathbb P}(N^{-1/2}) =o_{\mathbb P}(N^{-1/2})$, using~\eqref{eq:foldpt} and $\rho_{g,N}=o_{\mathbb P}(1)$. Hence
  \begin{equation}\label{eq:step3}
    \frac1N\sum_k\sum_i\psi(O_i;\beta_N^*,\hat\eta_k)
    = \frac1N\sum_k\sum_i\psi(O_i;\beta_k^*,\hat\eta_k) + o_{\mathbb P}(N^{-1/2}).
  \end{equation}

  Step 4 (stochastic equicontinuity). For $i\in\mathcal I_k$ set $D_{ik}:=\psi(O_i;\beta_k^*,\hat\eta_k)-\psi(O_i;\beta_0,\eta_0)$; note $\psi(O_i;\beta_0,\eta_0)=U_iV_i$. Conditional on $\mathcal D_{-k}$ the $\{O_i\}_{i\in\mathcal I_k}$ are i.i.d.\ and $(\beta_k^*,\hat\eta_k)$ is fixed, so $\E[D_{ik}\mid\mathcal D_{-k}]=\E[\psi(O;\beta_k^*,\hat\eta_k)\mid \mathcal D_{-k}]-\E[UV]=(A_k-\beta_k^*B_k)-0=0$ by definition of $\beta_k^*$. By Lemma~\ref{lem:l2cont}, $\E[D_{ik}^2\mid\mathcal D_{-k}]=\|\psi(\cdot;\beta_k^*,\hat\eta_k)- \psi(\cdot;\beta_0,\eta_0)\|^2\le L^2(|\beta_k^*-\beta_0|+\sum_j \rho_{j,k}^{1/2})^2=:\delta_{N,k}^2=o_{\mathbb P}(1)$. Consequently $\mathrm{Var}\big(N^{-1/2}\sum_{i\in\mathcal I_k}D_{ik}\mid\mathcal D_{-k}\big) =\tfrac{n_k}{N}\mathrm{Var}(D_{ik}\mid\mathcal D_{-k})\le \tfrac1K\delta_{N,k}^2=o_{\mathbb P}(1)$, and conditional Chebyshev followed by bounded convergence gives $N^{-1/2}\sum_{i\in\mathcal I_k}D_{ik}= o_{\mathbb P}(1)$; summing over the $K$ fixed folds,
  \begin{equation}\label{eq:step4}
    \frac1N\sum_k\sum_{i\in\mathcal I_k}\psi(O_i;\beta_k^*,\hat\eta_k)
    = \frac1N\sum_{i=1}^N U_iV_i + o_{\mathbb P}(N^{-1/2}).
  \end{equation}

  Step 5 (central limit theorem and Jacobian). Chaining \eqref{eq:step2}--\eqref{eq:step4},
  \[
    \sqrt N\big(\hat\beta_{\mathrm{naive}}-\beta_N^*\big)
    = \hat J_{\mathrm{nv}}^{-1}\Big(N^{-1/2}\sum_{i=1}^N U_iV_i + o_{\mathbb P}(1)\Big).
  \]
  The summands $U_iV_i$ are i.i.d.\ with $\E[UV]=0$ and variance $\Omega_0^{\mathrm{nv}}$. Assumption~\ref{ass:regB}(i) gives finiteness because $\E[U^2V^2]\le\|U\|_{\mathbb P,4}^2\|V\|_{\mathbb P,4}^2<\infty$, and part~(iv) imposes $\Omega_0^{\mathrm{nv}}>0$. The Lindeberg--L\'evy theorem therefore gives $N^{-1/2}\sum_iU_iV_i\xrightarrow{d}\mathcal N(0,\Omega_0^{\mathrm{nv}})$. For the Jacobian, conditional on $\mathcal D_{-k}$, $\E[\bar j_k\mid\mathcal D_{-k}]=B_k=J_0+o_{\mathbb P}(1)$ and $\mathrm{Var}(\bar j_k\mid\mathcal D_{-k})=o_{\mathbb P}(1)$. Hence $\bar j_k=J_0+o_{\mathbb P}(1)$ and $\hat J_{\mathrm{nv}}=\sum_k\tfrac1K\bar j_k \xrightarrow{p}J_0$. Slutsky gives $\sqrt N(\hat\beta_{\mathrm{naive}}-\beta_N^*)\xrightarrow{d} \mathcal N(0,J_0^{-2}\Omega_0^{\mathrm{nv}})$, the claim in part~(iv).

  Step 6 (variance estimator; part~(iv)). With $\hat\psi_i:=\psi(O_i;\hat\beta_{\mathrm{naive}},\hat\eta_k)$ and $\psi_i^0:=U_iV_i$, Cauchy--Schwarz gives $\big|\tfrac1N\sum_i\hat\psi_i^2-\tfrac1N\sum_i(\psi_i^0)^2\big|\le (\tfrac1N\sum_i(\hat\psi_i-\psi_i^0)^2)^{1/2} (\tfrac1N\sum_i(\hat\psi_i+\psi_i^0)^2)^{1/2}$. For the first factor, split $\hat\psi_i-\psi_i^0=-(\hat\beta_{\mathrm{naive}}-\beta_k^*)j(O_i;\hat\eta_k) +D_{ik}$. The contribution of the first term is $\sum_k\tfrac{n_k}{N}(\hat\beta_{\mathrm{naive}}-\beta_k^*)^2\cdot \tfrac1{n_k}\sum_{i\in\mathcal I_k}j(O_i;\hat\eta_k)^2$; here $\hat\beta_{\mathrm{naive}}-\beta_k^*=(\hat\beta_{\mathrm{naive}}-\beta_N^*) +(\beta_N^*-\beta_k^*)=O_{\mathbb P}(N^{-1/2})+O_{\mathbb P}(\rho_{g,N}) =o_{\mathbb P}(1)$ and $\tfrac1{n_k}\sum_{i\in\mathcal I_k} j(O_i;\hat\eta_k)^2=O_{\mathbb P}(1)$, so the contribution is $o_{\mathbb P}(1)$. The contribution of $D_{ik}$ is $\tfrac1N\sum_k\sum_iD_{ik}^2$, whose conditional mean is $\sum_k\tfrac{n_k}{N}\E[D_{ik}^2\mid\mathcal D_{-k}]\le\tfrac1K\sum_k \delta_{N,k}^2=o_{\mathbb P}(1)$, hence $o_{\mathbb P}(1)$ by conditional Markov. Thus the first factor is $o_{\mathbb P}(1)$; the second factor is $O_{\mathbb P}(1)$ because $\tfrac1N\sum_i(\psi_i^0)^2\xrightarrow{p} \Omega_0^{\mathrm{nv}}$ (law of large numbers) and the cross term is controlled by the first factor. Therefore the plug-in variance estimator
  \begin{equation}\label{eq:naivevar}
    \hat\Omega^{\mathrm{nv}} := \frac1N\sum_i\hat\psi_i^2
    = \frac1N\sum_k\sum_{i\in\mathcal I_k}\psi_{\mathrm{naive}}(O_i;\hat\beta_{\mathrm{naive}},\hat\eta_k)^2
  \end{equation}
  satisfies $\hat\Omega^{\mathrm{nv}}\xrightarrow{p}\Omega_0^{\mathrm{nv}}$. With $\hat J_{\mathrm{nv}}\xrightarrow{p}J_0$, the studentized statistic $\sqrt N(\hat\beta_{\mathrm{naive}}-\beta_N^*)/(\hat J_{\mathrm{nv}}^{-1} \sqrt{\hat\Omega^{\mathrm{nv}}})\xrightarrow{d}\mathcal N(0,1)$, so the nominal interval covers $\beta_N^*$ with probability tending to $1-\alpha$.
\end{proof}

\begin{proof}[Proof of Corollary~\ref{cor:cover}]
  By Theorem~\ref{thm:driftvalid}(iii)--(iv) and~\eqref{eq:crossfitpt},
  \[
    \frac{\sqrt N(\hat\beta_{\mathrm{naive}}-\beta_0)}{\sigma_{\mathrm{nv}}}
    = \frac{\sqrt N(\hat\beta_{\mathrm{naive}}-\beta_N^*)}{\sigma_{\mathrm{nv}}}
    + \frac{\sqrt N\,J_0^{-1}\bar D_N}{\sigma_{\mathrm{nv}}} + o_{\mathbb P}(1)
    =: Z_N + \mu_N + o_{\mathbb P}(1),
    \qquad Z_N\xrightarrow{d}\mathcal N(0,1).
  \]
  Let $\hat\sigma_{\mathrm{nv}}^2:=\hat J_{\mathrm{nv}}^{-2}\hat\Omega^{\mathrm{nv}} \xrightarrow{p}\sigma_{\mathrm{nv}}^2$ (Step 6). Coverage of the nominal interval for $\beta_0$ equals $\mathbb P(|T_N^0|\le z_{1-\alpha/2})$ with $T_N^0:=\sqrt N(\hat\beta_{\mathrm{naive}}-\beta_0)/ \hat\sigma_{\mathrm{nv}}=(Z_N+\mu_N)(\sigma_{\mathrm{nv}}/ \hat\sigma_{\mathrm{nv}})+o_{\mathbb P}(1)$.

  (a) If $\mu_N\xrightarrow{p}0$, then $T_N^0\xrightarrow{d} \mathcal N(0,1)$ by Slutsky, and coverage $\to\mathbb P(|\mathcal N(0,1)|\le z_{1-\alpha/2})=1-\alpha$.

  (b) If $\mu_N\xrightarrow{p}\mu$ finite, then $\sigma_{\mathrm{nv}}/\hat\sigma_{\mathrm{nv}}\xrightarrow{p}1$ and $T_N^0=Z_N+\mu_N+o_{\mathbb P}(1)\xrightarrow{d}\mathcal N(\mu,1)$, so coverage $\to\mathbb P(-z_{1-\alpha/2}\le\mathcal N(\mu,1)\le z_{1-\alpha/2}) =\Phi(z_{1-\alpha/2}-\mu)-\Phi(-z_{1-\alpha/2}-\mu)=:m(\mu)$. Writing $m(\mu)=\int_{-z_{1-\alpha/2}}^{z_{1-\alpha/2}}\phi(t-\mu)\,dt$, the mass a standard normal places on a fixed-width interval centered at $-\mu$ is, by symmetry and unimodality of $\phi$, strictly maximized at $\mu=0$; hence $m(\mu)<m(0)=1-\alpha$ for every $\mu\neq0$.

  (c) If $|\mu_N|\xrightarrow{p}\infty$, then $|T_N^0|\ge (|\mu_N|-|Z_N|-o_{\mathbb P}(1))\cdot(\sigma_{\mathrm{nv}}/ \hat\sigma_{\mathrm{nv}})\xrightarrow{p}\infty$ since $Z_N= O_{\mathbb P}(1)$, so coverage $\to0$.

  Finally, $\mu_N\xrightarrow{p}0\iff\sqrt N\bar D_N \xrightarrow{p}0$ because $J_0,\sigma_{\mathrm{nv}}$ are fixed and positive. This condition makes the average realized loading $\bar D_N=\sum_k(n_k/N)\E[\mathcal H\Delta g_k\mid \mathcal D_{-k}]$ equal to $o_{\mathbb P}(N^{-1/2})$.
\end{proof}

\begin{proof}[Empirical-process supplement for Theorems~\ref{thm:driftvalid} and~\ref{thm:fixedtarget}] If every $\rho_{j,k}=o_{\mathbb P}(N^{-1/4})$, then $\rho_{m,k}^2+\rho_{g,k}^2=o_{\mathbb P}(N^{-1/2})$ and $(\rho_{l,k}+\rho_{q,k})(\rho_{m,k}+\rho_{g,k})=o_{\mathbb P}(N^{-1/2})$. Thus the single rate implies Assumption~\ref{ass:ratesP}.

  Expansion following Theorem~\ref{thm:driftvalid}. Combining parts~(iii)--(iv), \eqref{eq:crossfitpt}, and Step~5 gives $\hat\beta_{\mathrm{naive}}-\beta_0=(\hat\beta_{\mathrm{naive}}-\beta_N^*)+(\beta_N^*-\beta_0)=J_0^{-1}N^{-1}\sum_iU_iV_i+J_0^{-1}\bar D_N+o_{\mathbb P}(N^{-1/2})$. Here $\bar D_N=\sum_k(n_k/N)\E[\mathcal H\Delta g_k\mid\mathcal D_{-k}]$ is the cross-fitted counterpart of $\E[\mathcal H\Delta g_X]$ and $\psi_{\mathrm{naive}}(O_i;\beta_0,\eta_0)=U_iV_i$. This gives the expansion following the theorem, with Step~4 supplying the empirical-process argument.

  Fixed-target estimator. Apply Steps~1--6 with $\psi_{\mathrm{CRC}}$ in place of $\psi_{\mathrm{naive}}$, the negative $\beta$-slope $j_{\mathrm{CRC}}(O;\eta):=(X-m)(g-m)+(X-g)(g-m)=\psi_X(O;\eta)$, and the fold pseudo-true values $\beta_k^{\mathrm{crc},*}$. Lemma~\ref{lem:exact} gives $\E[\psi_{\mathrm{CRC}}(O;\beta_0,\hat\eta_k)\mid \mathcal D_{-k}]=R_{Y,k}-\beta_0R_{X,k}$, where $R_{Y,k}=\E[\Delta l_k\Delta m_k-\Delta g_k\Delta q_k\mid\mathcal D_{-k}]$ and $R_{X,k}=\E[(\Delta m_k)^2-(\Delta g_k)^2\mid\mathcal D_{-k}]$. Hence $\beta_k^{\mathrm{crc},*}-\beta_0=(R_{Y,k}-\beta_0R_{X,k})/(J_0+R_{X,k})$. Cauchy--Schwarz gives $|R_{Y,k}|\le\rho_{l,k}\rho_{m,k}+\rho_{g,k}\rho_{q,k}$ and $|R_{X,k}|\le\rho_{m,k}^2+\rho_{g,k}^2$. The bound $\rho_{l,k}\rho_{m,k}+\rho_{g,k}\rho_{q,k}\le(\rho_{l,k}+\rho_{q,k})(\rho_{m,k}+\rho_{g,k})$ and Assumption~\ref{ass:ratesP} make both remainders $o_{\mathbb P}(N^{-1/2})$. Therefore $\beta_k^{\mathrm{crc},*}-\beta_0=o_{\mathbb P}(N^{-1/2})$ for every fold and $\beta_N^{\mathrm{crc},*}-\beta_0=o_{\mathbb P}(N^{-1/2})$. There is no linear drift because $\psi_{\mathrm{CRC}}$ is Neyman orthogonal. Steps~2--6 then give $\sqrt N(\hat\beta_{\mathrm{CRC}}-\beta_0)\xrightarrow{d}\mathcal N(0,J_0^{-2}\Omega_0)$, where $\Omega_0=\E[\psi_{\mathrm{CRC}}(O;\beta_0,\eta_0)^2]$, and $\hat\Omega\xrightarrow{p}\Omega_0$. The empirical CRC score slope used in this argument converges to $J_0$. The estimator $\hat J=N^{-1}\sum_i(\hat g_{X,i}-\hat m_{X,i})^2$ stated in Theorem~\ref{thm:fixedtarget}(ii) has the same limit by the cross-fitted law of large numbers and $L_2$ consistency of $\hat g_X-\hat m_X$. Finally, $R_N=R_Y-\beta_0R_X$ is the population counterpart of $R_{Y,k}-\beta_0R_{X,k}$.
\end{proof}

\subsection*{Proofs for Section~\ref{sec:diagnosticsbias}: Worst-case bias over a realization ball}

\begin{proof}[Proof of Proposition~\ref{prop:worst}]
  (i) Upper bound. By Lemma~\ref{lem:exact} and Cauchy--Schwarz, for $\eta\in\mathcal T(r)$,
  \[
    \big|\E[\psi_{\mathrm{naive}}(O;\beta_0,\eta)]\big|
    \le |\E[\mathcal H\Delta g]| + |\E[(-\Delta l+\beta_0\Delta m)
      (\Delta g-\Delta m)]|
    \le \|\mathcal H\|\,r_g + (r_l+|\beta_0|r_m)(r_g+r_m),
  \]
  using $\|\Delta g\|\le r_g$, $\|{-\Delta l+\beta_0\Delta m}\|\le r_l+|\beta_0| r_m$ and $\|\Delta g-\Delta m\|\le r_g+r_m$.

  Lower bound. Take $\eta^\star=(l_Y,m_X,q_Y,g_X+r_g\mathcal H/ \|\mathcal H\|)$, so $\Delta l=\Delta m=\Delta q=0$ and $\Delta g=r_g\mathcal H/\|\mathcal H\|$. Then $\|\Delta g\|=r_g\le r_g$ and $\|g_X+\Delta g\|_\infty\le\|g_X\|_\infty+r_g\|\mathcal H\|_\infty/ \|\mathcal H\|\le\|g_X\|_\infty+\|\mathcal H\|_\infty/\|\mathcal H\|\le\bar C$ (as $r_g\le1$), so $\eta^\star\in\mathcal T(r)$. By Lemma~\ref{lem:exact}, the second term vanishes and $\E[\psi_{\mathrm{naive}}(O;\beta_0,\eta^\star)]=\E[\mathcal H\Delta g]= (r_g/\|\mathcal H\|)\E[\mathcal H^2]=r_g\|\mathcal H\|$. Hence the supremum is at least $r_g\|\mathcal H\|$.

  (ii) By Lemma~\ref{lem:exact} and Cauchy--Schwarz, $|\E[\psi_{\mathrm{CRC}}(O;\beta_0,\eta)]|\le|\E[\Delta l\Delta m]|+ |\E[\Delta g\Delta q]|+|\beta_0|(|\E[(\Delta m)^2]|+|\E[(\Delta g)^2]|) \le r_lr_m+r_gr_q+|\beta_0|(r_m^2+r_g^2)$.

  Finally, along a sequence with $\sqrt N r_g\to\infty$, the naive bound $r_g\|\mathcal H\|$ satisfies $\sqrt N r_g\|\mathcal H\|\to\infty$ whenever $\|\mathcal H\|>0$ (Definition~\ref{def:esshet}), while under Assumption~\ref{ass:ratesP} the CRC bound is $o(N^{-1/2})$.
\end{proof}

Corollary~\ref{cor:adv}, referenced from Section~\ref{sec:diagnosticsbias}, realizes the worst case with an explicit adversarial sequence.

\begin{corollary}[An Admissible Adversarial Sequence]\label{cor:adv}
  Fix $\gamma \in (1/4, 1/2)$ and define the deterministic first-stage sequence $\hat{g}_X^{(N)} := g_X + c_N\,\mathcal{H}$ with $c_N := N^{-\gamma}$, setting $(\hat{l}_Y, \hat{m}_X, \hat{q}_Y) := (l_Y, m_X, q_Y)$. This sequence satisfies Assumptions~\ref{ass:regB}--\ref{ass:ratesP} after enlarging $\bar{C}$. Cross-fitting is automatic because the sequence is data-independent. Along it, let Assumptions~\ref{ass:exog}, \ref{ass:relevance}, and~\ref{ass:regB} hold, impose the heterogeneity condition $\E[\mathcal{H}^2] > 0$ from Definition~\ref{def:esshet}, and suppose $\Omega_0^{\mathrm{nv}}>0$ and $\Omega_0>0$. Then:
  \begin{itemize}
    \item[(i)] $\beta_N^* - \beta_0 = c_N\,\E[\mathcal{H}^2]/J_0$ exactly, for
          every $N$;
    \item[(ii)] $\sqrt{N}\,\big|\hat{\beta}_{\mathrm{naive}} - \beta_0\big|
            \xrightarrow{p} \infty$ and the coverage of the naive interval for $\beta_0$ tends to $0$;
    \item[(iii)] $\sqrt{N}\,\big(\hat{\beta}_{\mathrm{CRC}} - \beta_0\big)
            \xrightarrow{d} \mathcal{N}\big(0, J_0^{-2}\Omega_0\big)$ and the CRC robust interval's coverage tends to $1-\alpha$.
  \end{itemize}
\end{corollary}

\begin{proof}[Proof of Corollary~\ref{cor:adv}]
  The sequence $\hat g_X^{(N)}=g_X+c_N\mathcal H$ has $\rho_{g,k}=\|c_N\mathcal H\|=N^{-\gamma}\|\mathcal H\|=o(N^{-1/4})$ because $\gamma>1/4$, and $\rho_{l,k}=\rho_{m,k}=\rho_{q,k}=0$, so Assumption~\ref{ass:ratesP} holds. Boundedness holds after enlarging $\bar C$ to $\|g_X\|_\infty+\|\mathcal H\|_\infty$. Since the sequence is deterministic, all conditional expectations $\E[\cdot\mid\mathcal D_{-k}]$ reduce to ordinary expectations.

  (i) With $\Delta l=\Delta m=\Delta q=0$ and $\Delta g=c_N\mathcal H$, the quantities of Step~1 of the proof of Theorem~\ref{thm:driftvalid} are, for every fold, $r_{2,k}=0$, $s_k=\E[\tilde g_X\,c_N\mathcal H]=c_N\E[\mathcal H \tilde g_X]=0$ by Lemma~\ref{lem:hgeom}(ii), and $D_k=\E[\mathcal H\,c_N \mathcal H]=c_N\E[\mathcal H^2]$. Hence $B_k=J_0$ and $\beta_k^*-\beta_0= D_k/J_0=c_N\E[\mathcal H^2]/J_0$ identically across folds, so $\beta_N^*-\beta_0=c_N\E[\mathcal H^2]/J_0$ exactly for every $N$.

  (ii) By Theorem~\ref{thm:driftvalid}(iv), $\hat\beta_{\mathrm{naive}}- \beta_N^*=O_{\mathbb P}(N^{-1/2})$, so $\sqrt N(\hat\beta_{\mathrm{naive}}- \beta_0)=O_{\mathbb P}(1)+N^{1/2-\gamma}\E[\mathcal H^2]/J_0$; since $\gamma<1/2$ and $\E[\mathcal H^2]>0$, this diverges in probability. Because $\mu_N=N^{1/2-\gamma}J_0^{-1}\E[\mathcal H^2]/\sigma_{\mathrm{nv}}\to\infty$, Corollary~\ref{cor:cover}(c) gives naive coverage $\to0$.

  (iii) For the CRC score, $R_{Y}=\E[\Delta l\Delta m-\Delta g\Delta q] =0$ and $R_{X}=\E[(\Delta m)^2-(\Delta g)^2]=-c_N^2\E[\mathcal H^2]$, so the CRC pseudo-true value satisfies $\beta_N^{\mathrm{crc},*}-\beta_0= (0-\beta_0R_X)/(J_0+R_X)=\beta_0c_N^2\E[\mathcal H^2]/(J_0-c_N^2 \E[\mathcal H^2])$. As $2\gamma>1/2$, $\sqrt N c_N^2=N^{1/2-2\gamma}\to0$, so $\sqrt N(\beta_N^{\mathrm{crc},*}-\beta_0)\to0$. Theorem~\ref{thm:fixedtarget}(ii) then gives $\sqrt N (\hat\beta_{\mathrm{CRC}}-\beta_0)\xrightarrow{d}\mathcal N(0,J_0^{-2} \Omega_0)$ and asymptotic coverage $1-\alpha$.
\end{proof}

\subsection*{Proof for Section~\ref{sec:crcscore}: The score-difference identity}

\begin{proof}[Proof of Lemma~\ref{lem:scoreid}]
  From~\eqref{eq:psinaive}, $\psi_{\mathrm{naive}}(O;\beta,\eta)=(Y-l)(g-m)- \beta(X-m)(g-m)$; from~\eqref{eq:psiY}--\eqref{eq:psicrc}, $\psi_{\mathrm{CRC}}(O;\beta,\eta)=(Y-l)(g-m)+(X-g)(q-l)-\beta[(X-m)(g-m)+ (X-g)(g-m)]$. Subtracting,
  \[
    \psi_{\mathrm{CRC}}-\psi_{\mathrm{naive}}=(X-g)(q-l)-\beta(X-g)(g-m)
    =(X-g)\big[(q-l)-\beta(g-m)\big].
  \]
  At $(\beta_0,\eta_0)$ this is $e\big[(q_Y-l_Y)-\beta_0(g_X-m_X)\big]= e\,\mathcal H$ by definition of $\mathcal H$. For the zero-mean property, the bracket $f:=(q_Y-l_Y)-\beta(g_X-m_X)$ is $(Z,W)$-measurable, so $\E[(X-g_X)f]=\E[ef]=0$ by~\eqref{eq:eorth}, for every $\beta$. Finally $\E[(e\mathcal H)^2]=\E[e^2\mathcal H^2]=\E[\E[e^2\mid Z,W]\mathcal H^2]= \E[\sigma_X^2(Z,W)\mathcal H^2]$, so $e\mathcal H=0$ a.s.\ iff $\sigma_X^2\mathcal H^2=0$ a.s.
\end{proof}

\subsection*{Proof for Section~\ref{sec:aligndiag}: The Alignment Diagnostic}

\begin{proof}[Proof of Proposition~\ref{prop:align}]
  Write $e_i:=X_i-g_{X,i}$, $\Delta g_{k}:=\hat g_{X,k}-g_X$, and $\Delta_{H,k}:=\hat{\mathcal H}_k^0-\mathcal H$, where $\hat{\mathcal H}_k^0:=(\hat q_{Y,k}-\hat l_{Y,k})-\beta_0(\hat g_{X,k}- \hat m_{X,k})$ is the plug-in of $\mathcal H$ at the true $\beta_0$; thus $\Delta_{H,k}=(\Delta q_k-\Delta l_k)-\beta_0(\Delta g_k-\Delta m_k)$ is an affine combination of nuisance errors with $\|\Delta_{H,k}\|\le\rho_{q,k}+ \rho_{l,k}+|\beta_0|(\rho_{g,k}+\rho_{m,k})$, and $\mathcal H,\Delta_{H,k}$ are bounded by a constant $C(\bar C,|\beta_0|)$ (Assumption~\ref{ass:regB}(ii)). Since $\hat{\mathcal H}_k=\hat{\mathcal H}_k^0-(\hat\beta_{\mathrm{CRC}}- \beta_0)(\hat g_{X,k}-\hat m_{X,k})$,
  \begin{equation}\label{eq:alignsplit}
    \widehat{\mathcal A}_N
    = \underbrace{\frac1N\sum_k\sum_{i\in\mathcal I_k}(X_i-\hat g_{X,k,i})
      \hat{\mathcal H}_{k,i}^0}_{=:\,\mathrm{A}}
    \;-\;(\hat\beta_{\mathrm{CRC}}-\beta_0)\,
    \underbrace{\frac1N\sum_k\sum_{i\in\mathcal I_k}(X_i-\hat g_{X,k,i})
      (\hat g_{X,k,i}-\hat m_{X,k,i})}_{=:\,\mathrm{B}} .
  \end{equation}

  Term $\mathrm{B}$ (the $\beta$-leakage). The conditional mean of a summand is $\E[(e-\Delta g_k)(V+\Delta g_k-\Delta m_k)\mid\mathcal D_{-k}]= -\E[\Delta g_k V]-\E[(\Delta g_k)^2]+\E[\Delta g_k\Delta m_k]=O_{\mathbb P} (\rho_{g,k})$ (the $e$-terms vanish by~\eqref{eq:eorth}); the empirical fluctuation about it is $O_{\mathbb P}(N^{-1/2})$ by bounded conditional variance (finite via the $L_4$ bound on $X$). Hence $\mathrm{B}= O_{\mathbb P}(\rho_{g,N})+O_{\mathbb P}(N^{-1/2})=o_{\mathbb P}(1)$, and with $\hat\beta_{\mathrm{CRC}}-\beta_0=O_{\mathbb P}(N^{-1/2})$ (Theorem~\ref{thm:an}), the second term of~\eqref{eq:alignsplit} is $O_{\mathbb P}(N^{-1/2})\cdot o_{\mathbb P}(1)=o_{\mathbb P}(N^{-1/2})$.

  Term $\mathrm{A}$. Substituting $X_i-\hat g_{X,k,i}=e_i-\Delta g_{k,i}$ and $\hat{\mathcal H}_{k,i}^0=\mathcal H_i+\Delta_{H,k,i}$ and expanding,
  \[
    \mathrm{A}=\underbrace{\frac1N\sum_i e_i\mathcal H_i}_{\mathrm{A}_1}
    +\underbrace{\frac1N\sum_k\sum_i e_i\Delta_{H,k,i}}_{\mathrm{A}_2}
    -\underbrace{\frac1N\sum_k\sum_i \Delta g_{k,i}\mathcal H_i}_{\mathrm{A}_3}
    -\underbrace{\frac1N\sum_k\sum_i \Delta g_{k,i}\Delta_{H,k,i}}_{\mathrm{A}_4}.
  \]
  We multiply each by $\sqrt N$. $\mathrm{A}_1$: $\sqrt N\,\mathrm{A}_1= N^{-1/2}\sum_i e_i\mathcal H_i$, the stated leading term ($\E[e\mathcal H]=0$ by~\eqref{eq:eorth}). $\mathrm{A}_3$: for $i\in\mathcal I_k$, $\E[\Delta g_{k,i}\mathcal H_i\mid\mathcal D_{-k}]=\E[\Delta g_k\mathcal H]=D_k$, and the within-fold fluctuation has conditional variance $\le n_k^{-1} \|\mathcal H\|_\infty^2\rho_{g,k}^2$, so $\tfrac1{n_k}\sum_{i\in\mathcal I_k} \Delta g_{k,i}\mathcal H_i=D_k+O_{\mathbb P}(\rho_{g,k}N^{-1/2})$; averaging, $\sqrt N\,\mathrm{A}_3=\sqrt N\bar D_N+O_{\mathbb P}(\rho_{g,N})=\sqrt N\bar D_N +o_{\mathbb P}(1)$. $\mathrm{A}_2$: $\Delta_{H,k}$ is $(Z,W)$-measurable, so $\E[e_i\Delta_{H,k,i}\mid\mathcal D_{-k}]=\E[\E[e\mid Z,W]\Delta_{H,k}]=0$ by~\eqref{eq:eorth}; the conditional variance of $N^{-1/2}\sum_{i\in \mathcal I_k}e_i\Delta_{H,k,i}$ is $\le\tfrac1K\E[e^2\Delta_{H,k}^2\mid \mathcal D_{-k}]\le\tfrac1K\|e\|_{\mathbb P,4}^2\|\Delta_{H,k}\|_\infty \|\Delta_{H,k}\|=o_{\mathbb P}(1)$ (interpolation bound), so $\sqrt N\, \mathrm{A}_2=o_{\mathbb P}(1)$ by conditional Chebyshev. $\mathrm{A}_4$: $\E[\Delta g_{k}\Delta_{H,k}\mid\mathcal D_{-k}]\le\rho_{g,k}\|\Delta_{H,k}\| =O_{\mathbb P}\big(\rho_{g,k}(\rho_{l,k}+\rho_{q,k})+\rho_{g,k}^2+\rho_{g,k} \rho_{m,k}\big)=o_{\mathbb P}(N^{-1/2})$ under Assumption~\ref{ass:ratesP} (the first summand is dominated by $(\rho_l+\rho_q)(\rho_m+\rho_g)$, the rest by $\tfrac32(\rho_g^2+\rho_m^2)$), and the fluctuation has conditional variance $\le\tfrac1K \|\Delta_{H,k}\|_\infty^2\rho_{g,k}^2=o_{\mathbb P}(1)$, so $\sqrt N\, \mathrm{A}_4=o_{\mathbb P}(1)$. Collecting,
  \[
    \sqrt N\,\mathrm{A}=N^{-1/2}\sum_i e_i\mathcal H_i-\sqrt N\bar D_N
    +o_{\mathbb P}(1),
  \]
  and combining with the $o_{\mathbb P}(1)$ contribution of Term $\mathrm{B}$ gives part (i)'s first display. For $\widehat\Xi_N$: it equals $\tfrac1N\sum_i\phi_{\mathrm{diag}}(O_i;\hat\beta_{\mathrm{CRC}},\hat\eta_k)^2$ with $\phi_{\mathrm{diag}}(O_i;\beta_0,\eta_0)=e_i\mathcal H_i$; by the argument of Step~6 of Theorem~\ref{thm:driftvalid} applied to $\phi_{\mathrm{diag}}$ (which Lemma~\ref{lem:l2cont} covers, splitting off the $\hat\beta_{\mathrm{CRC}}$-dependence through affinity as there), $\widehat \Xi_N\xrightarrow{p}\E[(e\mathcal H)^2]=\Xi_0$.

  (ii) Under $H_0$, part (i) gives $\sqrt N\widehat{\mathcal A}_N= N^{-1/2}\sum_i e_i\mathcal H_i+o_{\mathbb P}(1)\xrightarrow{d}\mathcal N(0, \Xi_0)$ by Lindeberg--L\'evy ($\E[e\mathcal H]=0$, $\Var(e\mathcal H)=\Xi_0\in (0,\infty)$); with $\widehat\Xi_N\xrightarrow{p}\Xi_0>0$, Slutsky yields $T_N\xrightarrow{d}\mathcal N(0,1)$, so the size is $\alpha$.

  (iii) If $\sqrt N|\bar D_N|\xrightarrow{p}\infty$, the term $-\sqrt N\bar D_N$ dominates the $O_{\mathbb P}(1)$ term $N^{-1/2}\sum_i e_i\mathcal H_i$, so $|\sqrt N\widehat{\mathcal A}_N|\xrightarrow{p}\infty$; as $\widehat\Xi_N\xrightarrow{p}\Xi_0<\infty$, $|T_N| \xrightarrow{p}\infty$.

  (iv) By the CRC influence expansion (Theorem~\ref{thm:an}) and Lemma~\ref{lem:scoreid} ($\psi_{\mathrm{CRC}}(O;\beta_0,\eta_0)=UV+e\mathcal H$), $\sqrt N(\hat\beta_{\mathrm{CRC}}-\beta_0)=J_0^{-1}N^{-1/2}\sum_i(U_iV_i+ e_i\mathcal H_i)+o_{\mathbb P}(1)$; by the naive expansion of Theorem~\ref{thm:driftvalid}, $\sqrt N (\hat\beta_{\mathrm{naive}}-\beta_0)=J_0^{-1}[N^{-1/2}\sum_i U_iV_i+\sqrt N \bar D_N]+o_{\mathbb P}(1)$. Subtracting and multiplying by $J_0$,
  \[
    J_0\sqrt N(\hat\beta_{\mathrm{CRC}}-\hat\beta_{\mathrm{naive}})
    = N^{-1/2}\sum_i e_i\mathcal H_i-\sqrt N\bar D_N+o_{\mathbb P}(1)
    = \sqrt N\widehat{\mathcal A}_N+o_{\mathbb P}(1)
  \]
  by part (i). This is the Hausman form; $T_N=J_0\sqrt N (\hat\beta_{\mathrm{CRC}}-\hat\beta_{\mathrm{naive}})/\widehat\Xi_N^{1/2}+ o_{\mathbb P}(1)$, the studentized contrast.

  (Remark~\ref{rem:degen}.) If $\mathcal H\equiv0$ then $\Xi_0=0$ and $\bar D_N=0$; part (i) gives $\sqrt N\widehat{\mathcal A}_N=o_{\mathbb P}(1)$ and part (iv) gives $\hat\beta_{\mathrm{naive}}-\hat\beta_{\mathrm{CRC}}= o_{\mathbb P}(N^{-1/2})$.
\end{proof}

\subsection*{Supplementary results for Section~\ref{sec:robustsets}: The standard error channel}

A first stage can center the naive point estimator at $\beta_0$ while leaving its reported standard error inconsistent. The next assumption and proposition state this distinction.

\begin{assumption}[First-Stage Asymptotic Linearity along $\mathcal{H}$]\label{ass:fsal}
  There exists a bounded, $(Z,W)$-measurable function $a$ with $\E[a^2] < \infty$ such that, for each fold $k$,
  \begin{equation}\label{eq:condS}
    \E\big[\mathcal{H}\,\Delta g_{X,k} \mid \mathcal{D}_{-k}\big]
    \;=\; \frac{1}{N - n_k}\sum_{j \notin \mathcal{I}_k}
    \big(X_j - g_X(Z_j, W_j)\big)\,a(Z_j, W_j)
    \;+\; \varepsilon_{k,N},
    \qquad
    \sqrt{N}\,\max_k |\varepsilon_{k,N}| \xrightarrow{p} 0 .
  \end{equation}
\end{assumption}

\begin{proposition}[The Standard Error Channel]\label{prop:sech}
  Let Assumptions~\ref{ass:exog}, \ref{ass:relevance}, \ref{ass:regB}, \ref{ass:ratesP}, and~\ref{ass:fsal} hold, write $e := X - g_X$, and suppose $\Omega_S := \E\big[(UV + e\,a)^2\big] > 0$. Then:
  \begin{itemize}
    \item[(i)] $\sqrt{N}\,\big(\hat{\beta}_{\mathrm{naive}} - \beta_0\big)
            \xrightarrow{d} \mathcal{N}\big(0,\; J_0^{-2}\,\Omega_S\big)$, with $\Omega_S = \Omega_0^{\mathrm{nv}} + 2\,\E[U V e a] + \E[e^2 a^2]$;
    \item[(ii)] the naive plug-in variance estimator satisfies
          $\hat{J}_{\mathrm{nv}}^{-2}\hat{\Omega}^{\mathrm{nv}} \xrightarrow{p} J_0^{-2}\,\Omega_0^{\mathrm{nv}}$, so the true asymptotic coverage of the nominal $(1-\alpha)$ naive interval for $\beta_0$ is $2\,\Phi\big(z_{1-\alpha/2}\sqrt{\Omega_0^{\mathrm{nv}}/\Omega_S}\big) - 1$. If $\Cov(U, X \mid Z, W) = 0$ a.s., then $\E[UVea] = 0$, hence $\Omega_S = \Omega_0^{\mathrm{nv}} + \E[e^2 a^2] > \Omega_0^{\mathrm{nv}}$ whenever $\E[e^2 a^2] > 0$, and the naive interval strictly under-covers even though its centering is correct;
    \item[(iii)] if $a = \mathcal{H}$, then $\Omega_S = \Omega_0$, the
          CRC robust variance of Theorem~\ref{thm:fixedtarget}, and $\hat{\beta}_{\mathrm{naive}} - \hat{\beta}_{\mathrm{CRC}} = o_{\mathbb{P}}(N^{-1/2})$, so the first stage orthogonalizes implicitly, yet the naive sandwich still reports $J_0^{-2}\Omega_0^{\mathrm{nv}}$.
  \end{itemize}
\end{proposition}

Parts (ii)--(iii) concern inference on the fixed target $\beta_0$ and therefore do not contradict Theorem~\ref{thm:driftvalid}(iv), which concerns inference around the moving target $\beta_N^*$. The naive sandwich consistently estimates the variance around $\beta_N^*$ but omits the first-stage term needed for inference on $\beta_0$. When $\Cov(U,X\mid Z,W)=0$, this omission produces strict undercoverage whenever $\E[e^2a^2]>0$. Proposition~\ref{prop:series} gives a concrete first stage satisfying Assumption~\ref{ass:fsal}: projection estimators can supply the point-estimate correction implicitly \citep{newey1994asymptotic}, while the naive sandwich remains inconsistent for fixed-target inference unless its probability limit equals the true variance. Outside that projection structure, even correct centering is not guaranteed.

\begin{proof}[Proof of Proposition~\ref{prop:sech}]
  (i) With equal folds $n_k=N/K$, Assumption~\ref{ass:fsal} gives $\tfrac{n_k}{N}\E[\mathcal H\Delta g_k\mid\mathcal D_{-k}]= \tfrac1{N(K-1)}\sum_{j\notin\mathcal I_k}e_ja_j+\tfrac1K\varepsilon_{k,N}$, because $\tfrac{n_k}{N}\cdot\tfrac1{N-n_k}=\tfrac1{N(K-1)}$. Summing over $k$ and using that each index $j$ is omitted from exactly its own fold, hence appears in $\sum_{j\notin\mathcal I_k}$ for $K-1$ of the $K$ folds, $\sum_k\sum_{j\notin\mathcal I_k}e_ja_j=(K-1)\sum_{j=1}^N e_ja_j$, so
  \[
    \bar D_N=\sum_k\frac{n_k}{N}\E[\mathcal H\Delta g_k\mid\mathcal D_{-k}]
    =\frac1N\sum_{j=1}^N e_j a_j+\frac1K\sum_k\varepsilon_{k,N},
  \]
  while
  \[
    \sqrt N\,\Big|\tfrac1K\sum_k\varepsilon_{k,N}\Big|\le\sqrt N\max_k
    |\varepsilon_{k,N}|=o_{\mathbb P}(1),
  \]
  whence $\sqrt N\bar D_N=N^{-1/2}\sum_j e_ja_j+o_{\mathbb P}(1)$. Substituting into the naive expansion of Theorem~\ref{thm:driftvalid},
  \[
    \sqrt N(\hat\beta_{\mathrm{naive}}-\beta_0)
    = J_0^{-1}\Big[N^{-1/2}\sum_i U_iV_i+\sqrt N\bar D_N\Big]+o_{\mathbb P}(1)
    = J_0^{-1}N^{-1/2}\sum_i(U_iV_i+e_ia_i)+o_{\mathbb P}(1).
  \]
  The summands $U_iV_i+e_ia_i$ are i.i.d.\ with mean $\E[UV]+\E[ea]=0$ (the second by~\eqref{eq:eorth}, $a$ being $(Z,W)$-measurable) and variance $\Omega_S=\E[(UV+ea)^2]=\Omega_0^{\mathrm{nv}}+2\E[UVea]+\E[e^2a^2]$, so Lindeberg--L\'evy and Slutsky give $\sqrt N(\hat\beta_{\mathrm{naive}}-\beta_0) \xrightarrow{d}\mathcal N(0,J_0^{-2}\Omega_S)$.

  (ii) Step~6 of the proof of Theorem~\ref{thm:driftvalid} shows $\hat J_{\mathrm{nv}}^{-2}\hat\Omega^{\mathrm{nv}}\xrightarrow{p}J_0^{-2} \Omega_0^{\mathrm{nv}}$ using only Assumptions~\ref{ass:regB}--\ref{ass:ratesP} (it does not use Assumption~\ref{ass:fsal}), so the reported variance is $J_0^{-2}\Omega_0^{\mathrm{nv}}$ while the true one is $J_0^{-2}\Omega_S$. The nominal interval has half-width $z_{1-\alpha/2}\sqrt{J_0^{-2} \Omega_0^{\mathrm{nv}}/N}(1+o_{\mathbb P}(1))$ and $\hat\beta_{\mathrm{naive}}-\beta_0=\sqrt{J_0^{-2}\Omega_S/N}\,Z(1+ o_{\mathbb P}(1))$ with $Z\xrightarrow{d}\mathcal N(0,1)$; hence coverage $\to\mathbb P(|Z|\le z_{1-\alpha/2}\sqrt{\Omega_0^{\mathrm{nv}}/\Omega_S})= 2\Phi(z_{1-\alpha/2}\sqrt{\Omega_0^{\mathrm{nv}}/\Omega_S})-1$. For the sign, compute $\E[UVea]=\E[Va\,\E[Ue\mid Z,W]]$ (as $V,a$ are $(Z,W)$-measurable), and $\E[Ue\mid Z,W]=\E[UX\mid Z,W]-g_X\E[U\mid Z,W]=\Cov(U,X\mid Z,W)+ \mathcal H g_X-g_X\mathcal H=\Cov(U,X\mid Z,W)$ by~\eqref{eq:UZW}. Thus $\E[UVea]=\E[Va\,\Cov(U,X\mid Z,W)]$, which vanishes when $\Cov(U,X\mid Z,W)=0$ a.s.; then $\Omega_S=\Omega_0^{\mathrm{nv}}+\E[e^2a^2]> \Omega_0^{\mathrm{nv}}$ whenever $\E[e^2a^2]>0$, and $2\Phi(z_{1-\alpha/2}\sqrt{\Omega_0^{\mathrm{nv}}/\Omega_S})-1<1-\alpha$.

  (iii) If $a=\mathcal H$, then $\Omega_S=\E[(UV+e\mathcal H)^2]= \E[\psi_{\mathrm{CRC}}(O;\beta_0,\eta_0)^2]=\Omega_0$ by Lemma~\ref{lem:scoreid}. The influence function of $\hat\beta_{\mathrm{naive}}$ in (i) is then $J_0^{-1}(UV+e\mathcal H)$, identical to that of $\hat\beta_{\mathrm{CRC}}$ (Theorem~\ref{thm:fixedtarget} with Lemma~\ref{lem:scoreid}), so $\hat\beta_{\mathrm{naive}}-\hat\beta_{\mathrm{CRC}}=o_{\mathbb P}(N^{-1/2})$. By part (ii), however, the naive sandwich converges to $J_0^{-2}\Omega_0^{\mathrm{nv}}$, which need not equal $J_0^{-2}\Omega_0$. Lemma~\ref{lem:scoreid} gives the exact condition:
  \[
    \Omega_0-\Omega_0^{\mathrm{nv}}
    =2\E[UVe\mathcal H]+\E[e^2\mathcal H^2].
  \]
  Thus the naive sandwich is inconsistent precisely when the right-hand side is nonzero. If $\Cov(U,X\mid Z,W)=0$, the cross term vanishes, and inconsistency follows whenever $\E[e^2\mathcal H^2]>0$.
\end{proof}

\begin{proposition}[Series Least Squares Verifies Assumption~\ref{ass:fsal}]\label{prop:series}
  Let $p(Z,W) \in \mathbb{R}^{d_p}$ be a fixed, bounded dictionary with $Q := \E[p\,p']$ nonsingular; suppose $g_X = p'\theta_g$ and $\mathcal{H} = p'\theta_{\mathcal{H}}$ for some $\theta_g, \theta_{\mathcal{H}} \in \mathbb{R}^{d_p}$; and let $\hat{g}_{X,k} = p'\hat{\theta}_k$, where $\hat{\theta}_k$ is the ordinary least squares coefficient of $X$ on $p$ computed on $\mathcal{D}_{-k}$. Then Assumption~\ref{ass:fsal} holds with $a = \mathcal{H}$, and $\rho_{g,k} = O_{\mathbb{P}}(N^{-1/2})$.
\end{proposition}

\begin{proof}[Proof of Proposition~\ref{prop:series}]
  Let $\hat Q_k:=\tfrac1{N-n_k}\sum_{j\notin\mathcal I_k}p_jp_j'$ and note the fold-$k$ OLS coefficient is $\hat\theta_k=\hat Q_k^{-1}\tfrac1{N-n_k} \sum_{j\notin\mathcal I_k}p_jX_j$. Since $X=p'\theta_g+e$ with $\E[pe]=\E[p\,\E[e\mid Z,W]]=0$ by~\eqref{eq:eorth},
  \[
    \hat\theta_k-\theta_g=\hat Q_k^{-1}\frac1{N-n_k}\sum_{j\notin\mathcal I_k}
    p_j e_j .
  \]
  As $\Delta g_k=p'(\hat\theta_k-\theta_g)$ and $\mathcal H=p'\theta_{\mathcal H}$, and $\hat\theta_k$ is $\mathcal D_{-k}$-measurable,
  \[
    \E[\mathcal H\Delta g_k\mid\mathcal D_{-k}]
    =\theta_{\mathcal H}'\,\E[pp']\,(\hat\theta_k-\theta_g)
    =\theta_{\mathcal H}'Q\hat Q_k^{-1}\frac1{N-n_k}\sum_{j\notin\mathcal I_k}
    p_je_j .
  \]
  Split $\theta_{\mathcal H}'Q\hat Q_k^{-1}=\theta_{\mathcal H}'+ \theta_{\mathcal H}'(Q\hat Q_k^{-1}-I)$. The first piece gives $\theta_{\mathcal H}'\tfrac1{N-n_k}\sum_{j\notin\mathcal I_k}p_je_j= \tfrac1{N-n_k}\sum_{j\notin\mathcal I_k}(p_j'\theta_{\mathcal H})e_j= \tfrac1{N-n_k}\sum_{j\notin\mathcal I_k}(X_j-g_{X,j})\mathcal H_j$, which is the leading term of Assumption~\ref{ass:fsal} with $a=\mathcal H$. The second piece is $\varepsilon_{k,N}=\theta_{\mathcal H}'(Q\hat Q_k^{-1}-I) \tfrac1{N-n_k}\sum_{j\notin\mathcal I_k}p_je_j$; using $Q\hat Q_k^{-1}-I= -Q\hat Q_k^{-1}(\hat Q_k-Q)Q^{-1}$ with $\hat Q_k-Q=O_{\mathbb P}(N^{-1/2})$ (bounded dictionary) and $\tfrac1{N-n_k}\sum_{j\notin\mathcal I_k}p_je_j= O_{\mathbb P}(N^{-1/2})$ (mean zero, finite variance since $p$ bounded and $e\in L_2$), we get $\varepsilon_{k,N}=O_{\mathbb P}(N^{-1})$, so $\sqrt N \max_k|\varepsilon_{k,N}|=o_{\mathbb P}(1)$. Thus Assumption~\ref{ass:fsal} holds with $a=\mathcal H$. Finally $\rho_{g,k}^2=\|p'(\hat\theta_k-\theta_g)\|^2= (\hat\theta_k-\theta_g)'Q(\hat\theta_k-\theta_g)\le\lambda_{\max}(Q)\| \hat\theta_k-\theta_g\|^2=O_{\mathbb P}(N^{-1})$, i.e.\ $\rho_{g,k}= O_{\mathbb P}(N^{-1/2})$.
\end{proof}

\subsection*{Proof for Section~\ref{sec:robustsets}: Identification Robust Sets}

\begin{proof}[Proof of Proposition~\ref{thm:ar}]
  Decompose $\sqrt N\bar\psi(\beta_0)=N^{-1/2}\sum_i\psi_{\mathrm{CRC}}(O_i; \beta_0,\eta_0)+N^{-1/2}\sum_k\sum_{i\in\mathcal I_k}E_{ik}$ with $E_{ik}:= \psi_{\mathrm{CRC}}(O_i;\beta_0,\hat\eta_k)-\psi_{\mathrm{CRC}}(O_i;\beta_0, \eta_0)$. By Lemma~\ref{lem:exact} (conditional form) and $\E[\psi_{\mathrm{CRC}} (O;\beta_0,\eta_0)]=0$,
  \[
    \E[E_{ik}\mid\mathcal D_{-k}]=R_{Y,k}-\beta_0R_{X,k},\qquad
    |R_{Y,k}-\beta_0R_{X,k}|\le\rho_{l,k}\rho_{m,k}+\rho_{g,k}\rho_{q,k}+
    |\beta_0|(\rho_{m,k}^2+\rho_{g,k}^2),
  \]
  which is $o_{\mathbb P}(N^{-1/2})$ under Assumption~\ref{ass:ratesP}; hence the fold-mean contribution $N^{-1/2}\sum_k n_k(R_{Y,k}-\beta_0R_{X,k})=\sqrt N \sum_k\tfrac{n_k}{N}(R_{Y,k}-\beta_0R_{X,k})=o_{\mathbb P}(1)$. The centered part $N^{-1/2}\sum_{i\in\mathcal I_k}[E_{ik}-\E(E_{ik}\mid\mathcal D_{-k})]$ has conditional variance $\le\tfrac1K\E[E_{ik}^2\mid\mathcal D_{-k}]=\tfrac1K\| \psi_{\mathrm{CRC}}(\cdot;\beta_0,\hat\eta_k)-\psi_{\mathrm{CRC}}(\cdot;\beta_0, \eta_0)\|^2\le\tfrac{L^2}K(\sum_j\rho_{j,k}^{1/2})^2=o_{\mathbb P}(1)$ by Lemma~\ref{lem:l2cont} (with $\beta'=\beta=\beta_0$), so it is $o_{\mathbb P}(1)$ by conditional Chebyshev. Therefore $\sqrt N\bar\psi(\beta_0) =N^{-1/2}\sum_i\psi_{\mathrm{CRC}}(O_i;\beta_0,\eta_0)+o_{\mathbb P}(1) \xrightarrow{d}\mathcal N(0,\Omega_0)$ by Lindeberg--L\'evy. For the denominator, the same Step-6 argument (now with $\beta_0$ fixed, so only the nuisance-continuity part is needed) gives $\hat\sigma^2(\beta_0) \xrightarrow{p}\E[\psi_{\mathrm{CRC}}(O;\beta_0,\eta_0)^2]=\Omega_0>0$. Slutsky yields $\mathrm{AR}_N(\beta_0)\xrightarrow{d}\mathcal N(0,1)$, so $\mathbb P(\beta_0\in\mathcal C_{1-\alpha})=\mathbb P(|\mathrm{AR}_N(\beta_0)| \le z_{1-\alpha/2})\to1-\alpha$. No step divides by $J_0$ or uses a lower bound on it.

  (Remark~\ref{rem:fieller}.) With $\bar\psi(\beta)=S_Y-\beta S_X$ and $\hat\sigma^2(\beta)=M_{YY}-2\beta M_{XY}+\beta^2M_{XX}$, the defining inequality $N\bar\psi(\beta)^2\le z^2\hat\sigma^2(\beta)$ becomes
  \[
    A\beta^2-2B\beta+C\le0.
  \]
  Let $D:=B^2-AC$. If $A>0$, the set is empty when $D<0$ and is the closed interval with endpoints $(B\pm\sqrt D)/A$ when $D\ge0$. If $A<0$, it is the whole line when $D\le0$ and the complement of the open interval between the two roots when $D>0$. If $A=0$ and $B\ne0$, it is a half-line with boundary $C/(2B)$; if $A=B=0$, it is the whole line when $C\le0$ and empty when $C>0$. These are the Fieller--Anderson--Rubin cases.
\end{proof}

\end{document}